\documentclass[11pt]{article}

\usepackage[T1]{fontenc}
\usepackage{lmodern, microtype}
\usepackage[left=1.3in, right=1.3in, top=1.25in, bottom=1.25in]{geometry}
\usepackage[small]{titlesec}
\usepackage{mathtools,comment,paralist}
\usepackage{epigraph}

\usepackage{amssymb,amsmath,amsthm,fancyhdr,bbm}

\usepackage{xcolor}
\usepackage{hyperref}
\hypersetup{
  colorlinks=true,
  linkcolor=blue,
  citecolor=red
}
\usepackage{natbib}
\setcitestyle{authoryear,open={(},close={)}}

\makeatletter

\newtheorem{theorem}{Theorem}
\newtheorem{lemma}{Lemma}
\newtheorem{axiom}{Axiom}[section]
\newtheorem*{axiom*}{Axiom}
\newtheorem{proposition}{Proposition}
\newtheorem{claim}{Claim}

\newtheorem*{theorem*}{Theorem}

\newtheorem*{definition*}{Definition}
\usepackage{graphicx}
\usepackage{pstricks, enumerate, pst-node, pst-text, pst-plot}

\newcommand{\R}{\mathbb{R}}

\newcommand{\Q}{\mathbb{Q}}

\newcommand{\eps}{\varepsilon}
\newcommand{\Rp}{\mathbb{R}_{+}}

\usepackage{xparse}
\DeclareDocumentCommand\Pr{ m g }{\ensuremath{
    {   \IfNoValueTF {#2}
      {\mathbb{P}\left[{#1}\right]}
      {\mathbb{P}\left[{#1}\middle\vert{#2}\right]}%
    }
}}
\DeclareDocumentCommand\E{ m g }{\ensuremath{
    {   \IfNoValueTF {#2}
      {\mathbb{E}\left[{#1}\right]}
      {\mathbb{E}\left[{#1}\middle\vert{#2}\right]}%
    }
}}

\def\dd{\mathrm{d}}
\def\ee{\mathrm{e}}

\begin{document}

\title{\bf \Large Monotone Additive Statistics\footnote{We would like to thank Kim Border, Sebastian Ebert, Giacomo Lanzani, George Mailath, Meg Meyer, Pietro Ortoleva, Doron Ravid and Weijie Zhong for helpful comments and suggestions.}}

\author{ \large Xiaosheng Mu\thanks{Princeton University. Email: xmu@princeton.edu.} \ \ \ Luciano Pomatto\thanks{Caltech. Email: 	luciano@caltech.edu.} \ \ \ %
Philipp Strack\thanks{Yale University. Email:  philipp.strack@yale.edu. Philipp Strack was supported by a Sloan fellowship.} \ \ \ %
Omer Tamuz\thanks{Caltech. Email: tamuz@caltech.edu. Omer Tamuz was supported by a grant from the Simons Foundation (\#419427), a Sloan fellowship, a BSF award (\#2018397) and a National Science Foundation CAREER award (DMS-1944153).}}

\date{\today}

\maketitle

\begin{abstract}
The expectation is an example of a descriptive statistic that is monotone with respect to stochastic dominance, and additive for sums of independent random variables. We provide a complete characterization of such statistics, and explore a number of applications to models of individual and group decision-making. These include a representation of stationary monotone time preferences, extending the work of \cite{fishburn1982time} to time lotteries. This extension offers a new perspective on  risk attitudes toward time, as well as on the aggregation of multiple discount factors. We also offer a novel class of nonexpected utility preferences over gambles which satisfy invariance to background risk as well as betweenness, but are versatile enough to capture mixed risk attitudes.

% as well as a characterization of risk-averse preferences over monetary gambles that are invariant to mean-zero background risk.

%We study {\em statistics}: mappings from distributions to real numbers. We characterize  statistics that are monotone with respect to first-order stochastic dominance, and additive for sums of independent random variables. We explore a number of applications to models of individual and group decision making, including a representation of stationary, monotone time preferences, generalizing Fishburn and Rubinstein (1982) to time lotteries. We also characterize risk-averse preferences over monetary gambles that are invariant to mean-zero background risk, and compare this property with the more familiar property of wealth invariance featured by CARA expected utility.  
\end{abstract}

%\newpage
%\setcounter{page}{1}

\section{Introduction}

How should a random quantity be summarized by a single number? In Bayesian statistics, point estimators capture an entire posterior distribution. In finance, risk measures quantify the risk of a financial position. And in economics, certainty equivalents characterize an agent's preference for uncertain outcomes.

We use the term {\em descriptive statistic}, or simply \textit{statistic}, to refer to a map that assigns a number to each bounded random variable.
We study statistics that are monotone with respect to first-order stochastic dominance, and additive for sums of independent random variables. An example of a monotone additive statistic is the expectation. The median is monotone but not additive, while the variance is additive, but not monotone.

Monotonicity is a well studied property of statistics \cite[see, e.g.,][]{bickel2012descriptiveI,bickel2012descriptiveII}, and holds, for example, for certainty equivalents of monotone preferences over lotteries. Additivity is a stronger assumption. We focus on this property because of its conceptual simplicity and because it serves as a baseline assumption in many settings. As we argue, additivity corresponds to stationarity in the context of preferences over time lotteries (see \S\ref{sec:time}). In the context of choices over monetary gambles it corresponds to invariance to background risk (\S\ref{sec:background-risk-invariance}), or to a form of separability across independent decision problems (\S\ref{sec:combined-choices}).

Beyond the expectation, an additional example of a monotone additive statistic is the map $K_a$ that, given $a \in \R$, assigns to each random variable $X$ the value 
\begin{equation}
  K_a(X) = \frac{1}{a}\log \E{\ee^{aX}}.
\end{equation}
In the fields of probability and statistics, this function is known as the \textit{(normalized) cumulant generating function} evaluated at $a$. In finance it is called the \textit{entropic risk measure}. In economics, it corresponds to the certainty equivalent of an expected utility maximizer who exhibits constant absolute risk aversion (CARA) over gambles. For bounded random variables, the essential minimum and maximum provide further examples of such statistics; as we explain later, they are the limits of $K_a$ as $a$ approaches $\pm \infty$. The expectation is equal to $K_0$, the limit of $K_a$ as $a$ approaches $0$.

Our main result  establishes that these examples, and their weighted averages, are the only monotone additive statistics. That is, we show that if a statistic $\Phi$ is monotone, additive and normalized so that it satisfies $\Phi(c)=c$ for every constant $c$, then it is of the form
\begin{align*}
    \Phi(X) = \int K_a(X)\,\dd\mu(a)
\end{align*}
for some probability measure $\mu$ over the extended real line. This result provides a simple representation of a natural family of statistics, which one may a priori have expected to be much richer. 

\medskip

Our first application is to time preferences. The starting point for our analysis is the work by \cite{fishburn1982time}, who study preferences over dated rewards\textemdash a monetary reward, together with the time at which it will be received. 
%They show that exponential discounting of time arises as the only continuous preference such that the decision maker prefers higher and earlier rewards that is \emph{stationary}. Stationarity requires that preferences are invariant if the dated rewards under consideration are shifted by the same amount of time. 
They show that exponential discounting of time arises from a set of axioms, of which the most substantial, \emph{stationarity}, postulates that preferences between two dated rewards are unaffected by the addition of a common delay. 

We extend the analysis of \cite{fishburn1982time} to {\em time lotteries}, which consist of a monetary reward $x$ and a random time $T$ at which it will be received. In this setting, we too introduce a stationarity axiom that requires preferences to be invariant with respect to random independent delays. As we argue in the main text, this stationarity axiom captures an assumption of dynamic consistency, together with the idea that preferences do not depend on calendar time.

We show that a monotone and stationary preference over time lotteries admits a representation of the form $$u(x) \ee^{-r \Phi(T)},$$ where $\Phi$ is a monotone additive statistic (Theorem~\ref{thm:time-lotteries}). Thus, $\Phi(T)$ is the certainty equivalent of the random time $T$, i.e.\ the deterministic time that is as desirable as $T$. Over deterministic dated rewards, the above representation coincides with standard  discounted utility. General time lotteries are reduced to deterministic ones through the certainty equivalent $\Phi$. By our main representation theorem, it takes the form  $\Phi(T) = \int K_a(T)\,\dd\mu(a)$. In this context, each $K_a(T)$ is the certainty equivalent of $T$ under an expected discounted preference with discount rate $-a$. The different certainty equivalents are then averaged according to the measure $\mu$.

 In this representation it is as if the decision maker had in mind not one but multiple discount factors. Thus, $\Phi$ can be interpreted as the certainty equivalent of a decision maker who is uncertain about the correct discount factor, or as the aggregated certainty equivalent of a group of agents with different discount factors.

 Our representation theorem for monotone and stationary time preferences has implications for understanding the relation between stationarity and risk attitudes toward time. How people choose among prospects that involve risk over time has been studied both theoretically and experimentally \citep{chesson2003commonalities,onay2007intertemporal,ebert2020decision, dejarnette2018time,ebert2018prudent}. A basic paradox these papers highlight is that many subjects display risk aversion over the time dimension, even though the standard theory of expected discounted utility predicts that people are risk-seeking with respect to time lotteries.
 
 This raises the question of whether there is a tractable and well-motivated class of preferences that allows risk-aversion, risk-seeking, or a combination of the two, as well as stationarity. As pointed out by \cite*{dejarnette2018time}, this is not entirely obvious. The preferences we study in this paper offer a solution to this problem, as they satisfy a strong notion of stationarity while allowing for a variety of risk attitudes.
 
 We further apply the characterization of monotone stationary preferences to the problem of aggregating heterogeneous time preferences. It is well known that when aggregating expected discounted utility preferences, a utilitarian approach that averages individual utility functions leads to present bias \cite[see][]{jackson2014present,jackson2015collective}. Based on this observation, the literature concludes that within expected discounted utility, it is impossible to aggregate individual preferences into a social preference unless the latter is dictatorial.

 We show that this difficulty is not due to stationarity, but rather to an insistence on the idea that the social preference should conform to expected discounted utility. When preferences are allowed to belong to the more general class of monotone stationary preferences, then a social preference obtained by averaging the certainty equivalents of the individuals satisfies Pareto efficiency and stationarity (Proposition~\ref{prop:pareto-possibility}). Moreover, under some assumptions, every Paretian and stationary social preference is obtained in this way (Proposition~\ref{prop:pareto-characterization}). 

\medskip

 Monotone additive statistics also find applications to models of choice over monetary gambles. It is well known that an expected utility agent whose preferences are invariant to independent background risk must have CARA preferences. This invariance property makes CARA utility functions useful modeling tools when the analyst does not observe the agents' wealth level or the additional risks they face \citep*[see, e.g.,][]{barseghyan2018estimating}. Beyond expected utility, monotone preferences that are invariant to background risk have certainty equivalents that are monotone additive statistics. Thus, by our main theorem, any such preference can be represented by a weighted average of CARA certainty equivalents, where the mixing measure $\mu$ is now over the coefficient of absolute risk aversion. In this representation, the decision maker entertains multiple utility functions, each defining a different certainty equivalent. Every lottery is evaluated by averaging over these certainty equivalents.
 
 An interesting feature of preferences represented by monotone additive statistics is that they can display behavior that is not uniformly risk-averse nor risk-seeking, such as that of an agent buying both lottery tickets and insurance \citep{friedman1948utility}, all while maintaining invariance to background risk. At the same time, a potential difficulty for this class of preferences is that their defining parameter, the measure $\mu$ over the coefficient of risk aversion, is infinite-dimensional. To narrow down the parameter space, we focus on those preferences that also satisfy \emph{betweenness}, a well-known weakening of the independence axiom that has been extensively studied in the literature \citep*[see][]{dekel1986betweenness, gul1991theory, cerreia2020explicit}. We show that a preference represented by a monotone additive statistic $\Phi$ satisfies betweenness if and only if it is of the form
 \begin{align*}
     \Phi(X) = \beta K_{-a\beta}(X) + (1-\beta)K_{a(1-\beta)}(X).
 \end{align*}
 The parameter $\beta \in [0,1]$ controls the the relative weights of the risk-averse and risk-seeking components, with increased $\beta$ making the decision maker more risk-averse. The parameter $a$ is a scale parameter. This is a simple two-parameter family, but it is rich enough to accommodate preferences that are neither risk-averse nor risk-seeking, while maintaining invariance to background risk.

\medskip

 Our final application concerns group decision-making under risk. We consider an organization that employs multiple agents, each of whom makes decisions following an individual preference relation, which can be seen as a decision rule prescribed by the firm. We show that in order for the agents' independent choices to not violate stochastic dominance when combined, it is sufficient and necessary that their preferences are represented by the same monotone additive statistic. Thus, these are the only preferences with the property that decentralized decisions cannot result in stochastically dominated outcomes for the organization.

\subsection{Related Literature}

 A large literature in statistics studies descriptive statistics of probability distributions. A representative example is the work of \cite{bickel2012descriptiveI,bickel2012descriptiveII}, who study location statistics using an axiomatic, non-parametric approach that is similar to ours. This literature has however focused on different properties, and, to the best of our knowledge, does not contain a similar characterization of additivity and monotonicity. The mathematics literature has studied additive statistics as homomorphisms from the convolution semigroup to the real numbers \citep[see][]{ruzsa1988algebraic, mattner1999cumulants, mattner2004cumulants}, but without imposing monotonicity.

 In finance and actuarial sciences, $-K_{a}(X)$ is also known as an {\em entropic risk measure}, and is used to assess the riskiness of a financial position $X$. It is a canonical example of a coherent risk measure \citep[see][]{follmer2002convex,follmer2011stochastic,follmer2011entropic}. In this literature, \cite*{goovaerts2004comonotonic} study additive statistics that are monotone with respect to all entropic risk measures, i.e.\ those with the property that $K_a(X) \geq K_a(Y)$ for all $a \in \R$ implies $\Phi(X) \geq \Phi(Y)$, and show that they must be weighted averages of entropic risk measures, as in our main representation theorem. In contrast, we do not assume this property of $\Phi$, but instead  show that it is implied by monotonicity and additivity of $\Phi$.
 
 In an earlier paper, \cite*{pomatto2018stochastic} show that on the domain of random variables that have all moments, the only monotone additive statistic is the expectation.\footnote{The same phenomenon is studied more in depth in \cite*{fritz2020there}. It is shown there that the expectation remains the unique monotone additive statistic on the domain $L^p$, for any $p \geq 1$, while there are no monotone additive statistics on $L^p$ with $p < 1$, or on the domain of all random variables.
} This is because the larger domain includes fat-tailed random variables, which rule out all other monotone additive statistics.  In contrast, in this paper we primarily study the domain of bounded random variables, which allows for much richer preferences with a variety of risk attitudes.

 Monotone additive statistics also relate to what we called {\em additive divergences} in \cite*{mu2021blackwell}. An additive divergence is a map defined over Blackwell experiments that satisfies monotonicity with respect to the Blackwell order and additivity for product experiments. While some of the techniques used in the two papers are similar, the main mathematical argument is fundamentally different. The last step of the proof of  Theorem~\ref{thm:main} is an application of the Riesz Representation Theorem and is similar to the argument used in the previous paper. However, because the Blackwell order has different properties from first-order stochastic dominance, the remainder of the proof is different, with the previous paper having no analogue of Theorem~\ref{thm:marginal}, which is the main technical step in the proof of Theorem~\ref{thm:main}. This new technique is also needed for the proof of Theorem~\ref{thm:domains}, which characterizes monotone additive statistics beyond bounded random variables. 

 \cite*{dejarnette2018time} study preferences over time lotteries that display risk aversion. One class of preferences they propose is a generalization of expected discounted utility (GEDU) that for a random prize $X$ delivered at a random time $T$ takes the form $\mathbb{E}\left[\phi(u(X)\ee^{-rT})\right]$ for some strictly increasing transformation $\phi$. The curvature of $\phi$ determines the attitude towards risk. While GEDU satisfies stationarity for deterministic $X$ and $T$, stationarity  does not in general hold once random times $T$ are considered, even with respect to adding a deterministic delay. In contrast, we impose a strong form of stationarity but do not insist on expected utility. The only intersection between our model and GEDU are preferences represented by $K_a$, corresponding to a point-mass mixing measure $\mu$. These preferences have the standard EDU representation, but perhaps with a negative discount rate, as we explain in \S\ref{sec:risk-attitudes-time}.\footnote{When the prize $X$ is held constant, a GEDU preference reduces to an expected utility preference over random times. In contrast, our prizes are always deterministic, and the stationary preferences over random times are represented by monotone additive statistics, which are not expected utility unless the mixing measure $\mu$ is a point mass (see Proposition~\ref{prop:independence} in \S\ref{appx:independence} of the online appendix).}

 Applied to choice under risk, our representation also bears resemblance to cautious expected utility theory \citep*{ortoleva2015cautious}, in which a gamble is evaluated according to the minimum certainty equivalent across a family of utility functions. The two representations are conceptually related, as both involve uncertainty over a utility function. Our axioms are however different in that we study preferences that are invariant to adding an independent gamble, while \cite*{ortoleva2015cautious} consider the effect of mixing with another gamble.

 Decision criteria that aggregate multiple certainty equivalents have appeared before in the literature. \cite{myerson2019probability} advocate the maximization of a sum of certainty equivalents as an effective rule for risk sharing. \cite{chambers2012does} formalize and characterize this rule as a social welfare functional.

\medskip

The remainder of the paper is organized as follows. In \S\ref{sec:monotone-additive} we introduce monotone additive statistics and state our main result. In \S\ref{sec:time} we apply this result to time lotteries, and in \S\ref{sec:monetary} we apply it to monetary gambles. Finally, \S\ref{sec:proof-sketch} provides an overview of the proof of our main result. The appendix and online appendix contain omitted proofs for the results in the main text.

\section{Monotone Additive Statistics}
\label{sec:monotone-additive}

 We denote by $L^\infty$ the collection of bounded real random variables, defined over a nonatomic probability space $(\Omega,\mathcal{F},\mathbb{P})$. We will identify each $c \in \R$ with the corresponding constant random variable taking value $X(\omega)=c$ at each $\omega \in \Omega$.

 We say that a map $\Phi \colon L^\infty \to \R$ is a {\em statistic} if it satisfies (i) $\Phi(X)=\Phi(Y)$ whenever $X,Y \in L^\infty$ have the same distribution, and (ii) $\Phi(c) = c$ for every $c \in \R$; that is, $\Phi$ assigns $c$ to the constant random variable $c$. %Condition (ii) may appear restrictive, but it amounts to a simple normalization when combined with the next assumptions.\footnote{Under monotonicity and additivity, any $\Phi$ that satisfies (i) and is not identically zero must have $\Phi(1) \neq 0$, and furthermore $\Phi(X)/\Phi(1)$ is a monotone additive statistic that satisfies (ii).} 
 We are interested in statistics that satisfy monotonicity with respect to first-order stochastic dominance and additivity for sums of independent random variables. Formally, $\Phi$ is
 \begin{itemize}
     \item {\em additive} if $\Phi(X+Y) = \Phi(X) + \Phi(Y)$ whenever $X$ and $Y$ are independent, and
     \item {\em monotone} if $X \geq_1 Y$ implies $\Phi(X) \geq \Phi(Y)$, where $\geq_1$ denotes first-order stochastic dominance.
 \end{itemize}
 Since, by assumption, the value $\Phi(X)$ depends only the distribution of the random variable $X$, monotonicity is equivalent to the requirement that $\Phi(X) \geq \Phi(Y)$ whenever $X \geq Y$ almost surely. This equivalence is based on the well-known fact that $X \geq_1 Y$ if and only if there are random variables $\tilde X,\tilde Y$ such that $X$ and $\tilde X$ are identically distributed, $Y$ and $\tilde Y$ are identically distributed, and $\tilde{X} \geq \tilde{Y}$ almost surely.\footnote{An alternative, equivalent definition for a statistic is to let the domain of $\Phi$ be the set of probability distributions on $\mathbb{R}$ with bounded support. In this domain, additivity would be defined with respect to convolution. We choose to have the domain consist of random variables, as this approach offers some notational advantages.}

 We denote by $\overline{\R} = \R\cup \{-\infty,\infty\}$ the extended real numbers. Given $X \in L^\infty$ and $a \in \overline{\R} \setminus\{0,\pm\infty\}$, we consider the statistic
\begin{align}
    \label{eq:K_a}
    K_{a}(X) = \frac{1}{a}\log\E{\ee^{a X}}.
\end{align}
 The value $K_a(X)$ is the certainty equivalent of $X$ for a CARA utility function with coefficient of risk aversion $-a$. In probability and statistics, $K_a(X)$ is known as the \textit{(normalized) cumulant generating function} of $X$, evaluated at $a$.
 
 If $X$ and $Y$ have the same distribution, then $K_a(X) = K_a(Y)$. Moreover, $K_a(c) = c$ for every $c \in \R$, so $K_a$ is a statistic. If $X$ and $Y$ are independent, then $\E{\ee^{a(X+Y)}} = \E{\ee^{aX}} \E{\ee^{aY}}$, and hence $K_a$ is additive. It is also monotone.
 
 We additionally define $K_0(X),K_{\infty}(X),K_{-\infty}(X)$ to be the expectation, the essential maximum, and the essential minimum of $X$, respectively.\footnote{\label{ref:ess-max}The essential maximum and minimum are the maximum and minimum of the support: $\max[X] = \sup \{a\,:\, \Pr{X\leq a} < 1\}$ and $\min[X] = \inf \{a\,:\, \Pr{X\leq a} > 0\}$.} This choice of notation makes $a \mapsto K_a(X)$ a continuous function from $\overline{\R}$ to $\R$, for any $X$. Our main result is a representation theorem for monotone additive statistics:

\begin{theorem}
\label{thm:main}
$\Phi \colon L^\infty \to \R$ is a monotone additive statistic if and only if there exists a (unique) Borel probability measure $\mu$ on $\overline{\R}$ such that for every $X \in L^\infty$
\begin{align}
\label{eq:rep}
    \Phi(X) = \int_{\overline{\R}}K_a(X)\,\dd\mu(a).
\end{align}
\end{theorem}

We refer to $\mu$ as the \textit{mixing measure} of the statistic $\Phi$. Each $K_a$ satisfies monotonicity and additivity, and it is immediate that these two properties are preserved under convex combinations. Theorem~\ref{thm:main} says that the one-parameter family $\{K_a\}$ forms the extreme points of the set of monotone additive statistics; every such statistic must be a weighted average obtained by mixing over this family. In \S\ref{sec:proof-sketch} we provide an overview of the proof of Theorem~\ref{thm:main}.

 Theorem~\ref{thm:main} can be extended to other domains of random variables. We consider the set $L_M$ of random variables $X$ for which $K_a(X)$, as defined in \eqref{eq:K_a}, is finite for all $a \in \R$. The domain $L_M$ contains those unbounded random variables whose distributions have sub-exponential tails, as in the case of the Gaussian distribution.

\begin{theorem}
\label{thm:domains}
$\Phi \colon L_M \to \R$ is a monotone additive statistic if and only if it admits a (unique) representation of the form \eqref{eq:rep} where the measure $\mu$ has compact support in $\R$.
\end{theorem}

 %The domain $L^\infty_+$  will be important in \S\ref{sec:time} for studying preferences over time lotteries, where a random variable $X$ corresponds to the stochastic future time at which a payoff is obtained by a decision maker.
 
 The extension of Theorem~\ref{thm:main} to the larger domain $L_M$ adds to the applicability of our representation, as it includes distributions with unbounded support, such as Gaussian or Poisson, for which the function $K_a$ has closed-form expressions. For example, Theorem~\ref{thm:domains} implies that when applied to a Gaussian random variable $Z$, a monotone additive statistic $\Phi$ defined on $L_M$ takes the simple mean-variance form $\Phi(Z) = \E{Z}+c\mathrm{Var}[Z]/2$, where $c \in \R$ is the mean of the measure $\mu$ characterizing $\Phi$.
 %L: I commented this out, as it reads like a side remark 
 %Theorem~\ref{thm:main} also extends to bounded random variables that are positive (see Proposition~\ref{prop:non-negative}) or restricted to take integer values (see the working paper version).
 %To prove Theorem~\ref{thm:domains} for the cases of $L = L^\infty_+$, we first prove that any monotone additive statistic defined on these smaller domains can be extended to $L^\infty$, and then invoke Theorem~\ref{thm:main}. The case of the larger domain $L_M$ is more difficult, and its proof requires some different techniques.
 
 A few additional remarks are in order. First, Theorems~\ref{thm:main} and \ref{thm:domains} answer an open question in the mathematical finance literature on risk measures posed by \cite*{goovaerts2004comonotonic}, who asked if entropic risk measures are the only additive risk measures. Second, a possible strengthening of our additivity condition requires $\Phi(X + Y) = \Phi(X) + \Phi(Y)$ to hold for all pairs of random variables, rather than just the independent ones. As is well known, the only statistic that is additive in this more restrictive sense is the expectation \citep[see, for example,][]{deFinetti1970probability}. A different strengthening is additivity with respect to uncorrelated random variables. It follows from the analysis of \cite{chambers2020spherical} that on a finite probability space the expectation is, again, the only monotone statistic that is additive for uncorrelated random variables. 
 
 %Finally, one could also consider a weakening of additivity to a \emph{sub-additivity} condition, i.e.\ statistics that satisfy $\Phi(X+Y) \leq \Phi(X) + \Phi(Y)$ for independent $X, Y$. In the supplementary appendix \S\ref{appx:sub/super}, we develop this extension and obtain a characterization of monotone sub-additive statistics which generalizes Theorem \ref{thm:main}.

%%%%%%%%
\section{Monotone Stationary Time Preferences}
\label{sec:time}

 Next, we apply monotone additive statistics to the study of time preferences.
 We consider decision problems where an agent is asked to choose between time lotteries that pay a fixed payoff at a future random time, as in the case of a driver choosing between different routes, where some routes are more likely than others to face heavy traffic, or a company choosing between projects with different random completion times. We argue that in this context additivity is connected to a notion of stationarity, according to which a choice between future rewards is not affected by the addition of an independent delay. In this section we study preferences over time lotteries that are monotone and stationary, characterize them using monotone additive statistics, discuss the risk attitudes they can model, and apply them to the problem of aggregating heterogeneous time preferences.
 
 %As shown by \cite{fishburn1982time}, in the absence of risk, stationarity together with other minimal assumptions characterizes exponentially discounted utility, the standard criterion where a reward $x$ obtained at time $t$ is evaluated as $u(x)\ee^{-rt}$, for some utility function $u$ and discount rate $r$. We show that for preferences over \textit{time lotteries}, i.e.\ fixed rewards delivered at future random times, stationarity characterizes a new and richer class of preferences. In our representation, a reward obtained at a random time $T$ is evaluated as $u(x)\ee^{-r\Phi(T)}$ for a monotone additive statistic $\Phi$. 
 
 %In contrast to expected discounted utility, which rules out risk aversion over time lotteries, these preferences can accommodate risk-averse behavior, as is commonly observed in laboratory experiments \citep*{chesson2003commonalities,onay2007intertemporal,dejarnette2018time,ebert2020decision}. Additionally, these preferences provide a new solution to the problem of aggregating the preferences of different exponentially discounting agents. 
 
 %\cite{fishburn1982time} show that stationarity implies a representation where a reward $x$ obtained at time $t$ is evaluated as $u(x)\ee^{-rt}$ for some utility $u$ and discount rate $r > 0$. We obtain an analogous representation for time lotteries, where a reward $x$ obtained at a random time $T$ is evaluated as $u(x)\ee^{-r\Phi(T)}$ for a monotone additive statistic $\Phi$.

\subsection{Domain and Axioms}

 A {\em time lottery} is a monetary reward received by a decision maker at a future, random time. Formally, it consists of a pair $(x,T)$, where $x \in \R_{++}$ is a positive payoff and $T \in L^{\infty}_+$ is the random time at which it realizes.\footnote{Per standard notation, $L^{\infty}_{+}$ denotes the set of non-negative bounded random variables.}
 Thus, time is non-negative and continuous. Our primitive is a complete and transitive binary relation $\succeq$ on the domain $\R_{++} \times L^{\infty}_+$. We denote by $\sim$ the indifference relation induced by $\succeq$. To avoid notational confusion, in the rest of this section $x$ and $y$ always denote monetary payoffs, $t$, $s$ and $d$ denote deterministic times, and $T,S$, and $D$ denote random times. 

 We say that a preference relation $\succeq$ on $\R_{++}\times L^\infty_+$ is a {\em monotone stationary time preference} (henceforth, MSTP) if it satisfies the following axioms:

\begin{axiom}[More is Better]\label{ax:monotonicity in money}
    If $x > y$ then $(x, T) \succ (y, T)$.
\end{axiom}

\begin{axiom}[Earlier is Better]\label{ax:monotonicity in time}
    If $s > t$ then $(x,t) \succ (x,s)$, and if $S \geq_1 T$ then $(x,T) \succeq (x,S)$.
\end{axiom}

\begin{axiom}[Stochastic Stationarity]\label{ax:stationarity in time}
    If $(x,T) \succeq (y,S)$ then $(x,T+D) \succeq (y,S+D)$ for any $D$ that is independent from $T$ and $S$.
\end{axiom}

%\begin{axiom}[Certainty Equivalent]\label{ax:certainty-equivalent}
%    For every $(x,T)$ and $y \geq x$ there exists $t$ such that $(x,T) \sim (y,t)$.
%\end{axiom}

\begin{axiom}[Continuity]\label{ax:continuity time lotteries}
    For any $(y, S)$, the sets $\{(x,t)\colon (x,t) \succeq (y,S)\}$ and $\{(x,t)\colon(x,t) \preceq (y, S)\}$ are closed in $\R_{++} \times \R_{+}$. 
\end{axiom}

 The first two Axioms \ref{ax:monotonicity in money} and \ref{ax:monotonicity in time} are standard conditions that directly generalize those in \cite{fishburn1982time}, who studied preferences over dated rewards $(x,t)$ with a deterministic time $t$. They require the decision maker to prefer higher payoffs, and to prefer earlier times. Axiom~\ref{ax:continuity time lotteries} is a standard continuity assumption that does not require a choice of topology over random times. The most substantive of our axioms is stochastic stationarity. In \S\ref{sec:dynamic} we discuss this axiom in depth and motivate it using the notions of time invariance and dynamic consistency \citep{halevy2015time}.\footnote{It is worthwhile to note that we implicitly assume agents to be indifferent with respect to the timing of resolution of uncertainty. We think of the choice as being made at time $0$, and we do not distinguish between situations where the realization of the random time $T$ is revealed immediately, gradually until time $T$, or only at time $T$. Modeling preferences over the timing of resolution of uncertainty would require enlarging the choice domain beyond time lotteries.}

\subsection{Representation}

Our next result characterizes monotone stationary time preferences:

\begin{theorem}\label{thm:time-lotteries}
A preference relation $\succeq$ over time lotteries is an MSTP if and only if there exist a monotone additive statistic $\Phi$, a constant $r>0$, and a continuous and  increasing function $u\colon \R_{++} \to \R_{++}$ such that $\succeq$ is represented by
\begin{align}
    \label{eq:mono-stat-pref}
 V(x, T) =u(x) \cdot \ee^{-r \Phi(T)}.
\end{align}
\end{theorem}

As in \cite{fishburn1982time}, the parameter $r$ can be normalized to be any arbitrary positive constant by applying a monotone transformation to the representation $V$. We will often set $r$ appropriately to simplify the form of the representation. In contrast, the monotone additive statistic $\Phi$ is uniquely determined by the preference.

%As in \cite{fishburn1982time}, the parameter $r$ can be chosen arbitrarily; its choice determines $u$, up to a scaling factor. In contrast, the monotone additive statistic $\Phi$ is uniquely determined by the preference.

 Over the domain of deterministic time lotteries (i.e.\ dated rewards), $V$ coincides with an exponentially discounted utility representation. For general time lotteries, $\Phi(T)$ is the certainty equivalent of $T$, i.e.\ the unique deterministic time that satisfies $(x,T) \sim (x,\Phi(T))$. The monotonicity and continuity axioms ensure that such a certainty equivalent exists, and it is an implication of stochastic stationarity that $\Phi(T)$ is independent of the reward $x$. As we further show in the proof of Theorem~\ref{thm:time-lotteries}, the monotonicity and stochastic stationarity axioms formally translate into the certainty equivalent $\Phi$ being a monotone additive statistic. 
 
 Proposition~\ref{prop:non-negative} in the appendix shows that the representation in Theorem~\ref{thm:main} extends to the domain of non-negative bounded random variables. Thus every MSTP can be represented in the following form:
\begin{align}
    \label{eq:mono-stat-pref-expanded}
    V(x,T) = u(x)\cdot\ee^{-r \int K_a(T)\,\dd\mu(a)}.
\end{align}
 We recover expected discounted utility when $\mu$ is a point mass concentrated on a point $-a < 0$, in which case $\Phi$ takes the form $$\Phi(T) = K_{-a}(T) = \frac{1}{-a}\ \log \E{\ee^{-aT}}.$$ This certainty equivalent, with the normalization $r=a$, yields the familiar representation $V(x,T) = u(x)\mathbb{E}[\ee^{-aT}]$. For a general measure $\mu$, the statistic $\Phi(T) = \int K_a(T)\,\dd\mu(a)$ aggregates different discount rates by mixing over their corresponding certainty equivalents.

 The key feature of the representation \eqref{eq:mono-stat-pref-expanded} is that the average is not over discount factors, but instead over certainty equivalents induced by the discount factors. The resulting representation is behaviorally distinct from expected discounted utility whenever $\mu$ is not a point mass. Indeed, as we formally prove in \S\ref{appx:independence} in the online appendix, this representation satisfies the independence axiom if and only if $\mu$ is a point mass.
 
 %The representation in Theorem~\ref{thm:time-lotteries} can be extended to a discrete-time setting, which we consider in the supplementary appendix \S\ref{appx:discrete time lotteries}. One difficulty that arises is that a discrete time lottery need not have a certainty equivalent that is an integer time. Because of this, an additional step is required to first relate each time lottery to a deterministic dated reward. 

\subsection{Implications for Risk Attitudes toward Time}\label{sec:risk-attitudes-time}

Theorem~\ref{thm:time-lotteries} demonstrates that there are many ways to extend discounted utility to the domain of time lotteries, while maintaining stochastic stationarity. As is well known, standard expected discounted utility preferences are \emph{risk-seeking} over time, in the sense that a decision maker prefers receiving a reward at a random time $T$ rather than at the deterministic expected time $t = \E{T}$. But other monotone additive statistics lead to stationary time preferences that are not risk-seeking. As an example, for every $a > 0$ the statistic 
\begin{align}
  \label{eq:discount-pref}
  \Phi(T) = K_a(T) =  \frac{1}{a} \log \E{\ee^{aT}}  
\end{align}
leads, with the normalization $r = a$, to the representation
\begin{align}\label{eq:negative-discount-pref}
    V(x,T) = \frac{u(x)}{\E{\ee^{a T} }},
\end{align}
which is in fact \emph{risk-averse} over time. Under this preference, the decision maker applies a \emph{negative} discount rate $-a$ within the monotone additive statistic $\Phi$, and yet is impatient. These two aspects are  compatible because in the representation $u(x)\ee^{-r\Phi(T)}$ the statistic $\Phi$ controls the risk attitude, while the decision maker still prefers receiving prizes earlier rather than later, since $\Phi$ appears with a negative coefficient.% future is still discounted using a positive discount rate as usual.

%A key distinctive property of monotone stationary time preferences is that they can exhibit risk attitudes that are not uniform across time lotteries. Consider for example two decision problems with a fixed common reward $x =\$1000$. 
%In the first problem the decision is between receiving the reward after a day for sure, and a time lottery which gives the reward with $99\%$ probability immediately and with $1\%$ probability after 100 days. In the second decision problem the decision maker decides between receiving the reward after 99 days for sure and a time lottery which gives the reward immediately with $1$\% probability and after 100 days with $99$\% probability.
%Note that in either case the safe outcome equals the expected delay of the reward in the lottery, so that any decision maker who is risk-averse/risk-seeking toward time must prefer the safe/risky outcome in both problems. Nevertheless, a person may 

Another key distinctive property of monotone stationary time preferences is their flexibility in allowing for risk attitudes that are not uniform across time lotteries. To illustrate this point, consider two decision problems with a fixed common reward $x =\$1000$, where in the first problem the choice is between
\begin{enumerate}[(I)]
    \item receiving the reward after 1 day for sure, versus
    \item receiving the reward immediately with $99\%$ probability and after 100 days with $1\%$ probability.
\end{enumerate} 
In the second decision problem the choice is between
\begin{enumerate}[(I')]
    \item receiving the reward after 99 days for sure, versus
    \item receiving the reward immediately with $1$\% probability and after 100 days with $99$\% probability.
\end{enumerate} 
In both problems, the times at which the safe options I and I' deliver the prize are equal to the expected delay of the lotteries II and II', and thus a decision maker who is globally risk-averse or risk-seeking must either choose the safe options or the risky options in both problems. Nevertheless, it does not seem unreasonable for a person to choose I over II in order to avoid the risk of a long delay, but also choose II' to I', since the time lottery offers at least a chance of avoiding an otherwise very long delay.\footnote{We are grateful to Weijie Zhong for suggesting this example to us.}

Preferences based on monotone additive statistics are not necessarily globally risk-averse or risk-seeking, and can accommodate the aforementioned behavior.
 %Nevertheless, some monotone additive statistics lead to preferences that are neither risk-averse nor risk-seeking.
For example, the statistic
$$\Phi(T) = \frac{1}{2} K_{1}(T) + \frac{1}{2} K_{-1}(T) =  \frac{1}{2}\log \E{\ee^{T}} - \frac{1}{2}\log \E{\ee^{-T}}$$
%which leads to the representation
%\[
%    V(x,T) = \frac{\E{\ee^{-T}}}{\E{\ee^{rT}}}}
%\]
leads the decision maker to choose the safe option I in the first problem and the risky option II' in the second.

Empirically, both risk-averse and risk-seeking behavior over time lotteries are observed. For example, the experiment by \cite{ebert2018prudent} finds that there are risk-seeking and risk-averse subjects: ``Overall, therefore, and in
contrast to the evidence on wealth risk preferences, there is substantial heterogeneity in preferences
toward delay risk.'' Moreover, \cite*{dejarnette2018time} find that even the same subject often exhibits both risk aversion and risk seeking depending on the choice at hand.

%In their experiment, out of 5 different choices over time lotteries, only 2.9\% of subjects are always risk-seeking and only 12.4\% are always risk-averse. Thus 84.7\% of subjects exhibit behavior that is sometimes risk-seeking and sometimes risk-averse.\footnote{See Table 1 in \citet{dejarnette2018time}. These percentages are for a treatment with maximal delay of 12 weeks across all questions. \citet{dejarnette2018time} also measured risk preferences for time lotteries with a shorter maximal delay of 5 weeks. In this case an even higher number of 86.8\% percent of subjects is neither purely risk-seeking nor purely risk-averse across all choices.}

%As we illustrated in the last section our preferences over time-lotteries can accommodate these rich patterns in risk-aversion over time-lotteries. {\bf omer: do we really illustrate this in the last section? Also, isn't this explained in the next sentence? philipp: happy to remove the sentence.}

In \S\ref{sec:risk} below we provide a detailed analysis of the risk attitudes of preferences represented by monotone additive statistics, including a characterization of those statistics that give rise to mixed risk attitudes, as in the above example.

\subsection{Stationarity, Time Invariance and Dynamic Consistency}
\label{sec:dynamic}

  In the absence of risk, it was shown by \cite{halevy2015time} that stationarity can be understood as the implication of two more basic principles: that preferences are not affected by calendar time, and that the decision maker is dynamically consistent. As we next explain, Axiom~\ref{ax:stationarity in time} is related to a particular notion of dynamic consistency for time lotteries.
  
  We consider an enlarged framework where the decision maker is endowed with a profile $(\succeq_t)$ of preferences over time lotteries, with $\succeq_t$ representing the preference the decision maker expresses at time $t$. Formally, $\succeq_t$ is a preference over $\R_{++} \times L^\infty_+$, where in the context of $\succeq_t$ the pair $(x,T)$ represents a payoff of $x$ received at time $t + T$. %\footnote{Our notation thus deviates slightly from the notation used in \cite{halevy2015time} where times are not shifted by $t$.} 
 Adapting the definitions from \cite{halevy2015time} to our setting, we define  \emph{time invariance} and \emph{dynamic consistency} below:\footnote{These definitions are slightly different from his, and in particular his (deterministic) dynamic consistency axiom is slightly stronger, requiring the implication to hold in both directions.}
 \begin{definition*}
 The collection of preferences $(\succeq_t)$ satisfies \emph{time invariance} if all the preferences $\succeq_t$ are identical.
 \end{definition*}
  Intuitively, if the agent chooses $(x,T)$ over $(y,S)$ at some time $t$ then she makes the same choice at all other times. 
\begin{definition*}
 The collection $(\succeq_t)$ satisfies  \emph{deterministic dynamic consistency} if, for every pair of time lotteries $(x,T)$ and $(y,S)$, and every $d,t \in \R_+$ it holds that 
 \begin{align*}
 (x,T) \succeq_{t+d} (y,S) \text{ implies } (x,T+d) \succeq_t (y,S+d).
 \end{align*}
 \end{definition*}
That is, the decision maker does not reverse her choice between time $t$ and time $t+d$. Time invariance together with deterministic dynamic consistency imply stationarity with respect to deterministic delays, namely $(x,T) \succeq_t (y,S)$ implies $(x,T+d) \succeq_t (y,S+d)$.
 
Our next definition proposes a generalization of dynamic consistency to a choice between $(x,T)$ and $(y,S)$ made after a random delay $D$. What we call \emph{weak stochastic dynamic consistency} requires that if, at the random time $t+D$, the decision maker always prefers $(x,T)$ over $(y,S)$, then she would not revert her choice if asked to make the decision at time $t$ for her future self. In general, the realization of the delay $D$ could affect the distributions of $S$ and $T$ faced by the decision maker. Weak stochastic dynamic consistency considers only the case where the decision maker always faces the same choice independent of the delay, which mathematically corresponds to $D$ being independent of both $S$ and $T$.  
 %Given $T,D \in L^\infty_+$, denote by $T_d$ a random variable that has the distribution of $T$, conditioned on $D=d$. Dynamic consistency relates the choice of an agent who observes $D$ to an agent who does not observe $D$. An agent who observes $D$ has to choose, at the random time $D$, between $T_D$ and $S_D$. An agent who does not observe $D$ has to choose between $T+D$ and $S+D$.
 %relates the choice between $T+D$ and $S+D$ to the choice between $T_D$ and $S_D$ made at the random time $D$. When $D$ is independent of $T$ and $S$,  the distributions of $T_d,S_d$ are independent of $d$, so that for any realization of the random time $D$ the choice is between $T$ and $S$. 
 
 %What we call \emph{weak dynamic consistency} extends the definition of deterministic dynamic consistency. It requires that if, at time $t+D$, the decision maker always prefers $(x,T)$ over $(y,S)$ then she would not revert her choice if asked to make the decision at time $t$ for her future self.
 
 \begin{definition*}
 The collection $(\succeq_t)$ satisfies  \emph{weak stochastic dynamic consistency} if, for every pair of time lotteries $(x,T)$ and $(y,S)$, every $t \in \R_+,$ and every $D \in L^\infty_+$ independent of $S,T$ it holds that 
 \begin{align*}
  (x,T) \succeq_{t+d} (y,S)\,\text{for almost every realization $d$ of $D$} \,\, \implies (x,T+D) \succeq_t (y,S+D).
 \end{align*}
 \end{definition*}
 %The condition $(x,T) \succeq_{t+D} (y,S)\,\text{a.s.}$ in this definition means that for almost every realization $d$ of $D$ it holds that $(x,T) \succeq_{t+d} (y,S)$. %Thus, this axiom states that if a decision maker who observes $D$ always chooses $T$ over $S$, then she will also choose $T$ over $S$ before observing $D$.
 As we record in the next claim, our stochastic stationarity axiom is immediately implied by time invariance and weak stochastic dynamic consistency.
 \begin{claim}
 Suppose the collection $(\succeq_t)$ satisfies time invariance, so that $\succeq_t\ =\ \succeq$ for every $t$, and also satisfies weak stochastic dynamic consistency. Then the preference $\succeq$ satisfies stochastic stationarity.%\footnote{To see that this is true observe that by time invariance all preferences $\succeq_t$ are equal to some preference $\succeq$. By weak dynamic consistency it holds that  $(x,T) \succeq (y,S) \text{ implies }(x,T+D) \succeq (y,S+D).$}
 \end{claim}
 Indeed, by time invariance $(x,T) \succeq_t (y,S)$ implies $(x,T) \succeq_{t+d} (y,S)$ for every realization $d$ of $D$. Thus by weak stochastic dynamic consistency, $(x,T) \succeq_{t} (y,S)$ implies $(x,T+D) \succeq_{t} (y,S+D)$ whenever $D$ is independent of $S, T$. Conversely, if $(x,T) \succeq_{t+d} (y,S)$ for \emph{any} realization $d$ of $D$, then $(x, T) \succeq_{t} (y,S)$ by time invariance, and $(x,T+D) \succeq_{t} (y,S+D)$ would follow from stochastic stationarity. So stochastic stationarity also implies weak stochastic dynamic consistency under the assumption of time invariance. 
 
 Weak stochastic dynamic consistency considers the case where $D$ is independent of $S$ and $T$, which means that at the delayed time $t+D$ the agent always chooses between the same two time lotteries. A stronger dynamic consistency axiom would impose the same condition, but for an arbitrary delay $D$ that need not be independent of $S$ and $T$. To make this dependency more explicit we write $S_d,T_d$ for random variables that have the conditional distributions of $S,T$ when conditioning on $D=d$.
 \begin{definition*}
 The collection $(\succeq_t)$ satisfies  \emph{strong stochastic dynamic consistency} if, for every pair of time lotteries $(x,T)$ and $(y,S)$, every $t \in \R_+,$ and every $D \in L^\infty_+$ it holds that 
  \begin{align*}
  (x,T_d) \succeq_{t+d} (y,S_d)\,\text{for almost every realization $d$ of $D$} \,\, \implies (x,T+D) \succeq_t (y,S+D).
 \end{align*}
 \end{definition*}
 
%In general, different values $d$ of the delay correspond to different decision problems, each involving a choice between $(x,T_d)$ and $(y,S_d)$ where the distributions of $T_d$ and $S_d$ can depend on $d$. Strong stochastic dynamic consistency requires that if in each such problem the decision maker always prefers the first option, then she must also prefer the first option from an ex-ante perspective. 
 
 Intuitively, strong stochastic dynamic consistency requires consistency at different times across different decision problems, while weak stochastic dynamic consistency only requires it over the same decision problem. For instance, imagine a traveler who must choose between a train and a flight, which involve travel times $S$ and $T$ respectively, and who does not know the specific day of the month $D$ when they will need to travel. Dynamic consistency compares a traveler who must buy their ticket at the start of the month to one who can make the decision on the day of travel. Weak stochastic dynamic consistency applies when the distributions of travel times $S$ and $T$ are not dependent on the day of the month. Strong stochastic dynamic consistency applies further to cases where travel times $S_d$ and $T_d$ do depend on the day $d$.
 
 %In this case, it is reasonable to expect the two to make the same choice.
 %The traveler must take into account the conditional travel times for each day, which might not be cognitively feasible and can result in a violation of strong dynamic consistency.

 %Under time invariance, strong stochastic dynamic consistency immediately implies a strong stationarity axiom:
 %\begin{axiom}[Strong Stochastic Stationarity]\label{ax:strong stationarity in time}
    %For every pair of time lotteries $(x,T)$,  $(y,S)$ and every $D \in L^{\infty}_{+}$ not necessarily independent, if $(x,T_D) \succeq (y,S_D)$ a.s.\ then $(x,T+D) \succeq (y,S+D)$.
%\end{axiom}

 The following result shows that under time invariance, strong stochastic dynamic consistency constrains the preference over time to be represented by $K_a$ for some $a \in \overline{\R}$, rather than a general monotone additive statistic $\Phi$ as in Theorem~\ref{thm:time-lotteries}.
 \begin{proposition}\label{prop:time-lotteries-strong}
  Suppose $\succeq$ is an MSTP. Then the collection $(\succeq_t)$ with $\succeq_t\ =\ \succeq$ for every $t$ satisfies strong stochastic dynamic consistency if and only if $\succeq$ can be represented by 
\begin{align*}
    V(x, T) = u(x) \cdot \ee^{-r K_a(T)}
\end{align*}
for some $a \in \overline{\R}$, $r > 0$, and $u \colon \R_{++} \to \R_{++}$.

\end{proposition}

 In words, the preference over time is either risk-neutral, expected discounted utility, the discounted maximum or minimum, or the negatively discounted preference described in \eqref{eq:negative-discount-pref}. In particular, strong stochastic dynamic consistency would rule out the kind of mixed risk attitudes described in \S\ref{sec:risk-attitudes-time}. 
 
 Proposition~\ref{prop:time-lotteries-strong} follows from the fact that strong stochastic dynamic consistency, in combination with monotonicity and continuity, implies the classic independence axiom as we discuss in \S\ref{appx:independence} of the online appendix. Weak stochastic dynamic consistency does not imply the independence axiom and thus allows for a richer set of time preferences.

\subsection{Aggregation of Preferences over Time Lotteries}
\label{sec:aggregation}

 In this section we apply monotone stationary time preferences to collective decision problems. A company making a choice among projects with different expected completion dates, a public agency choosing which research projects to fund, or a family deciding which highway to take, are all examples of social decisions where the alternatives at hand can be seen as time lotteries. In such situations, even if individuals share the same views about the desirability of the possible outcomes, there still exists a need to compromise between different degrees of patience and risk tolerance. %For instance, a politician appointed to a public office and who is seeking re-election may favor short-term results compared to a seasoned public servant.

%%%%%%%%%
 We model this type of problem by studying a group of individuals where each agent, denoted by $i$, is equipped with a preference relation $\succeq_i$ over time lotteries. These preferences may display different degrees of patience. Following the approach in social choice, we ask how individual preferences can be aggregated into a social preference relation $\succeq$ that is aligned to the individual preferences by the Pareto principle. In this context, the Pareto principle requires that if all individuals agree that one time lottery is better than another, so should the social preference:
 
 \begin{axiom}[Pareto]\label{ax:pareto}
    If $(x,T) \succeq_i (y,S)$ for every $i$, then $(x,T) \succeq (y,S)$.
\end{axiom}

 We first consider the case where each individual preference admits a standard expected discounted utility representation $u_i(x)\mathbb{E}[\ee^{-a_iT}]$, where $u_i\colon \R_{++} \to \R_{++}$ is agent $i$'s utility function and $a_i > 0$ is her discount rate.\footnote{It is important to note that the parameter $a_i$ here is uniquely pinned down by agent $i$'s preference---when restricting to a fixed reward $x$, the preference is expected utility over random times $T$, so the discounting functions $\ee^{-a_it}$ are unique up to a linear transformation.} The next result shows that if one insists that the social preference also conforms to expected discounted utility, then dictatorship is the only admissible aggregation procedure satisfying the Pareto axiom whenever the individual discount rates are distinct.\footnote{When some agents have the same discount rate, Paretian aggregation need not be dictatorial. For example, if $a_1 = a_2$, then $u(x)\mathbb{E}[\ee^{-aT}]$ with $u = \frac{u_1 + u_2}{2}$ and $a = a_1 = a_2$ satisfies the Pareto axiom.}
 
\begin{proposition}\label{prop:pareto-impossibility}
    Let $(\succeq_1,\ldots,\succeq_n,\succeq)$ be expected discounted utility preferences over time lotteries, where each $\succeq_i$ is represented by $u_i(x)\mathbb{E}[\ee^{-a_iT}]$ and $\succeq$ is represented by $u(x)\mathbb{E}[\ee^{-aT}]$. Suppose $a_1, \ldots, a_n$ are distinct positive numbers. Then the Pareto axiom is satisfied if and only if $\succeq\ =\ \succeq_i$ for some agent $i$.
\end{proposition}

 As the proof shows, this impossibility result is a consequence of Harsanyi's utilitarian theorem \citep{harsanyi1955}. Similar impossibility results have been obtained in the setting of preferences over consumption streams \citep{gollier2005aggregation,zuber2011aggregation,jackson2014present,jackson2015collective,feng2018discount,chambers2018multiple}.%\footnote{A certain richness in the choice domain is necessary for this type of impossibility result to hold. Dictatorship need not be the only solution in the smaller domain of deterministic dated rewards: if $u(x) \ee^{-r_i t} \geq u(y) \ee^{-r_i s}$ for every $i$, then $u(x) \ee^{-r t} \geq u(y) \ee^{-r s}$ holds for every social discount rate $r$ that lies between $\min_i \{r_i\}$ and $\max_i \{r_i\}$.} 

 The next result offers a solution to this impossibility result. It shows that Paretian aggregation and stochastic stationarity are compatible, and do not necessarily result in a dictatorship, if we allow preferences to belong to the larger class of MSTPs.
 \begin{proposition}\label{prop:pareto-possibility}
 Let $(\succeq_1,\ldots,\succeq_n,\succeq)$ be MSTPs, where each $\succeq_i$ is represented by $u_i(x) \ee^{-r_i\Phi_i(T)}$ and $\succeq$ is represented by $u(x) \ee^{-r\Phi(T)}$ for some monotone additive statistics $(\Phi_i)$ and $\Phi$. If there exists $\lambda_1, \dots, \lambda_n \in \R_{+}$ such that \begin{equation}\label{eq:aggregate}
    r = \sum_{i=1}^{n} \lambda_i r_i, \text{~~~~~~} r\Phi = \sum_{i=1}^n \lambda_i r_i\Phi_i, 
    \end{equation}
    and $u = \Pi_{i=1}^n u_i^{\lambda_i}$,
    then the Pareto axiom is satisfied.
 \end{proposition}

 Note that as long as the individual certainty equivalents $\Phi_1, \dots, \Phi_n$ are not all identical, then for generic values of $\lambda_1, \dots, \lambda_n$, the social certainty equivalent $\Phi$ constructed from \eqref{eq:aggregate} is distinct from each of the individual certainty equivalents. Thus the resulting social preference $\succeq$ is generically not a dictatorship.

 The key insight of Proposition~\ref{prop:pareto-possibility} is that a linear aggregation of certainty equivalents preserves both stochastic stationarity and the Pareto axiom; as we show below, this is in fact the only way to preserve these properties. In the special case where individuals have expected discounted utility preferences, the proposition implies that we can aggregate preferences without violating stochastic stationarity by allowing the social preference to be an MSTP. This approach complements alternative solutions that have been proposed in the literature to resolve the tension between Paretian aggregation and stationarity.\footnote{For example, \cite{feng2018discount} define a different notion of Pareto efficiency that takes into account the preferences of individuals across generations. They show that a standard expected discounted social preference can satisfy this weaker Pareto axiom so long as it is more patient than all the individuals. \cite{chambers2018multiple} study a number of representations that weaken stationarity and generalize expected discounted utility.} 

 The next result gives a characterization of all social preferences that admit an MSTP representation and respect the Pareto axiom, in the special case where all agents share the same utility function and it satisfies a mild richness assumption. The assumption that all agents (and the social planner) share the same utility function is common in the literature on aggregating discount factors, following  \cite{weitzman2001gamma} and \cite{chambers2018multiple}.

 \begin{proposition}\label{prop:pareto-characterization}
 Let $(\succeq_1,\ldots,\succeq_n,\succeq)$ be MSTPs, where each $\succeq_i$ is represented by $u(x) \ee^{-r_i\Phi_i(T)}$ and $\succeq$ is represented by $u(x) \ee^{-r\Phi(T)}$ for some monotone additive statistics $(\Phi_i)$ and $\Phi$. Suppose that the common utility function satisfies either $\lim_{x \to 0} u(x) = 0$ or $\lim_{x \to \infty} u(x) = \infty$.
 
 Then, the Pareto axiom is satisfied if and only if  \eqref{eq:aggregate} holds for some $\lambda_1,\ldots,\lambda_n \in \R_{+}$ that sum to $1$.
 \end{proposition}

It follows from our proof that even if agents had different utility functions, the Pareto axiom would still require the social certainty equivalent $\Phi$ to be a convex combination of the individual $(\Phi_i)_{i=1}^{n}$. However, the implications of the Pareto axiom on the social utility function seem difficult to characterize in general.\footnote{To illustrate the difficulty, consider two individual EDU preferences represented by $u_1(x) \E{\ee^{-T}}$ and $u_2(x) \E{\ee^{-T}}$, as well as a social preference represented by $u(x) \E{\ee^{-T}}$, all with the same discount rate. In this case, one can show that the Pareto axiom reduces to the inequality condition $\frac{u(x)}{u(y)} \geq \min \{\frac{u_1(x)}{u_1(y)}, \frac{u_2(x)}{u_2(y)} \}$ for every pair of rewards $x, y$. Now suppose $u_2(x) = u_1(x)^2$ for every $x$, then the previous condition simplifies to $\frac{u_1(x)}{u_1(y)} \leq \frac{u(x)}{u(y)}  
\leq \left(\frac{u_1(x)}{u_1(y)}\right)^2$ for every $x > y$. A wide range of $u$ functions satisfies this condition, including $u_1^{\alpha}$ for any power $\alpha \in [1,2]$ and all convex combinations of such powers. This multiplicity of possible social utility functions makes it challenging to generalize Proposition~\ref{prop:pareto-characterization}.} We leave this question for future work.

\section{Preferences Over Gambles}
\label{sec:monetary}

%%%%%%%%%%%%%%%%%%%%

 In the theory of risk, CARA utility functions form a restrictive but useful class of expected utility preferences. Their usefulness stems  from the analytical tractability of the exponential form, as well as from their invariance properties.
 
 CARA utility functions are invariant to changes in wealth, so that a prospect $X$ is preferred to $Y$ if and only if $X + w$ is preferred to $Y + w$ for all wealth levels $w$. They are more generally invariant to the addition of background risk: if $X$ is preferred to $Y$ then $X + W$ is preferred to $Y + W$ for every independent random variable $W$.
 
 This property makes CARA utility functions a good approximation whenever stakes are small. In addition, they are used in empirical settings in which wealth is unknown.  For example, when estimating risk preferences from insurance choices, the CARA family ``\textit{has the advantage that it implies a household's prior wealth $w$, which frequently is unobserved, is irrelevant to the household's decisions.}'' \citep*{barseghyan2018estimating}. The stronger property of invariance to background risk is also important, since households' additional background risk---arising from, say, investments in the stock market or health conditions---may be unobservable.
 
 %  Throughout we ignore any income effects associated with price changes. (Einav, Finkelstein, Cullen)
 
 The invariance properties of CARA utility functions are conceptually distinct from the assumption that preferences obey the axioms of expected utility. In this section, we apply monotone additive statistics to study the general class of preferences that are monotone with respect to stochastic dominance and are invariant to background risk.% A first observation is that preferences that have these properties are represented by monotone additive statistics. 
 
 %Preferences over gambles represented by monotone additive statistics are a rich family that can accommodate varies risk attitudes, including mixed ones, all while maintaining invariance to background risk. We focus on a subset of these preferences that, in addition to these properties, also satisfy \emph{betweenness}: Indifference between $X$ and $Y$ implies indifference between $X$ and any random variable whose distribution is a convex combination of the distributions of $X$ and $Y$. This is a well known weakening of the independence axiom that has been studied extensively in the literature \cite{dekel1986betweenness,gul1991theory}. As we show, monotone preferences that are invariant to background risk and satisfy betweenness are represented by
 %\begin{align*}
 %    \Phi(X) = \beta K_{-a\beta}(X) + (1-\beta)K_{a(1-\beta)}(X)
 %\end{align*}
 % for some $\beta \in [0,1]$ and $a \in [0, \infty)$. For internal $\beta$ and $a$ these are mixed attitude preferences, in which the risk-attitude parameter $\beta$ controls the relative weights of the risk-averse and risk-seeking components, and the scale parameter $a$ sets the units, or, equivalently, the scale at which preferences are no longer risk-neutral. This simple two parameter family accommodates the classical setting of a decision maker who buys both lottery tickets and insurance \citep{friedman1948utility}.
 
 \subsection{Background-risk Invariant Preferences}\label{sec:background-risk-invariance}

 We consider a complete and transitive preference relation $\succeq$ over $L^\infty$, interpreted here as the space of monetary gambles. We assume that for every gamble $X$ there exists a unique certainty equivalent $\Phi(X)$ such that $\Phi(X) \sim X$. If the preference $\succeq$ is monotone with respect to first-order stochastic dominance then so is $\Phi$. We say that $\succeq$ is \textit{invariant to background risk} when it has the property that $X \succeq Y$ if and only if $X+Z \succeq Y+Z$ for $Z$ independent of $X$ and $Y$.
 
 As we now explain, a preference $\succeq$ is monotone and invariant to background risk if and only if its certainty equivalent is a monotone additive statistic. Indeed, invariance implies that $X + Y \sim \Phi(X) + Y$ for any two independent random variables $X$ and $Y$.  Likewise, $Y + \Phi(X) \sim \Phi(Y) + \Phi(X)$. Combining the two indifferences yields $X + Y \sim \Phi(X)+\Phi(Y)$. So, the certainty equivalent of $X+Y$ is given by the sum $\Phi(X) + \Phi(Y)$, and thus $\Phi$ is an additive. The converse is immediate to verify.%It is likewise immediate to verify that a preference represented by a monotone additive statistic is monotone and invariant to background risk.
 
 By Theorem~\ref{thm:main}, the certainty equivalent $\Phi$ of such a preference is a weighted average $\Phi(X) = \int K_a(X) \,\dd\mu(a)$ of the certainty equivalents of multiple CARA expected utility agents, where $\mu$ is a probability measure over the coefficient of absolute risk aversion.

\subsection{Risk Aversion}\label{sec:risk}

 In this section we characterize risk-averse and risk-seeking behavior for preferences that are represented by monotone additive statistics. A preference relation $\succeq$ over gambles is risk-averse if its certainty equivalent $\Phi$ satisfies $\Phi(X) \leq \E{X}$ for every gamble $X$, and risk-seeking if the opposite inequality holds. Risk aversion translates into a property of the support of the corresponding mixing measure $\mu$:

\begin{comment}
\begin{proposition}
A monotone stationary preference with representation $u(x) \ee^{-r\Phi(T)}$ is risk-averse (respectively risk-seeking) if and only if 
\[
\Phi(T) = \int_{\overline{\R}} K_a(T)\,\dd\mu(a)
\]
for a Borel probability measure $\mu$ supported on $[0, \infty]$ (respectively $[-\infty, 0]$). 
\end{proposition}
\end{comment}

\begin{proposition}\label{prop:risk-aversion}
A monotone additive statistic satisfies $\Phi(X) \leq \E{X}$ for every $X \in L^\infty$ if and only if  $\Phi(X) = \int_{\overline{\R}} K_a(X)\,\dd\mu(a)$ for a Borel probability measure $\mu$ supported on $[-\infty,0]$. Likewise, $\Phi(X) \geq \E{X}$ for every $X$ if and only if $\mu$ is supported on $[0, \infty]$.
\end{proposition}

%\begin{proposition}
%A monotone additive statistic satisfies $\E{X} \geq \Phi(X)$ for every $X \in L^\infty$ if and only if its mixing measure $\mu$ is supported on $[-\infty,0]$.
%\end{proposition}

%L: I changed the text to make it slightly more colloquial
In words, a preference that is invariant to background risk is additionally risk-averse (resp.\ risk-seeking) if and only if it ranks gambles by aggregating the certainty equivalents of risk-averse (resp.\ risk-seeking) CARA utility functions.\footnote{A corollary of Proposition~\ref{prop:risk-aversion} is that an additive statistic $\Phi$ is risk averse if and only if it is monotone with respect to \textit{second-order} stochastic dominance. This is perhaps surprising, since the two properties are in general not equivalent for a preference over gambles.}%\footnote{To see this, note that monotonicity in higher-order stochastic dominance implies risk aversion and thus constrains the support of $\mu$. Conversely, for each $a \leq 0$, the statistic $K_a(X) = \frac{1}{a}\log \mathbb{E}[\ee^{aX}]$ satisfies higher-order monotonicity because the function $\ee^{ax}$ has derivatives of all orders that alternate signs. By linearity, $\int K_a(X)\,\dd\mu(a)$ is also higher-order monotone whenever $\mu$ is supported on $[-\infty, 0]$.}

\subsection{Mixed Risk Aversion}\label{sec:mixed-risk-aversion}

 As pointed out in the classical work of \cite{friedman1948utility}, it is not uncommon to observe behavior that is neither risk-averse nor risk-seeking, such as that of a person who buys both lottery tickets and insurance. For concreteness, in analogy with our discussion in the time domain, consider a decision maker faced with the following two choices.
 
 In the first, the choice is between (I) facing a risk of losing \$100 with probability 1\%, or (II) paying \$1 and being fully insured against that risk. In the second decision problem the choice is between (I') paying \$1 dollar for a lottery ticket that yields \$100 with probability 1\%, or (II') not participating in the lottery. 

 Preferences represented by monotone additive statistics can model a decision maker who chooses (II) over (I) but (I') over (II'), while at the same time remaining invariant to background risk. This is the case, for example, for a preference whose certainty equivalent $\Phi(X)$ takes the form $\Phi(X) = \frac{1}{2}K_{-a}(X) + \frac{1}{2}K_{a}(X)$, with a mixing measure that puts equal weights on two coefficients of risk aversion $a$ and $-a$. 

 %Under expected utility, a utility function that exhibits concavity and convexity across different regions of its domain can rationalize the choices of buying insurance in the first problem and buying the lottery ticket in the second one. However, no such preference can predict such behavior at all wealth levels, let alone be invariant to background risk. On the other hand,

\subsection{Comparative Risk Attitudes}

%We discuss this further in \S\ref{appx:higher-order-MAS} in the appendix.

We now proceed to compare the risk attitudes expressed by different monotone additive statistics. For two preference relations $\succeq_1$ and $\succeq_2$ over gambles, with corresponding certainty equivalents $\Phi_1$ and $\Phi_2$, the preference $\succeq_1$ is \textit{more risk-averse} than $\succeq_2$ if $\Phi_1(X) \leq \Phi_2(X)$ for every gamble $X \in L^\infty$. That is, if the first decision maker assigns to every gamble a lower certainty equivalent. The next proposition characterizes comparative risk aversion for preferences represented by monotone additive statistics:

%Consider two preferences over time lotteries, represented by $u(x)\ee^{-r\Phi_{\mu}(T)}$ and $u(x)\ee^{-r\Phi_{\nu}(T)}$, where $\Phi_\mu$ and $\Phi_\nu$ are two different monotone additive statistics with corresponding measures $\mu$ and $\nu$. We say that the preference represented by $\Phi_{\mu}$ is \emph{more risk-averse than} the preference represented by $\Phi_{\nu}$ if $\Phi_{\mu}(T) \geq \Phi_{\nu}(T)$ for every $T \in L^{\infty}_+$. In words, we require the former preference to assign a worse certainty equivalent (i.e., later time) to every random time $T$. 

%Under what conditions on $\mu$ and $\nu$ is the first preference more risk-averse than the second? That is, when is it the case that $\Phi_\mu(T) \geq \Phi_\nu(T)$ for all $T$? Since $K_a(T)$ increases in $a$, first-order stochastic dominance $\mu \geq_1 \nu$ is clearly sufficient, but---as we show---it is not necessary.\footnote{An example is $\mu = \frac{1}{4} \delta_1 + \frac{3}{4} \delta_3$, whereas $\nu = \delta_2$. Condition (ii) in Proposition~\ref{prop:comparative-risk-aversion} is trivially satisfied, whereas condition (i) reduces to $\frac{1}{4} (1-b)^{+} + \frac{1}{4} (3-b)^{+} \geq \frac{1}{2} (2-b)^{+}$, which holds because the function $(a-b)^{+} = \max\{a-b, 0\}$ is convex in $a$. } We provide an exact characterization in the following result.

\begin{proposition}\label{prop:comparative-risk-aversion}
Let $\succeq_1$ and $\succeq_2$ be represented by monotone additive statistics with mixing measures $\mu_1$ and $\mu_2$, respectively. Then $\succeq_1$ is more risk-averse than $\succeq_2$ if and only if
\begin{enumerate}
    \item[(i)] For every $b > 0$, $\int_{[b,\infty]}\frac{a-b}{a}\,\dd\mu_1(a) \leq  \int_{[b,\infty]}\frac{a-b}{a}\,\dd\mu_2(a)$.
    \item[(ii)] For every $b < 0$, $\int_{[-\infty,b]}\frac{a-b}{a}\,\dd\mu_1(a) \geq \int_{[-\infty,b]}\frac{a-b}{a}\,\dd\mu_2(a)$.
\end{enumerate}
\end{proposition}
%As mentioned before our certainty equivalents $\Phi_1,\Phi_2$ correspond to the expected certainty equivalents of agents with constant absolute risk aversion drawn from $\mu_1,\mu_2$.
%A natural question is when such an agent whose risk aversion is uncertainty prop:comparative-risk-aversion. 
%Clearly a sufficient condition is that $\mu_1$ dominates $\mu_2$ in first-order stochastic dominance such that the agent, which implies that there exists a probability space on which the agent is more risk averse with probability one.
%Maybe surprisingly, this condition is not necessary.
%While first order stochastic dominance is represented by the test function $\mathbf{1}_{a\geq b}$ we can represent monotonic

 Since $K_a(X)$ increases in the parameter $a$, a sufficient condition for $\succeq_1$ being more risk-averse than $\succeq_2$ is that $\mu_2$ first-order stochastically dominates $\mu_1$. First-order stochastic dominance is, however, only a sufficient condition. The reason is that the cone generated by the functions of the form $K_{(\cdot)}(X)$, as we vary $X$, does not contain all increasing functions, and hence defines a strictly finer stochastic order over the mixing measures.\footnote{For a concrete example that the order characterized by Proposition~\ref{prop:comparative-risk-aversion} is strictly finer than first-order stochastic dominance, consider $\mu_1$ to be a point mass at $a = 2$ and $\mu_2$ to have $1/4$ mass at $a = 1$ and $3/4$ mass at $a = 3$. Clearly, neither one first-order dominates the other. Condition (ii) in Proposition~\ref{prop:comparative-risk-aversion} is trivially satisfied, whereas condition (i) reduces to $\frac{1}{2} (2-b)^{+} \leq \frac{1}{4} (1-b)^{+} + \frac{1}{4} (3-b)^{+}$, which holds because the function $(a-b)^{+} = \max\{a-b, 0\}$ is convex in $a$.}

 Proposition~\ref{prop:comparative-risk-aversion} characterizes this stochastic order by showing that the convex cone generated by the set of normalized cumulant generating functions is equal to the convex cone generated by a simple one-parameter family of test functions, of the form $g(a) = \frac{a-b}{a}1_{a \geq b}$ or $g(a) = -\frac{a-b}{a}1_{a \leq b}$.
%One way of thinking about Theorem~\ref{prop:comparative-risk-aversion} is that it offers an answer to the question of characterizing the set of all cumulant generating functions. In particular, it shows that the convex cone generated by this set is equal to the convex cone generated by a simple one-parameter family of test functions. Cumulant generating functions, which are also a recasting of the moment generating functions or  Laplace transforms, play a crucial role in probability, engineering and statistics, and so this result might be of independent interest.

 \subsection{Betweenness}\label{sec:betweenness}
 
 A disadvantage of the class of preferences represented by monotone additive statistics is that it is large, with the entire measure $\mu$ as an infinite-dimensional parameter of the preference. In this section we identify a small subset of such preferences that is indexed by only two parameters, and yet retains enough flexibility to accommodate interesting risk attitudes such as mixed risk aversion.
 
 To this end we study preferences that satisfy the \emph{betweenness} axiom. This well-known property, first studied by \cite{dekel1986betweenness} and \cite{chew1989axiomatic}, requires that the decision maker's preference over probability distributions displays indifference curves that are straight lines. In comparison, the standard independence axiom (which we study in \S\ref{appx:independence} of the online appendix) would additionally require the indifference curves to be parallel to each other. 

 Given two random variables $X$ and $Y$, we denote by $X_\lambda Y$ any random variable whose distribution is a convex combination that assigns weight $\lambda$ to the distribution of $X$ and weight $1-\lambda$ to the distribution of $Y$.
 
\begin{axiom}[Betweenness]\label{ax:betweenness}
 For all $X, Y$ and all $\lambda \in (0,1)$, $X \sim Y$ if and only if $X_\lambda Y \sim Y$.
\end{axiom}

 The betweenness axiom characterizes the following class of preferences:

\begin{theorem}\label{thm:betweenness}
Suppose a preference $\succeq$ on $L^\infty$ is represented by a monotone additive statistic $\Phi(X) = \int_{\overline{\R}} K_a(X)\,\dd\mu(a)$. Then $\succeq$ satisfies the betweenness axiom if and only if
\begin{align*}
    \Phi(X) = \beta K_{-a\beta}(X) + (1-\beta)K_{a(1-\beta)}(X)
\end{align*}
for some $\beta \in [0,1]$ and $a \in [0, \infty)$.
\end{theorem}

 This family of preferences is much smaller, as it is parameterized by only two numbers. It retains the properties of monotonicity, invariance to background risk, as well as the tractability of the CARA representation. Yet it is versatile enough to describe the kind of mixed risk attitude that leads to buying both insurance and lottery tickets.

 The risk-attitude parameter $\beta$ weights the levels of risk aversion/seeking, with $\beta=1$ corresponding to pure CARA risk aversion and $\beta=0$ corresponding to pure CARA risk seeking. For internal $\beta$, the preference exhibits mixed risk aversion as guaranteed by the previous Proposition~\ref{prop:risk-aversion}. Moreover, a simple calculation shows that for any $\beta \in (0,1)$, such a preference would buy both insurance and lottery tickets of the kind described in \S\ref{sec:mixed-risk-aversion} whenever those gambles entail a small probability of a large loss or gain.\footnote{It can be shown that if $\beta \neq 0.5$, then lottery tickets and insurance as described in \S\ref{sec:mixed-risk-aversion} are preferred if and only if the probability of gain/loss ($0.01$ in the example) is smaller than $\min(\beta, 1-\beta)$, and the corresponding gain/loss amount ($100$ in the example) is sufficiently large. If $\beta = 0.5$, then the same holds for any probability of gain/loss $< 0.5$, and for any gain/loss amount.} 

 The parameter $a$ is a scale parameter. It can be understood as the scale at which the preference deviates from risk neutrality. For gambles whose sizes are much smaller than $1/a$, the preference is very close to being risk-neutral. While for gambles that vary by much more than $1/a$, behavior will be far from risk-neutral. %Changing $a$ amounts to changing units, e.g., increasing $a$ by a factor of 100 is equivalent to measuring money in terms of cents rather than dollars. 
 % L: I commented out this last sentence because it could be confusing. It makes me think of money illusion

\subsection{Combined Choices over Gambles}\label{sec:combined-choices}

In large organizations, risky prospects are not always chosen through a deliberate, centralized process. Rather, they are combinations of independent choices, often carried out with limited coordination among the different actors. 

Consider, for example, a bank that employs two workers. The first is a trader who must choose between two contracts, the Lean Hog futures $X$ and $X'$. The second is an administrator who must choose between two insurance policies $Y$ or $Y'$ for the bank's building. Assuming the first worker chooses $X$ and the second $Y$, the resulting revenue for the bank is given by the random variable $X + Y$. When the agents face choice problems that belong to independent domains, so that $X$ and $X'$ are stochastically independent from $Y$ and $Y'$, it is natural to ask to what extent coordination is necessary for the organization.

In this section we make this question precise by asking under what conditions the agents' combined choices respect first-order stochastic dominance. Our result shows this is true if and only if individual preferences are identical and represented by a monotone additive statistic. Thus, this is the only class of preferences with the property that choices over independent domains can be decentralized without obvious harm to the organization.

We study the following model. We are given two preference relations $\succeq_1$ and $\succeq_2$ over $L^\infty$, the set of bounded gambles, that are complete and transitive (our result immediately generalizes to three or more agents). As in the example above, we think of each preference relation as describing the choices of a different agent, so that $X \succeq_i X'$ if agent $i$ chooses $X$ over $X'$. These preferences can be interpreted as being endogenous or as the result of exogenous incentives; for example, the bank trader's preferences could be driven by her contract with the employer. 

Our main axiom requires that whenever the two agents face independent decision problems, their choices, when combined, do not violate stochastic dominance:

\begin{axiom}[Consistency of Combined Choices]\label{ax:R-W}
Suppose $X,X'$ are independent of $Y,Y'$. If $X \succ_1 X'$ and $Y \succ_2 Y'$, then $X' + Y'$ does not strictly dominate $X + Y$ in first-order stochastic dominance.
\end{axiom}
If we interpret $\succeq_1$ and $\succeq_2$ as decision-making rules that are determined by the organization, then Axiom~\ref{ax:R-W} requires such rules to never result in an outcome that is stochastically dominated. That collective choices should not violate stochastic dominance is clearly a desirable requirement for a rational organization. A similar axiom was first introduced by \cite{rabin2009narrow} in the context of a model of narrow framing.

In addition to this axiom, we assume individual preference relations $\succeq_i$ satisfy basic continuity and monotonicity assumptions:

\begin{axiom}[Continuity]
  \label{ax:cont-gambles}
    If $X \succ_i Y$ then there exists $\eps > 0$ such that $X \succ_i Y + \eps$ and $X - \eps \succ_i Y$. 
\end{axiom}

\begin{axiom}[Responsiveness]
\label{ax:responsiveness}
$X+\eps \succ_i X$ for every $\eps > 0$.
\end{axiom}

We next show that under these axioms, the two preference relations must be represented by monotone additive statistics. Moreover, the statistic must be the same for both agents.
\begin{theorem}\label{thm:R-W gambles}
Two preference $\succeq_1,\succeq_2$ on $L^\infty$ satisfy Axioms~\ref{ax:R-W}, \ref{ax:cont-gambles}, and \ref{ax:responsiveness} if and only if there exists a monotone additive statistic that represents both $\succeq_1$ and $\succeq_2$.
\end{theorem}

Thus, when individual choices are not coordinated, their combination will, in general, lead to violations of stochastic dominance, even when agents' choices concern independent decision problems. The theorem singles out preferences represented by monotone additive statistics as the only class of preferences that are robust to this lack of coordination.

Theorem~\ref{thm:R-W gambles} admits an alternative interpretation, closely related to the work of \cite{rabin2009narrow} on narrow framing. In their paper, a decision maker faces multiple decisions and engages in ``narrow bracketing'' by choosing separately, in each problem, according to a fixed preference relation $\succeq$ over gambles. This is a special case of our model where $\succeq\ =\ \succeq_1\ =\ \succeq_2$. They show that the decision maker's combined choices result in dominated outcomes whenever $\succeq$ is not invariant to changes in wealth (i.e.\ for some $X, Y$ and $c \in \R$, $X \succ Y$ and $Y + c \succ X + c$), but leave open the question of characterizing the class of preferences, beyond expected utility, that satisfy Axiom~\ref{ax:R-W}. Theorem~\ref{thm:R-W gambles} provides a complete characterization of those preferences over gambles for which narrow framing does not lead to dominated choices.

\section{Overview of the Proof of Theorem~\ref{thm:main}}\label{sec:proof-sketch}

Our approach to the proof of Theorem~\ref{thm:main} is via a stochastic order known as the {\em catalytic stochastic order} \citep[see][and references therein]{fritz2017resource}. Given $X,Y \in L^\infty$, we say that $X$ dominates $Y$ in the catalytic stochastic order on $L^\infty$ if there exists a $Z \in L^\infty$, independent of $X$ and $Y$, such that $X+Z$ dominates $Y+Z$ in first-order stochastic dominance.

The applicability of this order to our problem is immediate. If $X$ dominates $Y$ in the catalytic stochastic order then
$$
  \Phi(X+Z) \geq \Phi(Y+Z)
$$
for some $Z$, independent of $X$ and $Y$. If $\Phi$ is also additive, then $\Phi(X+Z) = \Phi(X)+\Phi(Z)$ and $\Phi(Y+Z) = \Phi(Y)+\Phi(Z)$, and so we have that $\Phi(X) \geq \Phi(Y)$. Thus, any monotone additive $\Phi$ is monotone with respect to this order.

Clearly, if $X \geq_1 Y$ then $X$ also dominates $Y$ in the catalytic stochastic order, as one can take $Z=0$. A priori, one may conjecture that this is also a necessary condition. %As we next show, this is far from true.
But as Figure~\ref{fig:marginal} shows, it is easy to give examples of two random variables $X$ and $Y$ that are not ranked with respect to first-order stochastic dominance, but are ranked with respect to the catalytic stochastic order.\footnote{We are indebted to the late Kim Border for helping us construct this example.} The random variable $X$ equals $1$ with probability $1/3$ and $0$ with probability $2/3$, while $Y$ is uniformly distributed on $[-\frac{3}{5},\frac{2}{5}]$. As the figure shows, their c.d.f.s are not ranked, and hence they are not ranked in terms of first-order stochastic dominance.\footnote{\cite*{pomatto2018stochastic} give examples of random variables $X$ and $Y$ that are not ranked in stochastic dominance, but are ranked after adding an {\em unbounded} independent $Z$. In fact, they show that this is possible whenever $\E{X}>\E{Y}$. As we explain below, this result no longer holds when $Z$ is required to be bounded.}
\begin{figure}
    \centering
    \includegraphics[scale=0.4]{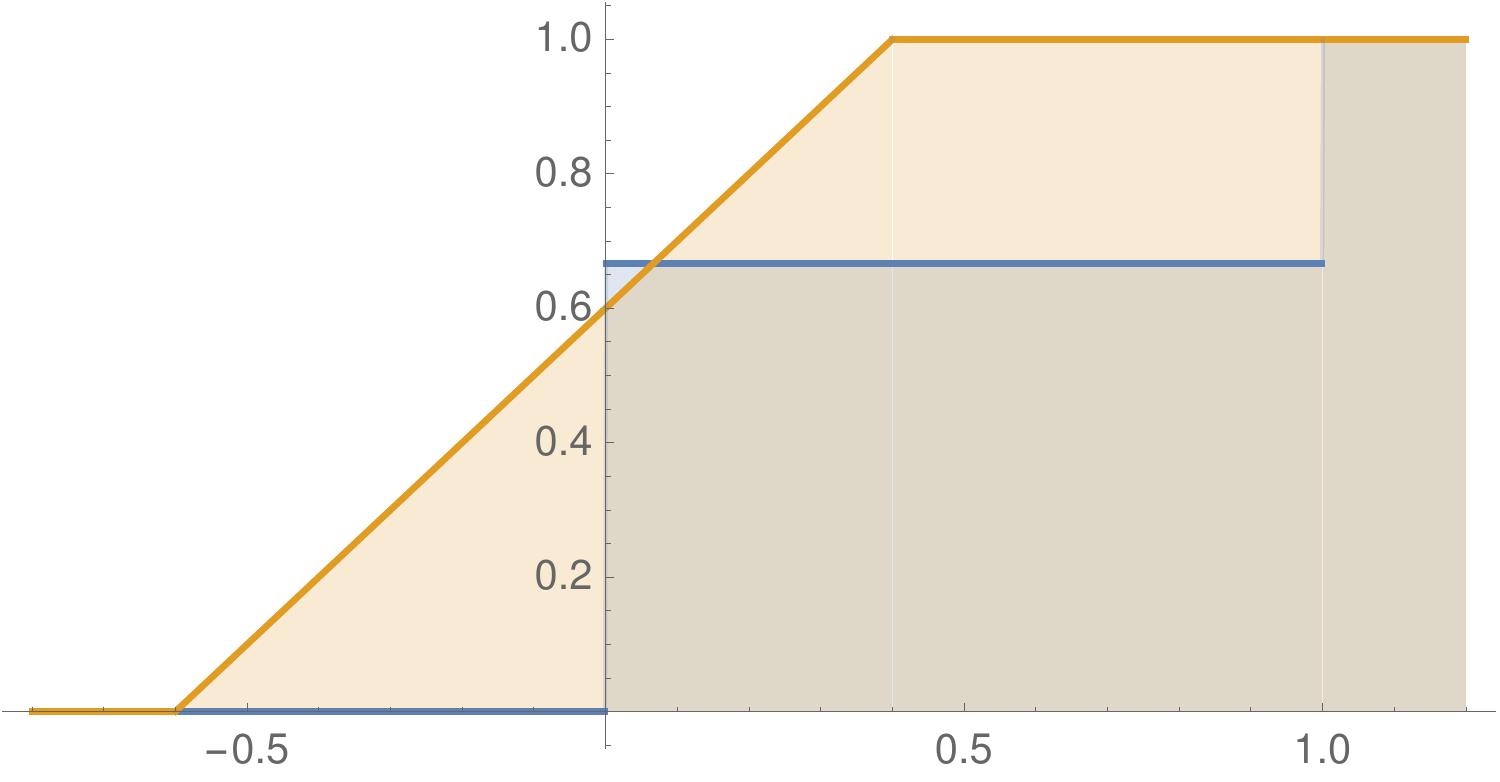}
    \caption{The c.d.f.s of $X$ (blue) and $Y$ (orange).}
    \label{fig:marginal}
\end{figure}
However, if we let $Z$ assign probability half to $\pm \frac{1}{5}$, then $X+Z >_1 Y+Z$. Intuitively, since the c.d.f.\ of $X+Z$ is the average of the two translations (by $\pm \frac{1}{5}$) of the c.d.f.\ of $X$, and since the same holds for the c.d.f.\ of $Y$, the result of adding $Z$ is the disappearance of the small ``kink'' in which the ranking of the c.d.f.s is reversed. This is depicted in Figure~\ref{fig:sum-cdfs}.
\begin{figure}
    \centering
    \includegraphics[scale=0.4]{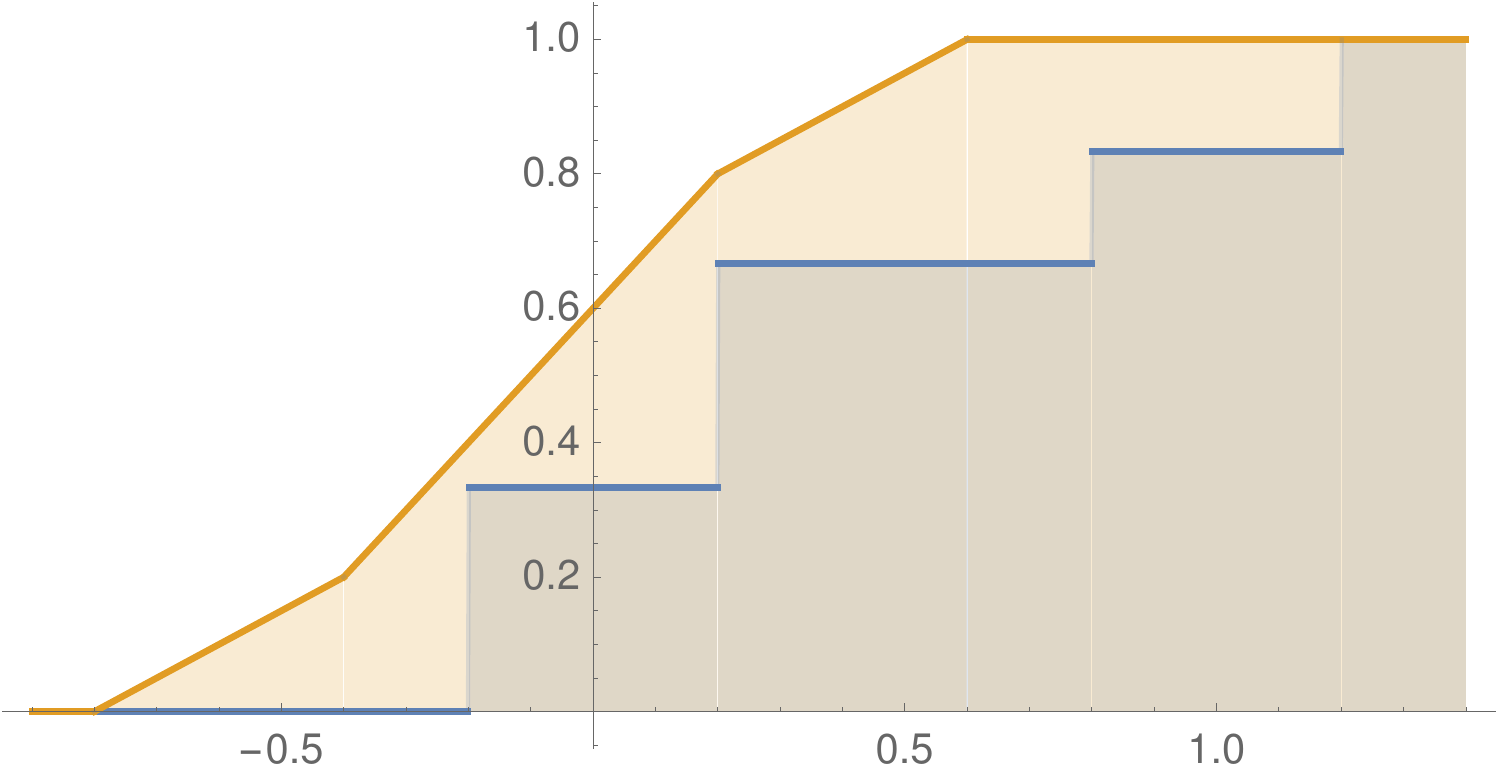}
    \caption{The c.d.f.s of $X+Z$ (blue) and $Y+Z$ (orange).}
    \label{fig:sum-cdfs}
\end{figure}

Every monotone additive statistic $\Phi$ provides an obstruction to dominance in the catalytic stochastic order. That is, if $\Phi(X) < \Phi(Y)$, then it is impossible that $X+Z \geq_1 Y+Z$ for some independent $Z$, since monotonicity would imply that $\Phi(X+Z) \geq \Phi(Y+Z)$, and additivity would then imply that $\Phi(X) \geq \Phi(Y)$. In particular, the existence of an $a$ for which $K_a(X) < K_a(Y)$ forms an obstruction to the existence of such a $Z$. The following result shows that these are, in a sense, the only possible obstructions:\footnote{In fact, except for the trivial case where $X$ and $Y$ have the same distribution, the  strict inequality $K_a(X) > K_a(Y)$ for all $a \in \R$ is necessary for the existence of a $Z$ such that $X+Z \geq_1 Y+Z$. This is because $X + Z \geq_1 Y + Z$ implies the strict inequality $K_a(X+Z) > K_a(Y+Z)$ for finite $a$ whenever $X+Z$ and $Y+Z$ have different distributions. Thus, Theorem~\ref{thm:marginal} below implies that for distributions with different minima and maxima, the condition $K_a(X) > K_a(Y)$ for all $a \in \R$ is both necessary and sufficient for dominance in the catalytic stochastic order.}
\begin{theorem}
  \label{thm:marginal}
  Let $X,Y\in L^\infty$ satisfy $K_a(X) > K_a(Y)$ for all $a \in \overline{\R}$. Then there exists a c.d.f.\ $H$ such that any independent $Z \in L^\infty$ with c.d.f.\ $H$ satisfies $X+Z \geq_1 Y+Z$.
\end{theorem}

 To prove Theorem~\ref{thm:marginal} we explicitly construct $H$ as a truncated Gaussian c.d.f.\  with appropriately chosen parameters. The idea behind the proof is as follows. Denote by $F$ and $G$ the c.d.f.s of $X$ and $Y$, respectively, and suppose that they are supported on $[-N,N]$. Let $h(x)=\frac{1}{\sqrt{2\pi V}} \ee^{-\frac{x^2}{2V}}$ be the density of a Gaussian  $Z$. Then the c.d.f.s of $X+Z$ and $Y+Z$ are given by the convolutions $F * h$ and $G * h$, and their difference is equal to 
\begin{align*}
    [G*h-F*h](y) &= \int_{-N}^N [G(x)-F(x)] \cdot h(y-x)\,\dd x\\
                   &= \frac{1}{\sqrt{2\pi V}}\ee^{-\frac{y^2}{2V}} \cdot \int_{-N}^{N} \underbrace{[G(x)-F(x)] \cdot
    \ee^{\frac{y}{V} \cdot x}}_{(*)} \cdot \underbrace{\ee^{-\frac{x^2}{2V}}}_{(**)} \,\dd x
\end{align*}
If we denote $a=\frac{y}{V}$, then by integration by parts, the integral of just $(*)$ is equal to $\frac{1}{a}\left(\E{\ee^{aX}}-\E{\ee^{aY}}\right)$, which is positive by the assumption that $K_a(X) > K_a(Y)$ and is in fact bounded away from zero by the Extreme Value Theorem. The term $(**)$ can be made arbitrarily close to 1---uniformly on the integral domain $[-N,N]$---by making $V$ large. This implies that $[G*h-F*h](y) > 0$ for all $y$, and we further show that the inequality still holds if we modify $H$ by truncating its tails, ensuring that it is in $L^\infty$.

Theorem~\ref{thm:marginal} leads to the following lemma, which is a key component of the proof of Theorem~\ref{thm:main}:
\begin{lemma}
  \label{lemma:monotone}
  Let $\Phi \colon L^\infty \to \R$ be a monotone additive statistic. If $K_a(X) \geq K_a(Y)$ for all $a \in \overline{\R}$ then $\Phi(X) \geq \Phi(Y)$.
\end{lemma}

\begin{proof}
  Suppose $K_a(X) \geq K_a(Y)$ for all $a \in \overline{\R}$. Given $\eps > 0$, let $\hat{X}$, $\hat{Y}$ and $Z$ in $L$ be such that: $\hat{X}$ has the same c.d.f.\ as $X+\eps$, $\hat{Y}$ has the same c.d.f.\ as $Y$, and $Z$ has the c.d.f.\ obtained by applying Theorem~\ref{thm:marginal} to $\hat{X}$ and $\hat{Y}$. We can indeed apply the theorem, since $K_a(\hat{X}) = K_a(X)+\eps > K_a(Y) = K_a(\hat{Y})$ for all $a$. Hence, $\hat{X} + Z \geq_1 \hat{Y} + Z$. Thus, by monotonicity of $\Phi$, $\Phi(\hat{X} + Z) \geq \Phi(\hat{Y}+Z)$, and by additivity $\Phi(\hat{X}) \geq \Phi(\hat{Y})$. This means that $\Phi(X) + \eps = \Phi(\hat{X}) \geq \Phi(\hat{Y})=\Phi(Y)$ for all $\eps > 0$, and hence $\Phi(X) \geq \Phi(Y)$.
\end{proof}

Once we have established Lemma~\ref{lemma:monotone}, the remainder of the proof uses functional analysis techniques (in particular the Riesz Representation Theorem) to deduce the integral representation in Theorem~\ref{thm:main}. See \S\ref{appx:main} in the appendix for the complete proof.

 An alternative proof of Lemma~\ref{lemma:monotone} can be given based on a different stochastic order. Given two random variables $X$ and $Y$, let $X_1,X_2,\ldots$ and $Y_1,Y_2,\ldots$ be i.i.d.\ copies of $X$ and $Y$, respectively. We say that $X$ dominates $Y$ in \textit{large numbers} if 
\begin{align*}
    X_1+\cdots+X_n \geq_1 Y_1+\cdots+Y_n
\end{align*}
for all $n$ large enough. Using large-deviations techniques, it was shown by \cite{aubrun2008catalytic} that if $K_a(X) > K_a(Y)$ for all $a \in \overline{\R}$, then $X$ dominates $Y$ in large numbers. This implies Lemma~\ref{lemma:monotone} since, by the additivity of $\Phi$, $\Phi(X) \geq \Phi(Y)$ holds if and only if $n\Phi(X) = \Phi(X_1+\cdots+X_n) \geq \Phi(Y_1+\cdots+Y_n) = n\Phi(Y)$. 

 Compared to this alternative argument, our proof of Lemma~\ref{lemma:monotone} based on Theorem~\ref{thm:marginal} is self-contained and more elementary. More importantly, (an analogue of) the catalytic stochastic order established in Theorem~\ref{thm:marginal} is essential for studying monotone additive statistics defined on a domain of unbounded random variables, for which the large numbers order is difficult to characterize as far as we know.\footnote{One particular challenge is that the large numbers order require a uniform comparison between the tail probabilities of $X_1 + \cdots + X_n$ versus those of $Y_1 + \cdots + Y_n$, for a fixed large $n$. For a given threshold of the tail, large-deviations theory can be used to show the desired comparison when $n$ is large enough. But making the required $n$ uniform across all thresholds becomes nontrivial when the random variables $X$ and $Y$ are unbounded.} This generalization of Theorem~\ref{thm:marginal} is presented in Lemma \ref{lemma:marginal-unbounded} in the online appendix, as a key step toward the proof of Theorem \ref{thm:domains}. %and is directly useful for proving some of our other results, including in particular Theorem~\ref{thm:risk-invariance}.

\clearpage

\bibliography{refs}

\clearpage

\appendix
%%%%%%%%
\begin{center}
    {\Large \textbf{Appendix}}
\end{center}

\bigskip
The appendix contains the omitted proofs for most of the results in the main text, in the order in which they appeared. The only exceptions are Theorem~\ref{thm:domains} regarding the larger domain $L_M$, Proposition \ref{prop:time-lotteries-strong} regarding strong stochastic dynamic consistency and a few results in Section \ref{sec:monetary}, whose proofs are relegated to the online appendix. %Additional results and proofs are presented in a separate supplementary appendix.

%The online appendix also describes additional results that have been mentioned, and provides their proofs. 

Throughout the proofs we will often use the notation $K_X(a) = K_a(X)$, so that $K_X$ is a map from $\overline{\R}$ to $\R$. The following facts are standard: 
\begin{lemma}
\label{lemma:L}
Let $X, Y \in L^\infty$.
\begin{enumerate}
    \item $K_X \colon \overline{\R} \to \R$ is well defined, non-decreasing and continuous. 
    \item If $K_X=K_Y$ then $X$ and $Y$ have the same distribution.
\end{enumerate}
\end{lemma}
\begin{proof}
 Over $\R$ the map $K_X$ is continuous and non-decreasing. This follows directly from the fact that $K_X(a)$ is the certainty equivalent of a CARA expected utility preference with coefficient of risk aversion equal to $-a$. That $\lim_{a \to \infty}K_X(a) = \max[X]$ and $\lim_{a \to -\infty}K_X(a) = \min[X]$  follow from a simple application of Laplace's method. It is a standard fact that $K_X=K_Y$ implies that $X$ and $Y$ have the same distribution \citep[see for instance][]{curtiss1942note}.
\end{proof}

\section{Proof of Theorem~\ref{thm:main}}\label{appx:main}

We follow the proof outlined in \S\ref{sec:proof-sketch} of the main text and first establish Theorem~\ref{thm:marginal}. 

\subsection{Proof of Theorem~\ref{thm:marginal}}

First, we can add the same constant $b$ to both $X$ and $Y$ so that $\min[Y+b] = -N$ and $\max[X+b] = N$ for some $N > 0$. Since translating both $X$ and $Y$ leaves the existence of an appropriate $Z$ unchanged (and also does not affect $K_X > K_Y$), we henceforth assume without loss of generality that $\min[Y] = -N$, and $\max[X] = N$. Since $K_X > K_Y$, we know that $\min[X] > -N$ and $\max[Y] < N$.

Denote the c.d.f.s of $X$ and $Y$ by $F$ and $G$, respectively. Let $\sigma(x) = G(x)-F(x)$. Note that $\sigma$ is supported on   $[-N,N]$ and bounded in absolute value by $1$. Moreover, by choosing $\eps > 0$ sufficiently small, we have that $\min[X] > -N+\eps$ and $\max[Y] < N-\eps$. So $\sigma(x)$ is positive on $[-N, -N+\eps]$ and on $[N-\eps, N]$. In fact, there exists $\delta > 0$ such that $\sigma(x) \geq \delta$ whenever $x \in [-N+\frac{\eps}{4}, -N+\frac{\eps}{2}]$ and $x \in [N-\frac{\eps}{2}, N-\frac{\eps}{4}]$. We also fix a large constant $A$ such that 
\begin{align*}
    \ee^{\frac{\eps A}{4}} \geq \frac{8N}{\eps \delta}.
\end{align*}
  
Define
$$
M_\sigma(a) = \int_{-N}^N \sigma(x)\ee^{ax}\,\dd x.
$$
Note that for $a \neq 0$, integration by parts shows $M_\sigma(a) = \frac{1}{a}\left(\E{ \ee^{aX}}-\E{\ee^{aY}} \right)$, and that $M_\sigma(0) = \E{X}-\E{Y}$. Therefore, since $K_X > K_Y$, we have that $M_\sigma$ is strictly positive everywhere. Since $M_\sigma(a)$ is clearly continuous in $a$, it is in fact bounded away from zero on any compact interval.

We will use these properties of $\sigma$ to construct a truncated Gaussian density $h$ such that 
\[
    [\sigma * h](y) = \int_{-N}^{N} \sigma(x) h(y-x) \,\dd x\geq 0
\]
for each $y \in \R$. If we let $Z$ be a random variable independent from $X$ and $Y$, whose distribution has density function $h$, then $\sigma * h = (G-F) * h$ is the difference between the c.d.f.s of $Y+Z$ and $X+Z$. Thus $[\sigma * h](y) \geq 0$ for all $y$ would imply $X + Z \geq_1 Y + Z$.

To do this, we write $h(x) = \ee^{-\frac{x^2}{2V}}$ for all $\vert x \vert \leq T$, where $V$ is the variance and $T$ is the truncation point to be chosen.\footnote{In general we need a normalizing factor to ensure $h$ integrates to one, but this multiplicative constant does not affect the argument.} We will show that given the above constants $N$ and $A$, $[\sigma * h](y)\geq 0$ holds for each $y$ when $V$ is sufficiently large and $T \geq AV + N$ . 
 
First consider the case where $y \in [-AV, AV]$. In this region, $\vert y-x \vert \leq T$ is automatically satisfied when $x \in [-N, N]$. So we can compute the convolution $\sigma * h$ as follows:
\begin{align}\label{eq:convo}
    \int \sigma(x) h(y-x) \,\dd x = \ee^{-\frac{y^2}{2V}} \cdot \int_{-N}^{N} \sigma(x) \cdot
    \ee^{\frac{y}{V} \cdot x} \cdot \ee^{-\frac{x^2}{2V}} \,\dd x.
\end{align}
Note that $\frac{y}{V}$ in the exponent belongs to the compact interval $[-A, A]$. So for our fixed choice of $A$, the integral  $M_\sigma(\frac{y}{V})=\int_{-N}^{N} \sigma(x) \cdot \ee^{\frac{y}{V} \cdot x} \,\dd x$ is uniformly bounded away from zero when $y$ varies in the current region. Thus, 
\begin{equation}\label{eq:convo1}
\begin{split}
\int_{-N}^{N} \sigma(x) \cdot
    \ee^{\frac{y}{V} \cdot x} \cdot \ee^{-\frac{x^2}{2V}}\,\dd x &=  M_\sigma\left(\frac{y}{V}\right) - \int_{-N}^{N} \sigma(x) \cdot
    \ee^{\frac{y}{V} \cdot x} \cdot (1-\ee^{-\frac{x^2}{2V}})\,\dd x \\
    &\geq M_\sigma\left(\frac{y}{V}\right) - 2N \cdot \ee^{AN} \cdot (1- \ee^{\frac{-N^2}{2V}}),
\end{split}
\end{equation}
which is positive when $V$ is sufficiently large. So the right-hand side of \eqref{eq:convo} is positive. 

Next consider the case where $y \in (AV, T+N-\eps]$; the case where $-y$ is in this range can be treated symmetrically. Here the convolution can be written as
\[
    [\sigma * h](y) = \int_{\max\{-N, y - T\}}^{N} \sigma(x) \cdot \ee^{\frac{-(y-x)^2}{2V}}
    \,\dd x.
\]
We break the range of integration into two sub-intervals: $I_1 = [\max\{-N, y - T\}, N-\eps]$ and $I_2 = [N-\eps, N]$. On $I_1$ we have $\sigma(x) = G(x) - F(x) \geq -1$. As long as $AV \geq N - \varepsilon$, we have $\ee^{\frac{-(y-x)^2}{2V}} \leq \ee^{\frac{-(y-N+\eps)^2}{2V}}$ for $y > AV$ and $x \leq N-\varepsilon$, and thus
\[
   \int_{x \in I_1} \sigma(x) \cdot \ee^{\frac{-(y-x)^2}{2V}} \,\dd x  \geq -2N \cdot \ee^{\frac{-(y-N+\eps)^2}{2V}}.
\]
On $I_2$ we have $\sigma(x) \geq 0$ by our choice of $\eps$, and furthermore $\sigma(x) \geq \delta$ when $x \in [N-\frac{\eps}{2}, N-\frac{\eps}{4}]$. Thus
\[
    \int_{x \in I_2} \sigma(x) \cdot \ee^{\frac{-(y-x)^2}{2V}} \,\dd x \geq \frac{\eps}{4} \cdot \delta \cdot \ee^{\frac{-(y-N+\frac{\eps}{2})^2}{2V}} \geq 2N \cdot \ee^{\frac{-(y-N+\frac{\eps}{2})^2}{2V} - \frac{\eps A}{4}},
\]
where the second inequality holds by the choice of $A$. Observe that when $y > AV$ and $V$ is large, the exponent $\frac{-(y-N+\frac{\eps}{2})^2}{2V} - \frac{\eps A}{4}$ is larger than $\frac{-(y-N+\eps)^2}{2V}$. Summing the above two inequalities then yields the desired result that $[\sigma * h](y) \geq 0$.

Finally, if $y \in (T+N-\eps, T+N]$, then the range of integration
in computing $[\sigma * h](y)$ is from $x = y-T$ to $x = N$, where $\sigma(x)$ is always positive. So the convolution is positive. And if $y > T+N$, then clearly the convolution is zero. These arguments symmetrically apply to $-y \in (T+N-\eps, T+N]$ and $-y > T+N$. We therefore conclude that $[\sigma * h](y) \geq 0$ for all $y$, completing the proof.

\subsection{Integral Representation}

For fixed $X$, $K_X(a) = K_a(X)$ is a function of $a$, from $\overline{\R}$ to $\R$. Let $\mathcal{L}$ denote the set of functions $\{K_X:\, X \in L^\infty\}$. If $\Phi$ is a monotone additive statistic and $K_X = K_Y$, then $X$ and $Y$ have the same distribution and $\Phi(X) = \Phi(Y)$. Thus there exists some functional $F \colon \mathcal{L} \to \R$ such that $\Phi(X) = F(K_X)$. It follows from the additivity of $\Phi$ and the additivity of $K_a$ that $F$ is additive: $F(K_X+K_Y) = F(K_X)+F(K_Y)$.\footnote{We note that $\mathcal{L}$ is closed under addition. This is because $K_X + K_Y = K_{X'} + K_{Y'}$ whenever $X',Y'$ are independently distributed random variables with the same distribution as $X,Y$.
Such random variables $X', Y'$ exist as the probability space is non-atomic, see for example Proposition~9.1.11 in \cite{bogachev2007measure}. Thus, for $K_X,K_Y \in \mathcal{L}$ we can find $X', Y'$ so that $K_X + K_Y = K_{X'} + K_{Y'} = K_{X'+Y'} \in \mathcal{L}$. } Moreover, $F$ is monotone in the sense that $F(K_X) \geq F(K_Y)$ whenever $K_X \geq K_Y$ (i.e., $K_X(a) \geq K_Y(a)$ for all $a \in \overline{\R}$); this follows from Lemma~\ref{lemma:monotone} which in turn is proved by Theorem~\ref{thm:marginal} (see \S\ref{sec:proof-sketch} in the main text).

The rest of this proof is a functional analysis exercise analogous to the proof of Theorem 2 in \cite*{mu2021blackwell}, but for completeness we provide the details below. The main goal is to show that the monotone additive functional $F$ on $\mathcal{L}$ can be extended to a positive linear functional on the entire space of continuous functions $\mathcal{C}(\overline{\R})$. We first equip $\mathcal{L}$ with the sup-norm of $\mathcal{C}(\overline{\R})$ and establish a technical claim. 

\begin{lemma}\label{lemma:Lipschitz}
$F \colon \mathcal{L} \to \R$ is 1-Lipschitz: 
\begin{align*}
    |F(K_X) - F(K_Y)| \leq \Vert K_X - K_Y\Vert.
\end{align*}
\end{lemma}

\begin{proof}
Let $\Vert K_X - K_Y\Vert = \eps$. Then $K_{X+\eps} = K_X + \eps \geq K_Y.$
Hence, by Lemma~\ref{lemma:monotone}, $F(K_Y) \leq F(K_{X+\eps})$, and so 
\begin{align*}
    F(K_Y) - F(K_X) \leq F(K_{X+\eps}) - F(K_X)= F(K_\eps) = \Phi(\eps) = \eps.
\end{align*}
Symmetrically we have $F(K_X) - F(K_Y) \leq \eps$, as desired.
\end{proof}

\begin{lemma}\label{lemma:extend}
Any monotone additive functional $F$ on $\mathcal{L}$ can be extended to a positive linear functional on $\mathcal{C}(\overline{\R})$. 
\end{lemma}

\begin{proof}
First consider the rational cone spanned by $\mathcal{L}$:
$$
    \text{Cone}_{\mathbb{Q}}(\mathcal{L}) = \{q L : q \in \mathbb{Q}_+, L \in \mathcal{L} \}.
$$
Define $G \colon\text{Cone}_{\mathbb{Q}}(\mathcal{L}) \to \R$ as $G(q L) = q F(L)$, which is an extension of $F$. The functional $G$ is well defined: If $\frac{m}{n} K_1 = \frac{r}{n} K_2$ for $K_1,K_2 \in \mathcal{L}$ and $n,m,r \in \mathbb{N}$, then, using the fact that $\mathcal{L}$ is closed under addition, we obtain $mF(K_1) = F(mK_1) = F(r K_2) = r F(K_2)$, hence $\frac{m}{n} F(K_1) = \frac{r}{n} F(K_2)$. $G$ is also additive, because
\[
    G\left(\frac{m}{n} K_1 \right ) + G\left(\frac{r}{n} K_2\right) = \frac{m}{n} F\left(K_1\right) + \frac{r}{n} F\left( K_2\right) = \frac{1}{n}F(mK_1 + rK_2) = G\left(\frac{m}{n} K_1 + \frac{r}{n} K_2\right).
\]
In the same way we can show $G$ is positively homogeneous over $\mathbb{Q}_{+}$ and monotone. 

Moreover, $G$ is Lipschitz: Lemma~\ref{lemma:Lipschitz} implies
\[
    \left\vert G\left(\frac{m}{n} K_1 \right) - G\left(\frac{r}{n} K_2 \right) \right \vert = \frac{1}{n} \left\vert F(mK_1) - F\left(rK_2 \right) \right\vert \leq \frac{1}{n} \left\Vert mK_1 - rK_2 \right\Vert =  \left\Vert \frac{m}{n}K_1 - \frac{r}{n}K_2 \right\Vert.
\]
Thus $G$ can be extended to a Lipschitz functional $H$ defined on the closure of $\text{Cone}_{\mathbb{Q}}(\mathcal{L})$ with respect to the sup norm. In particular, $H$ is defined on the convex cone spanned by $\mathcal{L}$:
\[
    \text{Cone}(\mathcal{L}) = \{\lambda_1 K_1 + \cdots + \lambda_k K_k : k \in \mathbb{N} \text{ and for each } 1\leq i \leq k, \lambda_i \in \R_+, K_i \in \mathcal{L}\}.
\]
It is immediate to verify that the properties of additivity, positive homogeneity (now over $\R_+$), and monotonicity extend, by continuity, from $G$ to $H$.

Consider the vector subspace $\mathcal{V} = \text{Cone}(\mathcal{L}) - \text{Cone}(\mathcal{L}) \subset \mathcal{C}(\overline{\R})$ and define $I \colon\mathcal{V} \to \R$ as
$$
    I(g_1  - g_2) = H(g_1) - H(g_2)
$$
for all $g_1, g_2 \in \text{Cone}(\mathcal{L})$. The functional $I$ is well defined and linear (because $H$ is additive and positively homogeneous). Moreover, by monotonicity of $H$, $I(f) \geq 0$ for any non-negative function $f \in \mathcal{V}$.

The lemma then follows from the next theorem of \cite{kantorovich1937moment}, a generalization of the Hahn-Banach Theorem. It applies not only to $\mathcal{C}(\overline{\R})$ but to any Riesz space \citep[see Theorem 8.32 in][]{guide2006infinite}. 

%\medskip

\begin{theorem*}
Let $\mathcal{V}$ be a vector subspace of $\mathcal{C}(\overline{\R})$ with the property that for every $f \in \mathcal{C}(\overline{\R})$ there exists a function $g \in \mathcal{V}$ such that $g \geq f$. Then every positive linear functional on $\mathcal{V}$ extends to a positive linear functional on $\mathcal{C}(\overline{\R})$.
\end{theorem*}

The ``majorization'' condition $g \geq f$ is satisfied because every function in $\mathcal{C}(\overline{\R})$ is bounded and $\mathcal{V}$ contains all of the constant functions.
\end{proof}

The integral representation in Theorem~\ref{thm:main} now follows from Lemma~\ref{lemma:extend} by the Riesz-Markov-Kakutani Representation Theorem.

\subsection{Uniqueness of Mixing Measure}

We complete the proof of Theorem~\ref{thm:main} by showing that the mixing measure $\mu$ is unique: %The following result shows that uniqueness holds even on the smaller domain $L^{\infty}_{\mathbb{N}}$ of non-negative integer-valued random variables. 

\begin{lemma}
\label{lemma:unique}
Suppose $\mu$ and $\nu$ are two Borel probability measures on $\overline{\R}$ such that 
\begin{align*}
    \int_{\overline{\R}}K_a(X)\,\dd\mu(a) = \int_{\overline{\R}}K_a(X)\,\dd\nu(a).
\end{align*}
for all $X \in L^{\infty}$.\footnote{The proof shows that it suffices to require such equality for non-negative integer-valued $X$.} Then $\mu = \nu$.
\end{lemma}

\begin{proof}
We first show $\mu(\{\infty\}) = \nu(\{\infty\})$. For any $\eps > 0$, consider the Bernoulli random variable $X_{\eps}$ that takes value $1$ with probability $\eps$ and value 0 with probability $1-\eps$. It is easy to see that as $\eps$ decreases to zero, $K_a(X_{\eps})$ also decreases to zero for each $a < \infty$ whereas $K_{\infty}(X_{\eps}) = \max[X_{\eps}] = 1$. Since $K_a(X_{\eps})$ is uniformly bounded in $[0,1]$, the Dominated Convergence Theorem implies 
\[
\lim_{\eps \to 0}  \int_{\overline{\R}}K_a(X_{\eps})\,\dd\mu(a) = \mu(\{\infty\}). 
\]
A similar identity holds for the measure $\nu$, so $\mu(\{\infty\}) = \nu(\{\infty\})$ follows from the assumption that $\int_{\overline{\R}}K_a(X_{\eps})\,\dd\mu(a) = \int_{\overline{\R}}K_a(X_{\eps})\,\dd\nu(a)$. 

We can symmetrically apply the above argument to the Bernoulli random variable that takes value 1 with probability $1-\eps$ and value $0$ with probability $\eps$. Thus $\mu(\{-\infty\}) = \nu(\{-\infty\})$ holds as well.

Next, for each $n \in \mathbb{N}_{+}$ and real number $b  > 0$, define a random variable $X_{n,b}$ by
\begin{align*}
    \Pr{X_{n,b}=n} &= \ee^{-bn}\\
    \Pr{X_{n,b}=0} &= 1-\ee^{-bn}.
\end{align*}
Then $K_a(X_{n,b}) = \frac{1}{a}\log\left[(1-\ee^{-bn}) + \ee^{(a-b)n}\right]$, and so
\begin{align*}
    \lim_{n \to \infty} \frac{1}{n}K_a(X_{n,b}) 
    &= \lim_{n \to \infty} \frac{1}{n}\frac{1}{a}\log\left[1 -  \ee^{-bn} + \ee^{(a-b)n}\right] \nonumber \\
    &=
    \begin{cases}
    0 & \text{if } a < b\\
    \frac{a-b}{a}& \text{if } a \geq b.
    \end{cases}
\end{align*}
This result holds also for $a = 0, \pm\infty$.

Note that $\frac{1}{n}K_a(X_{n,b})$ is uniformly bounded in $[0,1]$ for all values of $n, b, a$, since $K_a(X_{n,b})$ is bounded between $\min[X_{n,b}] = 0$ and $\max[X_{n,b}] = n$. Thus, by the Dominated Convergence Theorem, 
\begin{align}\label{eq:comparison1}
    \lim_{n \to \infty} \int_{\overline{\R}} \frac{1}{n}K_a(X_{n,b})\,\dd\mu(a)
    = \int_{[b,\infty]}\frac{a-b}{a}\,\dd\mu(a),
\end{align}
and similarly for $\nu$. It follows that for all $b > 0$,
\begin{align*}
    \int_{[b,\infty]}\frac{a-b}{a}\,\dd\mu(a) = \int_{[b,\infty]}\frac{a-b}{a}\,\dd\nu(a).
\end{align*}
As $\mu(\{\infty\}) = \nu(\{\infty\})$, we in fact have
\begin{align*}
    \int_{[b,\infty)}\frac{a-b}{a}\,\dd\mu(a) = \int_{[b,\infty)}\frac{a-b}{a}\,\dd\nu(a).
\end{align*}
This common integral is denoted by $f(b)$. 

We now define a measure $\hat{\mu}$ on $(0, \infty)$ by the condition $\frac{\dd \hat{\mu}(a)}{\dd \mu(a)} = \frac{1}{a}$; note that $\hat{\mu}$ is a positive measure, but need not be a probability measure. Then  
\begin{align*}
    f(b) = \int_{[b,\infty)}\frac{a-b}{a}\,\dd\mu(a) &= \int_{[b,\infty)}(a-b) \,\dd\hat{\mu}(a) 
    = \int_{b}^{\infty} \hat{\mu}([x, \infty)) \,\dd x, 
\end{align*}
where the last step uses Tonelli's Theorem. Hence $\hat{\mu}([b, \infty])$ is the negative of the left derivative of $f(b)$ (this uses the fact that $\hat{\mu}([b, \infty])$ is left continuous in $b$). In the same way, if we define $\hat{\nu}$ by $\frac{\dd \hat{\nu}(a)}{\dd \nu(a)} = \frac{1}{a}$, then $\hat{\nu}([b, \infty])$ is also the negative of the left derivative of $f(b)$. Therefore $\hat{\mu}$ and $\hat{\nu}$ are the same measure on $(0, \infty)$, which implies that $\mu$ and $\nu$ coincide on $(0, \infty)$. 

By a symmetric argument (with $n - X_{n,b}$ in place of $X_{n,b})$, we deduce that $\mu$ and $\nu$ also coincide on $(-\infty, 0)$. Finally, since they are both probability measures, $\mu$ and $\nu$ must have the same mass at $0$, if any. So $\mu = \nu$. 
\end{proof}

\section{Applications to Time Lotteries}

\subsection{Monotone Additive Statistics for Non-Negative Random Variables}

In our applications to time lotteries the random times are non-negative (bounded) random variables. We accordingly prove a version of Theorem~\ref{thm:main} that applies to this smaller domain.

\begin{proposition}\label{prop:non-negative}
 $\Phi \colon L^\infty_+ \to \R$ is a monotone additive statistic if and only if there exists a unique Borel probability measure $\mu$ on $\overline{\R}$ such that for every $X \in L^\infty$
\begin{align}
    \Phi(X) = \int_{\overline{\R}}K_a(X)\,\dd\mu(a).
\end{align}
 \end{proposition}
 \begin{proof}
 It suffices to show that a monotone additive statistic defined on $L^{\infty}_{+}$ can be extended to a monotone additive statistic defined on $L^{\infty}$. Suppose $\Phi$ is defined on $L^{\infty}_{+}$. Then for any bounded random variable $X$, we can define 
\[
\Psi(X) = \min[X] + \Phi(X - \min[X]), 
\]
where we note that $X - \min[X]$ is a non-negative random variable. 

Clearly $\Psi$ is a statistic that depends only on the distribution of $X$ (as $\Phi$ does), and $\Psi(c) = c + \Phi(0) = c$ for constants $c$. When $X$ is non-negative, the additivity of $\Phi$ gives $\Phi(X) = \Phi(\min[X]) + \Phi(X - \min[X]) = \min[X] + \Phi(X - \min[X])$, so $\Psi$ is an extension of $\Phi$. Moreover, $\Psi$ is additive because $\min[X+Y] = \min[X] + \min[Y]$, and $\Phi(X+Y-\min[X+Y]) = \Phi(X-\min[X]) + \Phi(Y-\min[Y])$ by the additivity of $\Phi$. Finally, to show $\Psi$ is monotone, suppose $X$ and $Y$ are bounded random variables satisfying $X \geq_1 Y$. Then we can choose a sufficiently large $n$ such that $X+n$ and $Y+n$ are both non-negative, and $X+n \geq_1 Y+n$. Since $\Phi$ is monotone for non-negative random variables, $\Phi(X+n) \geq \Phi(Y+n)$. Thus $\Psi(X+n) \geq \Psi(Y+n)$ by the fact that $\Psi$ extends $\Phi$, and $\Psi(X) \geq \Psi(Y)$ by the additivity of $\Psi$. This proves that $\Psi$ is a monotone additive statistic on $L^{\infty}$ that extends $\Phi$.
 \end{proof}

\subsection{Proof of Theorem~\ref{thm:time-lotteries}}

It is straightforward to check that the representation satisfies the axioms, so we focus on the other direction of deriving the representation from the axioms. In the first step, we fix any reward $x > 0$. Then by monotonicity in time and continuity, for each $(x, T)$ there exists a (unique) deterministic time $\Phi_x(T)$ such that $(x, \Phi_x(T)) \sim (x, T)$.
Clearly, when $T$ is a deterministic time, $\Phi_x(T)$ is simply $T$ itself. Note also that if $S$ first-order stochastically dominates $T$, then
\[
(x, \Phi_x(T)) \sim (x, T) \succeq (x, S) \sim (x, \Phi_x(S)),
\]
so that $\Phi_x(S) \geq \Phi_x(T)$. We next show that for any $T$ and $S$ that are independent, $\Phi_x(T + S) = \Phi_x(T) + \Phi_x(S)$. Indeed, by stochastic stationarity, $(x, \Phi_x(T)) \sim (x, T)$ implies $(x, \Phi_x(T) + S) \sim (x, T + S)$ and $(x, \Phi_x(S)) \sim (x, S)$ implies $(x, \Phi_x(T) + \Phi_x(S)) \sim (x, \Phi_x(T) + S)$. Taken together, we have 
\[
(x, \Phi_x(T) + \Phi_x(S)) \sim (x, T + S).
\]
Since $\Phi_x(T) + \Phi_x(S)$ is a deterministic time, the definition of $\Phi_x$ gives $\Phi_x(T) + \Phi_x(S) = \Phi_x(T+S)$ as desired. It follows that each $\Phi_x \colon L^\infty_+ \to \mathbb{R}$ is a monotone additive statistic.

In the second step, note that our preference $\succeq$ induces a preference on $\R_{++} \times \Rp$ consisting of deterministic dated rewards. By Theorem 2 in \cite{fishburn1982time}, for any given $r > 0$ we can find a continuous and strictly increasing utility function $u\colon\R_{++} \to \R_{++}$ such that for deterministic times $t, s \geq 0$
\[
(x, t) \succeq (y, s) \quad \text{if and only if} \quad u(x) \cdot \ee^{-rt} \geq u(y) \cdot \ee^{-rs}. 
\]
By definition, $(x, T) \sim (x, \Phi_x(T))$ for any random time $T$. Thus we obtain that the decision maker's preference is represented by 
\[
(x, T) \succeq (y, S) \quad \text{if and only if} \quad u(x) \cdot \ee^{-r\Phi_x(T)} \geq u(y) \cdot \ee^{-r\Phi_y(S)}.
\]

It remains to show that for all $x, y > 0$, $\Phi_{x}$ and $\Phi_{y}$ are the same statistic. For this we choose deterministic times $t$ and $s$ such that $(x, t) \sim (y, s)$, i.e., $u(x) \cdot \ee^{-r t} = u(y) \cdot \ee^{-r s}$.
For any random time $T$, stochastic stationarity implies $(x, t + T) \sim (y, s + T)$, so that
\[
u(x) \cdot \ee^{-r \Phi_{x}(t+T)} = u(y) \cdot \ee^{-r \Phi_{y}(s+T)}.
\]
Using the additivity of $\Phi_{x}$ and $\Phi_{y}$, we can divide the above two equalities and obtain $\Phi_{x}(T) = \Phi_{y}(T)$ as desired. Since this holds for all $T$ and all $x, y > 0$, we can write $\Phi_x(T) = \Phi(T)$ for a single monotone additive statistic $\Phi$. This completes the proof.

\subsection{Proof of Proposition~\ref{prop:pareto-impossibility}}

 Define, for every $t \geq 0$, $v_i(t) = \ee^{-a_i t}$ and $v(t) = \ee^{-a t}$. We have that for any two random times $S$ and $T$, $(1,S) \succeq_i (1,T)$ if and only if $\E{v_i(S)} \geq \E{v_i(T)}$, and $(1,S) \succeq (1,T)$ if and only if $\E{v(S)} \geq \E{v(T)}$. Thus it follows from the Pareto axiom that for any two random times $S$ and $T$, $\E{v_i(S)} \geq \E{v_i(T)}$ for all $i$ implies $\E{v(S)} \geq \E{v(T)}$.

 By Harsanyi's Theorem \citep[Theorem 2]{zhou1997harsanyi} there exist $(\lambda_i)$ in $\R_+$ and $c \in \R$ such that for every $t$, $v(t) = \sum_i \lambda_i v_i(t) + c$. By letting $t \to \infty$ we obtain $0 = c$ and by setting $t = 0$ it follows that $1 = \sum_i \lambda_i$. Further plugging in $t = 1$ and $t = 2$, we obtain 
 \[
    \sum_{i=1}^n \lambda_i\ee^{-2a_i } = \ee^{-2a} = \left(\ee^{-a}\right)^2 = \left(\sum_{i=1}^n \lambda_i \ee^{-a_i}\right)^2.
 \]
 But the Cauchy-Schwarz inequality gives %L: we could just use Jensen here
 \[
  \sum_{i=1}^n \lambda_i \ee^{-2a_i } = \left(\sum_{i=1}^n \lambda_i \ee^{-2a_i }\right) \cdot \left(\sum_{i=1}^n \lambda_i\right) \geq \left(\sum_{i=1}^n \lambda_i \ee^{-a_i}\right)^2. 
 \]
 Thus equality holds. Since the individual discount rates $\{a_i\}$ are assumed to be distinct, the equality condition of the Cauchy-Schwarz inequality implies that exactly one $\lambda_i$ is nonzero (in fact equal to $1$), and hence $a = a_i$ for some agent $i$. 

 Without loss of generality suppose $a = a_1$. It remains to show that $u(x)$ is a constant multiple of $u_1(x)$ so that the social preference coincides with agent 1. Note that by the same argument as above, $v_1(t) = \ee^{-a_1 t}$ cannot be expressed as a linear combination of $1, v_2(t), v_3(t), \cdots, v_n(t)$  whenever $a_1$ is distinct from $a_2, \cdots, a_n$. So the contrapositive of Harsanyi's Theorem implies the existence of random times $S$ and $T$ such that $\E{v_i(S)} \geq \E{v_i(T)}$ for all $i > 1$ but $\E{v_1(T)} > \E{v_1(S)}$. In what follows we fix these particular $S$ and $T$, and also fix $\eps > 0$ sufficiently small so that $\E{v_1(T)} \geq (1+\eps) \E{v_1(S)}$. 

 For any pair of rewards $x, y \in \R_{++}$, we now show that the Pareto property implies $\frac{u(y)}{u_1(y)} = \frac{u(x)}{u_1(x)}$ which will complete the proof. To do this, let $k$ be a sufficiently large positive integer, and define $T^{\oplus k}$, $S^{\oplus k}$ to be the random variables obtained by adding $k$ independent copies of $T$ and $S$. Since the moment generating function $\E{\ee^{-\alpha Z}}$ is multiplicative when we add two independent random variables $Z_1$ and $Z_2$, our previous assumptions about $S$ and $T$ imply that $\E{\ee^{-a_iS^{\oplus k}}} \geq \E{\ee^{-a_iT^{\oplus k}}}$ for all $i > 1$ but $\E{\ee^{-a_1T^{\oplus k}}} \geq (1+\eps)^k \E{\ee^{-a_1S^{\oplus k}}}$. 
 
 Next, let $t_k \in \R$ be the number that satisfies 
 \[
 \ee^{-a_1 t_k} \cdot u_1(x) \E{\ee^{-a_1T^{\oplus k}}} = u_1(y) \E{\ee^{-a_1S^{\oplus k}}}. 
 \]
 Thus, the time lottery $(x, T^{\oplus k} + t_k)$ is indifferent to $(y, S^{\oplus k})$ for agent 1. At the same time, the above equality implies $e^{a_1t_k} \geq (1+\eps)^k \cdot \frac{u_1(x)}{u_1(y)}$, so that $\lim_{k\to \infty} t_k = \infty$. In particular, we deduce that for $k$ large, $\ee^{a_i t_k} \geq \frac{u_i(x)}{u_i(y)}$ for every $i > 1$ and thus
 \[
 \ee^{-a_i t_k} \cdot u_i(x) \E{\ee^{-a_iT^{\oplus k}}} \leq u_i(y) \E{\ee^{-a_iS^{\oplus k}}}. 
 \]
 Therefore $(x, T^{\oplus k} + t_k)$ is less preferred than $(y, S^{\oplus k})$ for every agent $i > 1$. 
 
 Putting together the above analysis, we can find $k$ and $t_k$ such that $(x, T^{\oplus k} + t_k)$ is weakly less preferred than $(y, S^{\oplus k})$ for every agent, with indifference for agent $1$. By the Pareto property, $(x, T^{\oplus k} + t_k)$ must be weakly less preferred than $(y, S^{\oplus k})$ under the social preference. That is, we must have 
 \[
 \ee^{-a t_k} \cdot u(x) \E{\ee^{-aT^{\oplus k}}} \leq u(y) \E{\ee^{-aS^{\oplus k}}}. 
 \]
 But we already know $\ee^{-a_1 t_k} \cdot u_1(x) \E{\ee^{-a_1T^{\oplus k}}} = u_1(y) \E{\ee^{-a_1S^{\oplus k}}}$ and $a = a_1$, so after dividing out $e^{-at_k}$, $\E{\ee^{-aT^{\oplus k}}}$ and $\E{\ee^{-aS^{\oplus k}}}$ we obtain $\frac{u(y)}{u_1(y)} \geq \frac{u(x)}{u_1(x)}$. 
 
 Finally, since $x, y$ are arbitrary, we can switch them and use the same argument to deduce the opposite inequality $\frac{u(x)}{u_1(x)} \geq \frac{u(y)}{u_1(y)}$. This proves that $\frac{u(y)}{u_1(y)} = \frac{u(x)}{u_1(x)}$ for any pair of rewards $x, y$. Hence the social utility representation is a constant multiple of agent 1's. 

\subsection{Proof of Proposition~\ref{prop:pareto-possibility}}

We prove that the proposed representation for the social preference relation $\succeq$ satisfies the Pareto axiom. If $(x, T) \succeq_i (y, S)$ for every $i$, then $u_i(x) \ee^{-r_i \Phi_i(T)} \geq u_i(y) \ee^{-r_i \Phi_i(S)}$, which can be rewritten as
\[
    r_i(\Phi_i(S) - \Phi_i(T)) \geq \log\frac{u_i(y)}{u_i(x)}. 
\]
Summing across $i$ using the weights $\lambda_i$ we obtain
\[
\sum_{i = 1}^{n} \lambda_ir_i(\Phi_i(S) - \Phi_i(T)) \geq \sum_{i=1}^{n} \lambda_i \log\frac{u_i(y)}{u_i(x)}  = \log\frac{u(y)}{u(x)},
\]
where the last equality uses $u = \Pi_{i=1}^{n} u_i^{\lambda_i}$. Since $r\Phi = \sum_{i=1}^n \lambda_i r_i\Phi_i$, it follows that $r(\Phi(S) - \Phi(T)) \geq \log\frac{u(y)}{u(x)}$, which is equivalent to $u(x) \ee^{-r \Phi(T)} \geq u(y) \ee^{-r \Phi(S)}$. Thus $(x, T) \succeq (y, S)$ as desired.

\subsection{Proof of Proposition~\ref{prop:pareto-characterization}}

We assume the Pareto axiom holds and deduce its implications. Note that if $\Phi_i(T) \leq \Phi_i(S)$ for every $i$, then $(1,T) \succeq_i (1,S)$ for every $i$ and thus, by the Pareto axiom, $(1, T) \succeq (1,S)$ and $\Phi(T) \leq \Phi(S)$ also hold. 

We say that a collection of monotone additive statistics $(\Phi_1,\ldots,\Phi_n,\Phi)$ have the \emph{Pareto property} if $\Phi_i(T) \leq \Phi_i(S)$ for every $i$ implies $\Phi(T) \leq \Phi(S)$. We have the following result: 
%Turning to the ``only if'' direction, we first study a weaker version of the Pareto axiom which applies to a fixed reward $x$: for now we only require that $K_{-r_i}(T) \leq K_{-r_i}(S)$ for every $i$ implies $\Phi(T) \leq \Phi(S)$, where $\Phi$ is the certainty equivalent for the social preference $\succeq$. From this we can show that $\Phi$ must be an average of $K_{-r_i}$. In fact, we have a stronger result as follows: 

\begin{lemma}\label{lemma:aggregation}
Let $(\Phi_1,\ldots,\Phi_n,\Phi)$ be monotone additive statistics defined on $L^\infty_+$, and suppose that they satisfy the Pareto property. Then there exists a probability vector $(\beta_1, \dots, \beta_n)$ such that $\Phi = \sum_{i = 1}^n \beta_i \Phi_i$. 
\end{lemma}
\begin{proof}
Let $(\mu_1,\ldots,\mu_n,\mu)$ be the mixing measures on $\overline{\R}$ that correspond to the monotone additive statistics $(\Phi_1,\ldots,\Phi_n,\Phi)$. Define the linear functionals $(I_1,\ldots,I_n,I)$ on $\mathcal{C}(\overline{\R})$ as $I_i(f) = \int_{\overline{\R}} f \dd \mu_i $ and $I(f) = \int_{\overline{\R}} f \dd \mu$.  

We call a set of functions $\mathcal{D} \subseteq \mathcal{C}(\overline{\R})$ a \textit{Pareto domain} if for every $f,g \in \mathcal{D}$,
\[
    I_i(f) \geq I_i(g)~~ i=1,\ldots,n \implies I(f) \geq I(g).
\]
The Pareto property implies $\mathcal{L}_+ = \{K_X : X \in L^\infty_+\}$ is a Pareto domain. 
Define, as in the proof of Theorem~\ref{thm:main}, $\mathcal{L} = \{K_X : X \in L^\infty\}$ as well as the rational cone spanned by $\mathcal{L}$:
$$
    \text{cone}_{\mathbb{Q}}(\mathcal{L}) = \{q L : q \in \mathbb{Q}_+, L \in \mathcal{L} \} = \bigcup_{n=1}^\infty \frac{1}{n} \mathcal{L}
$$

We show that $\mathcal{L}$ and $\textup{cone}_{\mathbb{Q}}(\mathcal{L})$ are both Pareto domains. Given $X,Y \in L^\infty$, let $c$ be a large positive constant such that $X + c \geq 0$ and $Y + c \geq 0$. If $I_i(K_X) \geq I_i(K_Y)$ for all $i$ then $I_i(K_X + c) \geq I_i(K_Y + c)$ for all $i$ since each $I_i$ is linear. Thus, by the Pareto property  and the linearity of $I$, $I(K_X + c) \geq I(K_Y + c)$ and $I(K_X) \geq I(K_Y)$. This shows $\mathcal{L}$ is a Pareto domain. As for $\textup{cone}_{\mathbb{Q}}(\mathcal{L})$, observe that $I_i(\frac{1}{m} K_X) \geq I_i(\frac{1}{n}K_Y)$ for all $i$ is equivalent to $I_i(nK_X) \geq I_i(mK_Y)$ for all $i$, which implies $I(nK_X) \geq I(mK_Y)$ since $\mathcal{L}$ is a Pareto domain and is closed under addition. This shows $I(\frac{1}{m} K_X) \geq I(\frac{1}{n}K_Y)$ as desired.

Next we show that the closure of $\textup{cone}_{\mathbb{Q}}(\mathcal{L})$ (with respect to the usual sup norm) is also a Pareto domain. Let $f, g$ be in the closure, such that $I_i(f) \geq I_i(g)$ for all $i$. Pick sequences $(f_k)$ and $(g_k)$ in $\textup{cone}_{\mathbb{Q}}(\mathcal{L})$ converging to $f$ and $g$. Define $\eps_{i,k} = \vert I_i(f) - I_i(f_k) \vert + \vert I_i(g) - I_i(g_k) \vert$ and $\eps_k = \max_{1 \leq i \leq n} \varepsilon_{i,k}$. Then from $I_i(f) \geq I_i(g)$ we deduce $I_i(f_k) \geq I_i(g_k) - \eps_k = I_i(g_k - \eps_k)$ for every $i$. Note that $g_k - \eps_k$ belongs to $\textup{cone}_{\mathbb{Q}}(\mathcal{L})$ since the latter contains all the constant functions and is closed under addition. Thus by the fact that $\textup{cone}_{\mathbb{Q}}(\mathcal{L})$ is a Pareto domain, $I_i(f_n) \geq I_i(g_n - \eps_n)$ for every $i$ implies $I(f_k) \geq I(g_k - \eps_k) = I(g_k) - \eps_k$ for every $k$. Continuity of the functionals $(I_i)$ yields $\eps_k \to 0$. Continuity of $I$ thus yields $I(f) \geq I(g)$.

This proves that the closure of $\textup{cone}_{\mathbb{Q}}(\mathcal{L})$ is a Pareto domain. Since the subset of a Pareto domain is a Pareto domain, we conclude that $\textup{cone}(\mathcal{L})$ (i.e.\ the cone generated by $\mathcal{L}$) is a Pareto domain as well.

Now define $\mathcal{V} = \textup{cone}(\mathcal{L}) - \textup{cone}(\mathcal{L})$ to be the vector space generated by the cone. It is immediate to verify, using the linearity of the integral, that $\mathcal{V}$ is a Pareto domain as well. In particular, for any $f \in \mathcal{V}$, $I_i(f) \leq 0$ for every $i$ implies $I(f) \leq 0$. Corollary 5.95 in \cite{guide2006infinite} thus implies there exist non-negative scalars $\beta_1,\ldots,\beta_n$ such that $I = \sum_{i=1}^n \beta_i I_i$ on $\mathcal{V}$. So $I(K_X) = \sum_{i=1}^n \beta_i I_i(K_X)$ for every $X \in L^{\infty}$, which implies $\Phi(X) = \sum_{i=1}^n \beta_i \Phi_i(X)$. For constant $X$ this implies $\sum_i \beta_i = 1$, proving the lemma.
\end{proof}

Thus, the Pareto axiom implies that the social certainty equivalent $\Phi$ must be a convex combination of the individual $\Phi_i$. To complete the proof, we restrict to the case of identical utility functions $u_i = u$ which additionally satisfies $\lim_{x \to 0}u(x) = 0$ or $\lim_{x \to \infty}u(x) = \infty$. In this case, in order for $u = \Pi_{i=1}^{n} u_i^{\lambda_i}$ to hold, the weights $\lambda_1, \dots, \lambda_n$ must sum to $1$. Therefore the desired identity $r\Phi = \sum_{i=1}^{n} \lambda_i r_i \Phi_i$ requires us to show that not only $\Phi$ is a convex combination of $(\Phi_i)$, but $r\Phi$ is also a convex combination of $(r_i\Phi_i)$.

To prove this, we make use of the Pareto axiom when applied to time lotteries with different rewards. For any $S, T \in L^{\infty}_{+}$, the Pareto axiom says that if rewards $x,y$ are such that $r_i\Phi_i(S) - r_i \Phi_i(T) \geq \log\left(u(y)/u(x)\right)$ for all $i$, then $r\Phi(S) - r\Phi(T) \geq \log\left(u(y)/u(x)\right)$ also holds. By the richness assumption on $u$, we can choose $x, y$ with
\[
\log\left(u(y)/u(x)\right) = \min_{1 \leq i \leq n} \{r_i\Phi_i(S) - r_i \Phi_i(T)\}.
\]
Therefore the Pareto axiom implies that for any $S, T\in L^{\infty}_{+}$,
\begin{equation}\label{eq:pareto-rewrite}
r\Phi(S) - r\Phi(T) \geq \min_{1 \leq i \leq n} \{r_i\Phi_i(S) - r_i \Phi_i(T)\}. 
\end{equation}

The conclusion that $r\Phi$ is a convex combination of $(r_i\Phi_i)$ will follow from the condition \eqref{eq:pareto-rewrite} via an application of Farkas' Lemma. To rewrite this condition in linear algebra form, we let $m \leq n$ be the largest number of different $\Phi_i$ that are linearly independent (when viewed as functions on $L^{\infty}_{+}$). Reordering if necessary, we can assume $\Phi_1, \dots, \Phi_m$ are linearly independent, and every $\Phi_i$ is a (not necessarily positive) linear combination of those $m$. Thus we can find vectors $\gamma^1, \dots, \gamma^n \in \R^m$ such that every $r_i\Phi_i$ can be rewritten as the following inner product (i.e., linear combination):
\[
r_i \Phi_i = \gamma^i \cdot (\Phi_1, \dots, \Phi_m).
\]
Since $\Phi$ is a convex combination of $(\Phi_i)$, there also exists $\gamma \in \R^m$ such that $r \Phi = \gamma \cdot (\Phi_1, \dots, \Phi_m)$. 

Consider the following set of vectors: 
\[
\mathcal{W} = \{w \in \R^m: ~\gamma \cdot w \geq \min_{1 \leq i \leq n} \gamma^i \cdot w \}. 
\]
Let $\mathcal{D}$ be all vectors of the form $(\Phi_1(S)- \Phi_1(T), \dots, \Phi_m(S)-\Phi_m(T))$ for some $S, T \in L^{\infty}_{+}$. Condition \eqref{eq:pareto-rewrite} says that $\mathcal{D} \subseteq \mathcal{W}$. Note that $-\mathcal{D} = \mathcal{D}$, and $\mathcal{D}$ is closed under addition because every $\Phi_i$ is additive. Moreover, since the definition of $\mathcal{W}$ involve homogeneous inequalities,
$\frac{1}{N} \mathcal{D} \subseteq \mathcal{W}$ for every positive integer $N$. From these properties we deduce that any vector of the form $q_1 w_1 + \dots + q_k w_k$ with $q_j \in \Q$ and $w_j \in \mathcal{D}$ belongs to $\mathcal{W}$, because it can be written as $\frac{1}{N} w$ for some positive integer $N$ and $w \in \mathcal{D}$. Since $\mathcal{W}$ is a closed set, the span of $\mathcal{D}$ (not just the rational span) is also contained in $\mathcal{W}$. Finally note that $\mathcal{D}$ spans the entirety of $\R^m$. This is because by setting $T = 0$, $\mathcal{D}$ in particular includes vectors of the form $(\Phi_1(S), \dots, \Phi_m(S))$, and such vectors cannot all belong to a lower-dimensional subspace by the assumption that $\Phi_1, \dots, \Phi_m$ are linearly independent. 

Therefore, $\mathcal{D} = \mathcal{W} = \R^m$, which implies 
\begin{equation}\label{eq:pareto-gamma}
\gamma \cdot w \geq \min_{1 \leq i \leq n} \gamma^i \cdot w \,\,\, \text{for all} \,\, w \in \R^m. 
\end{equation}
For any $\varepsilon > 0$, this condition implies that there exists no $w \in \R^m$ such that $-\gamma^i \cdot w \leq -1-\varepsilon$ for every $i$ while $\gamma \cdot w \leq 1$. Let $A$ be an $(n+1) \times m$ matrix whose first $n$ rows are $-\gamma^1, \dots, -\gamma^n$, and whose last row is $\gamma$. Let $b$ be the $n+1$-dimensional vector $(-1-\varepsilon, \dots, -1-\varepsilon, 1)$. Then $A w \leq b$ has no solution $w \in \R^m$. 

By Farkas' Lemma, there exists a non-negative $n+1$-dimensional vector $z = (z_1, \dots, z_{n+1})$ such that $z' A = 0$ while $z \cdot b < 0$. The former implies $z_{n+1} \gamma = z_1 \gamma^1 + \cdots + z_n \gamma^n$, while the latter implies $z_{n+1} < (1+\varepsilon)(z_1+\cdots + z_n)$. Note that $z_{n+1}$ cannot be zero, for otherwise we have a positive linear combination of $\gamma^1,\dots,\gamma^n$ that gives the zero vector, leading to the impossible implication that a positive linear combination of $\Phi_1, \dots, \Phi_n$ equals zero. 

Thus we can write $\gamma = \alpha_1\gamma^1 + \cdots + \alpha_n\gamma^n$, 
with non-negative weights $\alpha_i = \frac{z_i}{z_{n+1}}$ whose sum is greater than $\frac{1}{1+\varepsilon}$. Consequently $r\Phi = \sum_{i = 1}^{n} \alpha_i r_i \Phi_i$, which implies $r = \sum_{i=1}^{n} \alpha_i r_i$ and thus $\alpha_i \leq \frac{r}{r_i}$ in any such representation. Since $\varepsilon$ is arbitrary, a compactness argument then yields that $\gamma = \sum_{i = 1}^{n} \alpha_i \gamma^i$ for some non-negative weights $\alpha_i$ with $\sum_{i=1}^{n}\alpha_i \geq 1$. 

We can also choose $\hat{b} = (1- \varepsilon, \dots, 1 - \varepsilon, -1)$ and deduce from \eqref{eq:pareto-gamma} that $A w \leq \hat{b}$ has no solution $w \in \R^m$. Then a similar analysis yields $\gamma = \hat{\alpha}_1\gamma^1 + \cdots + \hat{\alpha}_n\gamma^n$ for some weights $\hat{\alpha}_i \geq 0$ and $\sum_{i=1}^{n} \alpha_i < \frac{1}{1-\epsilon}$. Again by compactness, we can assume $\sum_{i=1}^{n}\hat{\alpha}_i \leq 1$. Finally, by suitably averaging between $\alpha_i$ and $\hat{\alpha}_i$, we can find non-negative weights $(\lambda_i)$ whose sum is equal to $1$, such that $\gamma = \sum_{i=1}^{n} \lambda_i \gamma^{i}$. So $r\Phi = \sum_{i=1}^{n} \lambda_i r_i \Phi_i$. Since $\Phi$ is also a convex combination of $(\Phi_i)$, it follows that $r = \sum_i\lambda_i r_i$, completing the proof. 

\section{Proof of Theorem~\ref{thm:betweenness}}

Since the preference $\succeq$ is represented by $\Phi$, the betweenness axiom is equivalent to the following: 
\begin{align*}
    \Phi(X) = \Phi(Y) ~~\text{if and only if}~~ \Phi(X_{\lambda}Y) = \Phi(Y).
\end{align*}
In this case, we say that the statistic $\Phi$ satisfies betweenness. We need to show that $\Phi(X)$ satisfies betweenness if and only if it is equal to $K_a(X)$ for some $a \in \R$ or equal to $\beta K_{-a\beta}(X) + (1-\beta) K_{a(1-\beta)}(X)$ for some $\beta \in (0,1)$ and $a \in (0, \infty)$. 

We first show the ``if'' direction. Specifically, when $\Phi(X) = K_a(X)$ for some $a \in \R$, then the preference is CARA expected utility, which satisfies independence and thus betweenness. When $\Phi(X) = \beta K_{-a\beta}(X) + (1-\beta) K_{a(1-\beta)}(X)$, we can use the definition of $K$ to rewrite it as
\[
\Phi(X) = \frac{1}{a}\left(\log\mathbb{E}[\ee^{a(1-\beta)X}] - \log \mathbb{E}[\ee^{-a\beta X}]\right).
\]
Thus $\Phi(X) = \Phi(Y)$ if and only if $\log  \E{\ee^{a(1-\beta)X}} - \log\E{\ee^{-a\beta X}} = \log  \E{\ee^{a(1-\beta)Y}} - \log\E{\ee^{-a\beta Y}}$, which in turn is equivalent to
\[
\frac{\E{\ee^{a(1-\beta)X}}}{\E{\ee^{a(1-\beta)Y}}} = \frac{\E{\ee^{-a\beta X}}}{\E{\ee^{-a \beta Y}}}. 
\]
Since $\E{\ee^{bX_{\lambda}Y}} = \lambda  \E{\ee^{bX}} + (1-\lambda) \E{\ee^{bY}}$ for every $b \in \R$, it is not difficult to see that the above ratio equality holds if and only if it holds when $X$ is replaced by  $X_{\lambda}Y$. Hence $\Phi(X) = \Phi(Y)$ if and only if $\Phi(X_{\lambda}Y) = \Phi(Y)$, i.e.\ betweenness is satisfied. 

\medskip

Turning to the ``only if'' direction. We will characterize any monotone additive statistic $\Phi$ that satisfies a weaker form of betweenness:

\begin{lemma}\label{lemma:weak betweenness}
Suppose $\Phi$ is a monotone additive statistic such that $\Phi(X) = c$ implies $\Phi(X_{\lambda}c) = c$ whenever $c$ is a constant. Then either $\Phi$ takes the form described by Theorem~\ref{thm:betweenness}, or $\Phi(X) = \beta \min[X] + (1-\beta) \max[X]$ for some $\beta \in [0,1]$. 
\end{lemma}
This result implies Theorem~\ref{thm:betweenness} because $\Phi(X) = \beta \min[X] + (1-\beta) \max[X]$ violates the original betweenness axiom. To see that, let $X = 0$ and choose any $Y$ supported on $\pm 1$. Then $X_{\lambda} Y$ and $Y$ have the same minimum and maximum, so that $\Phi(X_{\lambda}Y) = \Phi(Y)$.  But $\Phi(X) = \Phi(Y)$ cannot hold for \emph{all} $Y$ supported on $\pm 1$. 

The proof of Lemma \ref{lemma:weak betweenness} is in turn based on the following lemma which further relaxes betweenness: 

\begin{lemma}\label{lemma:quasiconvex}
Suppose $\Phi(X) = \int_{\R} K_a(X) \, \dd \mu(a)$ has the property that $\Phi(X) = c$ implies $\Phi(X_{\lambda}c) \leq c$. Then the measure $\mu$ restricted to $[0, \infty]$ is either the zero measure, or it is supported on a single point. 
\end{lemma}

\begin{proof}
It suffices to show that if $\mu$ puts positive mass on $(0, \infty]$, then that mass is supported on a single point and $\mu(\{0\}) = 0$. For this let $N > 0$ denote the essential maximum of the support of $\mu$; that is, $N = \min \{x:~ \mu((x, \infty]) = 0\}$. We allow $N = \infty$ when the support of $\mu$ is unbounded from above, or when $\mu$ has a non-zero mass at $\infty$. For any positive real number $b < N$, consider the same random variable $X_{n,b}$ as in the proof of Lemma~\ref{lemma:unique}, given by
\begin{align*}
    \Pr{X_{n,b}=n} &= \ee^{-bn}\\
    \Pr{X_{n,b}=0} &= 1-\ee^{-bn}.
\end{align*}
As shown in the proof of Lemma~\ref{lemma:unique}, $\frac{1}{n} K_a(X_{n,b})$ is uniformly bounded in $[0,1]$, and
\[
\lim_{n \to \infty} \frac{1}{n} K_a(X_{n,b}) = \frac{(a-b)^{+}}{a}.
\]
Thus if we let $c_n = \Phi(X_{n,b})$, then by the Dominated Convergence Theorem, 
\[
\lim_{n \to \infty} \frac{c_n}{n} = \lim_{n \to \infty} \frac{1}{n} \Phi(X_{n,b}) = \lim_{n \to \infty} \int_{\overline{\R}} \frac{1}{n} K_a(X_{n,b}) \, \dd \mu(a) = \int_{(b, \infty]} \frac{a-b}{a} \, \dd \mu(a).
\]
Denote $\gamma = \int_{(b, \infty]} \frac{a-b}{a} \, \dd \mu(a)$. This number $\gamma$ is strictly positive because $b < N$ implies $\mu((b,\infty]) > 0$. We can also assume $\gamma < 1$, since otherwise $\mu$ must be the point mass at $\infty$. 

Now, as $\Phi(X_{n,b}) = c_n$ we know by assumption that $\Phi(Y_{n,b}) \leq c_n$ for each $n$, where $Y_{n,b}$ is the mixture between $X_{n,b}$ and the constant $c_n$ (in what follows $\lambda$ is fixed as $n$ varies):
\begin{align*}
    \Pr{Y_{n,b}=n} &= \lambda\ee^{-bn}\\
    \Pr{Y_{n,b}=0} &= \lambda(1-\ee^{-bn}) \\
    \Pr{Y_{n,b}=c_n} &= 1-\lambda. 
\end{align*}
Using $\lim_{n \to \infty} c_n/n = \gamma$, we have
\begin{align*}
    \lim_{n \to \infty} \frac{1}{n}K_a(Y_{n,b}) 
    &= \lim_{n \to \infty} \frac{1}{n}\frac{1}{a}\log\left[\lambda\left(1 -  \ee^{-bn} + \ee^{(a-b)n}\right) + (1-\lambda) \ee^{a \cdot c_n}\right] \nonumber \\
    &=
    \begin{cases}
    0 & \text{if } a < 0\\
    (1-\lambda)\gamma & \text{if } a = 0 \\
    \gamma & \text{if } 0 < a < \frac{b}{1-\gamma} \\
    \frac{a-b}{a} & \text{if } a \geq \frac{b}{1-\gamma}.
    \end{cases}
\end{align*}
Note that the cutoff point $a = \frac{b}{1-\gamma}$ is where $a-b = a\gamma$. When $a$ is smaller than this, the dominant term in the bracketed sum above is $(1-\lambda) \ee^{a \cdot c_n}$. Whereas for larger $a$, the dominant term becomes $\lambda \ee^{(a-b) \cdot n}$.

Crucially, $\lim_{n \to \infty} \frac{1}{n}K_a(Y_{n,b}) \geq \frac{(a-b)^{+}}{a}$ holds for every $a$, with strict inequality for $a \in [0, \frac{b}{1-\gamma})$. Thus again by the Dominated Convergence Theorem, 
\[
    \lim_{n \to \infty} \frac{c_n}{n} \geq \lim_{n \to \infty} \frac{1}{n} \Phi(Y_{n,b}) = \lim_{n \to \infty} \int_{\overline{\R}} \frac{1}{n} K_a(Y_{n,b}) \, \dd \mu(a) \geq \int_{(b, \infty]} \frac{a-b}{a} \, \dd \mu(a).
\]
But we know that the far left is equal to the far right. So both inequalities hold equal, and in particular $\lim_{n \to \infty} \frac{1}{n} K_a(Y_{n,b}) = \frac{(a-b)^{+}}{a}$ holds $\mu$-almost surely. 

As discussed, $\lim_{n \to \infty} \frac{1}{n} K_a(Y_{n,b}) > \frac{(a-b)^{+}}{a}$ for any $a \in [0, \frac{b}{1-\gamma})$. So we can conclude that $\mu([0, \frac{b}{1-\gamma})) = 0$. This must hold for any $b \in (0, N)$ and corresponding $\gamma$. Letting $b$ arbitrarily close to $N$ thus yields $\mu([0, N)) = 0$ (since $\frac{b}{1-\gamma} > b$). It follows that when restricted to $[0, \infty]$ the measure $\mu$ is concentrated at the single point $N$, as we desire to show. 
\end{proof}

\begin{proof}[Proof of Lemma~\ref{lemma:weak betweenness}]
From Lemma~\ref{lemma:quasiconvex}, we know that the measure $\mu$ associated with $\Phi$ can only be supported on one point in all of $[0,\infty]$. By a symmetric argument, $\mu$ also has at most one point support in all of $[-\infty, 0]$. Thus either $\mu = \delta_{a}$ for some $a \in \overline{\R}$, or $\mu$ is supported on two points $\{a_1, a_2\}$ with $a_1 < 0 < a_2$. In the former case we are done, so below we study the latter case where $\mu$ has two-point support.

Suppose $\Phi(X) = \beta K_{a_1}(X) + (1-\beta) K_{a_2}(X)$ for some $\beta \in (0,1)$ and $a_1 < 0 < a_2$. If $a_1 = -\infty$ while $a_2 < \infty$, then $\Phi(X) = \beta \min[X] + (1-\beta) K_{a_2}(X)$. Take any non-constant $X$ and let $c$ denote $\Phi(X)$. Note that since $K_{a_2}(X) > \min[X]$, $c = \beta \min[X] + (1-\beta) K_{a_2}(X)$ lies strictly between $\min[X]$ and $K_{a_2}(X)$. Consider the mixture $X_{\lambda}c$, then $\min[X_{\lambda}c] = \min[X]$, whereas
\[
K_{a_2}(X_{\lambda}c) = \frac{1}{a_2} \log \left(\lambda \E{\ee^{a_2X}} + (1-\lambda) \ee^{a_2c}\right) < \frac{1}{a_2} \log \E{\ee^{a_2X}} = K_{a_2}(X),
\]
where the inequality uses $c < K_{a_2}(X) = \frac{1}{a_2} \log \E{\ee^{a_2X}}$ and $a_2 > 0$. 
We thus deduce that 
\[
\Phi(X_{\lambda}c) = \beta \min[X_{\lambda}c] + (1-\beta) K_{a_2}(X_{\lambda}c) < \beta \min[X] + (1-\beta) K_{a_2}(X) = c,
\]
contradicting the betweenness axiom. A symmetric argument rules out the possibility that $a_1 > -\infty$ while $a_2 = \infty$. 

Hence, either $a_1 = -\infty$ and $a_2 = \infty$, or $a_1 \in (-\infty, 0)$ and $a_2 \in (0, \infty)$. In the former case $\Phi(X)$ is an average of the minimum and the maximum, so we are again done. It remains to consider the latter case where $a_1, a_2$ are both finite. In this case we will show that $\beta = \frac{-a_1}{a_2-a_1}$. Once this is shown, we can let $a = a_2-a_1$ so that $a_1 = -a \beta$ and $a_2 = a(1-\beta)$. Thus $\Phi(X) = \beta K_{-a \beta}(X) + (1-\beta) K_{a(1-\beta)}(X)$ as desired. 

Let us take an arbitrary non-constant $X$, and let 
\[
    c = \Phi(X) = \frac{\beta}{a_1} \log \E{\ee^{a_1X}} + \frac{1-\beta}{a_2} \log \E{\ee^{a_2X}}.
\]
For an arbitrary $\lambda \in [0,1]$, we must also have
\begin{equation}\label{eq:betweenness1}
    c = \Phi(X_{\lambda}c) = \frac{\beta}{a_1} \log \E{\lambda\ee^{a_1X}+(1-\lambda)\ee^{a_1c}} + \frac{1-\beta}{a_2} \log \E{\lambda\ee^{a_2X}+(1-\lambda)\ee^{a_2c}}.
\end{equation}
Since \eqref{eq:betweenness1} holds for every $\lambda$, we can differentiate it with respect to $\lambda$ to obtain
\[
0 = \frac{\beta(\E{\ee^{a_1X}} - \ee^{a_1c})}{a_1 \E{\lambda\ee^{a_1X}+(1-\lambda)\ee^{a_1c}}} +  \frac{(1-\beta)(\E{\ee^{a_2X}} - \ee^{a_2c})}{a_2 \E{\lambda\ee^{a_2X}+(1-\lambda)\ee^{a_2c}}}.
\]
Plugging in $\lambda = 0$ and $\lambda = 1$ gives, respectively,
\begin{equation}\label{eq:betweenness2}
\frac{\beta(\E{\ee^{a_1X}} - \ee^{a_1c})}{a_1\ee^{a_1c}} = - \frac{(1-\beta)(\E{\ee^{a_2X}} - \ee^{a_2c})}{a_2\ee^{a_2c}}.
\end{equation}
\begin{equation}\label{eq:betweenness3}
\frac{\beta(\E{\ee^{a_1X}} - \ee^{a_1c})}{a_1\E{\ee^{a_1X}}} = -  \frac{(1-\beta)(\E{\ee^{a_2X}} - \ee^{a_2c})}{a_2\E{\ee^{a_2X}}}.
\end{equation}
Since $c = \beta K_{a_1}(X) + (1-\beta) K_{a_2}(X)$, the fact that $K_{a_2}(X) > K_{a_1}(X)$ implies $c$ is strictly between $K_{a_1}(X)$ and $K_{a_2}(X)$. Thus, using $a_1 < 0 < a_2$ we deduce $\ee^{a_1c} < \E{\ee^{a_1X}}$ and $\ee^{a_2c} < \E{\ee^{a_2X}}$. 

We can therefore divide \eqref{eq:betweenness2} by \eqref{eq:betweenness3} to obtain
\[
\frac{\E{\ee^{a_1X}}}{\ee^{a_1c}} = \frac{\E{\ee^{a_2X}}}{\ee^{a_2c}}.
\]
Plugging this back to \eqref{eq:betweenness2}, we conclude $\frac{\beta}{a_1} = - \frac{1-\beta}{a_2}$, so $\beta = \frac{-a_1}{a_2-a_1}$ as we desire to show. 
\end{proof}

\newpage
\begin{center}
    {\Large \textbf{Online Appendix}}
\end{center}

\section{Proof of Theorem~\ref{thm:domains}}

The proof is considerably more complex than the proof of Theorem \ref{thm:main}, so we break it into several steps below. 

\subsection{Step 1: Catalytic Order on $L_M$}

We first establish a generalization of Theorem~\ref{thm:marginal} to unbounded random variables. For two random variables $X$ and $Y$ with c.d.f.\ $F$ and $G$ respectively, we say that \emph{$X$ dominates $Y$ in both tails} if there exists a positive number $N$ with the property that 
\[
G(x) > F(x) \quad \text{ for all } \vert x \vert \geq N. 
\]
In particular, $X$ needs to be unbounded from above, and $Y$ unbounded from below. 

\begin{lemma}\label{lemma:marginal-unbounded}
Suppose $X, Y \in L_M$ satisfy $K_a(X) > K_a(Y)$ for every $a \in \R$. Suppose further that $X$ dominates $Y$ in both tails. Then there exists an independent random variable $Z \in L_M$ such that $X + Z \geq_1 Y + Z$. 
\end{lemma}

\begin{proof}
We will take $Z$ to have a normal distribution, which does belong to $L_M$. Following the proof of Theorem~\ref{thm:marginal}, we let $\sigma(x) = G(x) - F(x)$, and seek to show that 
$[\sigma * h](y) \geq 0$ for every $y$ when $h$ is a Gaussian density with sufficiently large variance. By assumption, $\sigma(x)$ is strictly positive for $\vert x \vert \geq N$. Thus there exists $\delta > 0$ such that $\int_{N+1}^{N+2} \sigma(x) \,\dd x > \delta$, as well as $\int_{-N-2}^{-N-1} \sigma(x) \,\dd x > \delta$. We fix $A > 0$ that satisfies $\ee^{A} \geq \frac{4N}{\delta}$.

Similar to \eqref{eq:convo}, we have for $h(x) = \ee^{-\frac{x^2}{2V}}$ that
\begin{equation}\label{eq:convo2}
\ee^{\frac{y^2}{2V}}\int \sigma(x) h(y-x) \,\dd x =  \int_{-\infty}^{\infty} \sigma(x) \cdot
    \ee^{\frac{y}{V} \cdot x} \cdot \ee^{-\frac{x^2}{2V}} \,\dd x.
\end{equation}
The variance $V$ is to be determined below. 

We first show that the right-hand side is positive if $V \geq (N+2)^2$ and $\frac{y}{V} \geq A$. Indeed, since $\sigma(x) > 0$ for $\vert x \vert \geq N$, this integral is bounded from below by 
\begin{align*}
&\int_{-N}^{N} \sigma(x) \cdot \ee^{\frac{y}{V} \cdot x} \cdot \ee^{-\frac{x^2}{2V}} \,\dd x + \int_{N+1}^{N+2} \sigma(x) \cdot \ee^{\frac{y}{V} \cdot x} \cdot \ee^{-\frac{x^2}{2V}} \,\dd x \\
\geq ~& -2N \cdot \ee^{\frac{y}{V}\cdot N} + \delta \cdot \ee^{\frac{y}{V} \cdot (N+1)} \cdot \ee^{-\frac{(N+2)^2}{2V}} \\
= ~& \ee^{\frac{y}{V}\cdot N} \cdot (-2N + \delta \cdot \ee^{\frac{y}{V}} \cdot \ee^{-\frac{(N+2)^2}{2V}}) \\
> ~&0,
\end{align*}
where the last inequality uses $\ee^{\frac{y}{V}} \geq \ee^A \geq \frac{4N}{\delta}$ and $\ee^{-\frac{(N+2)^2}{2V}} \geq \ee^{-\frac{1}{2}} > \frac{1}{2}$. By a symmetric argument, we can show that the right-hand side of \eqref{eq:convo2} is also positive when $\frac{y}{V} \leq -A$. 

It remains to consider the case where $\frac{y}{V} \in [-A, A]$. Here we rewrite the integral on the right-hand side of \eqref{eq:convo2} as 
\begin{align*}
\int_{-\infty}^{\infty} \sigma(x) \cdot
    \ee^{\frac{y}{V} \cdot x} \cdot \ee^{-\frac{x^2}{2V}}\,\dd x =  M_\sigma(\frac{y}{V}) - \int_{-\infty}^{\infty} \sigma(x) \cdot
    \ee^{\frac{y}{V} \cdot x} \cdot (1-\ee^{-\frac{x^2}{2V}})\,\dd x, 
\end{align*}
where $M_{\sigma}(a) = \int_{-\infty}^{\infty} \sigma(x) \cdot \ee^{ax} \,\dd x = \frac{1}{a}\E{\ee^{aX}} - \frac{1}{a}\E{\ee^{aY}}$ is by assumption strictly positive for all $a$. By continuity, there exists some $\eps > 0$ such that $M_{\sigma(a)} > \eps$ for all $\vert a \vert \leq A$. So it only remains to show that when $V$ is sufficiently large, 
\begin{equation}\label{eq:convo3}
\int_{-\infty}^{\infty} \sigma(x) \cdot
    \ee^{a x} \cdot (1-\ee^{-\frac{x^2}{2V}})\,\dd x < \eps \quad \text{ for all } \vert a \vert \leq A.
\end{equation}
To estimate this integral, note that $M_{\sigma}(A) = \int_{-\infty}^{\infty} \sigma(x) \cdot \ee^{Ax} \,\dd x$ is finite. Since $\sigma(x) > 0$ for $\vert x \vert$ sufficiently large, we deduce from the Monotone Convergence Theorem that $\int_{-\infty}^{T} \sigma(x) \cdot \ee^{Ax} \,\dd x$ converges to $M_{\sigma}(A)$ as $T \to \infty$. In other words, $\int_{T}^{\infty} \sigma(x) \cdot \ee^{Ax} \,\dd x \to 0$. We can thus find a sufficiently large $T > N$ such that $
\int_{T}^{\infty} \sigma(x) \cdot \ee^{Ax} \,\dd x < \frac{\eps}{4}$, and likewise $\int_{-\infty}^{-T} \sigma(x) \cdot \ee^{-Ax} \,\dd x < \frac{\eps}{4}$. 

As $1-\ee^{-\frac{x^2}{2V}} \geq 0$ and $\ee^{ax} \leq \ee^{A \vert x \vert}$ when $\vert a \vert \leq A$, we deduce that
\[
\int_{\vert x \vert \geq T} \sigma(x) \cdot
    \ee^{a x} \cdot (1-\ee^{-\frac{x^2}{2V}})\,\dd x < \frac{\eps}{2} \quad \text{ for all } \vert a \vert \leq A.
\]
Moreover, for this fixed $T$, we have $\ee^{-\frac{T^2}{2V}} \to 1$ when $V$ is large, and thus
\[
\int_{\vert x \vert \leq T} \sigma(x) \cdot
    \ee^{a x} \cdot (1-\ee^{-\frac{x^2}{2V}})\,\dd x < 2 T \ee^{AT}(1-\ee^{-\frac{T^2}{2V}}) < \frac{\eps}{2} \quad \text{ for all } \vert a \vert \leq A.
\]
These estimates together imply that \eqref{eq:convo3} holds for sufficiently large $V$. This completes the proof.
\end{proof}

\subsection{Step 2: A Perturbation Argument}

With Lemma~\ref{lemma:marginal-unbounded}, we know that if $\Phi$ is a monotone additive statistic defined on $L_M$, then $K_a(X) \geq K_a(Y)$ for all $a \in \R$ implies $\Phi(X) \geq \Phi(Y)$ \emph{under the additional assumption that} $X$ dominates $Y$ in both tails (same proof as for Lemma~\ref{lemma:monotone}). Below we deduce the same result without this extra assumption. To make the argument simpler, assume $X$ and $Y$ are unbounded both from above and from below; otherwise, we can add to them an independent Gaussian random variable without changing either the assumption or the conclusion. In doing so, we can further assume $X$ and $Y$ admit probability density functions. 

We first construct a heavy right-tailed random variable as follows:

\begin{lemma}
For any $Y \in L_M$ that is unbounded from above and admits densities, there exists $Z \in L_M$ such that $Z \geq 0$ and $\frac{\mathbb{P}[Z > x]}{\mathbb{P}[Y > x]} \to \infty$ as $x \to \infty.$
\end{lemma}

\begin{proof}
For this result, it is without loss to assume $Y \geq 0$ because we can replace $Y$ by $\vert Y \vert$ and only strengthen the conclusion. Let $g(x)$ be the probability density function of $Y$. We consider a random variable $Z$ whose p.d.f.\ is given by $cx g(x)$ for all $x \geq 0$, where $c > 0$ is a normalizing constant to ensure $\int_{x \geq 0} cx g(x) \,\dd x = 1$. Since the likelihood ratio between $Z = x$ and $Y = x$ is $cx$, it is easy to see that the ratio of tail probabilities also diverges. Thus it only remains to check $Z \in L_M$. This is because 
\[
\E{\ee^{aZ}}= c\int_{x \geq 0} x g(x) \ee^{ax} \,\dd x,
\]
which is simply $c$ times the derivative of $\E{\ee^{aY}}$ with respect to $a$. It is well-known that the moment generating function is smooth whenever it is finite. So this derivative is finite, and $Z \in L_M$.  
\end{proof}

In the same way, we can construct heavy left-tailed distributions:
\begin{lemma}
For any $X \in L_M$ that is unbounded from below and admits densities, there exists $W \in L_M$, such that $W \leq 0$ and $\frac{\mathbb{P}[W \leq x]}{\mathbb{P}[X \leq x]} \to \infty$ as $x \to -\infty$.
\end{lemma}

With these technical lemmata, we now construct ``perturbed'' versions of any two random variables $X$ and $Y$ to achieve dominance in both tails. For any random variable $Z \in L_M$ and every $\eps > 0$, let $Z_{\eps}$ be the random variable that equals $Z$ with probability $\eps$, and $0$ with probability $1-\eps$. Note that $Z_\eps$ also belongs to $L_M$. 

\begin{lemma}
Given any two random variables $X, Y \in L_M$ that are unbounded on both sides and admit densities. Let $Z \geq 0$ and $W \leq 0$ be constructed from the above two lemmata. Then for every $\eps > 0$, $X + Z_\eps$ dominates $Y + W_\eps$ in both tails.
\end{lemma}

\begin{proof}
For the right tail, we need $\mathbb{P}[X + Z_{\eps} > x] > \mathbb{P}[Y + W_{\eps} > x]$ for all $x \geq N$. Note that $W_{\eps} \leq 0$, so $\mathbb{P}[Y + W_{\eps} > x] \leq \mathbb{P}[Y > x]$. On other hand, 
\[
\mathbb{P}[X + Z_{\eps} > x] \geq \mathbb{P}[X \geq 0] \cdot \mathbb{P}[Z_{\eps} > x] = \mathbb{P}[X \geq 0] \cdot \eps \cdot \mathbb{P}[Z > x]. 
\]
Since by assumption $X$ is unbounded from above, the term $\mathbb{P}[X \geq 0] \cdot \eps$ is a strictly positive constant that does not depend on $x$. Thus for sufficiently large $x$, we have 
\[
\mathbb{P}[X \geq 0] \cdot \eps \cdot \mathbb{P}[Z > x] > \mathbb{P}[Y > x]
\]
by the construction of $Z$. This gives dominance in the right tail. The left tail is similar. 
\end{proof}

\subsection{Step 3: Monotonicity w.r.t.\ $K_a$} 

The next result generalizes the key Lemma~\ref{lemma:monotone} to our current setting:

\begin{lemma}
  \label{lemma:monotone-unbounded}
  Let $\Phi \colon L_M \to \R$ be a monotone additive statistic. If $K_a(X) \geq K_a(Y)$ for all $a \in \R$ then $\Phi(X) \geq \Phi(Y)$.
\end{lemma}

\begin{proof}
As discussed, we can without loss assume $X, Y$ are unbounded on both sides, and admit densities. Let $Z$ and $W$ be constructed as above, then for each $\eps > 0$, $X + Z_{\eps}$ dominates $Y + W_{\eps}$ in both tails, and $K_a(X+Z_{\eps}) > K_a(X) \geq K_a(Y) > K_a(Y+W_{\eps})$ for every $a \in \R$, where the inequalities are strict as $Z, W$ are not identically zero. 

Thus the pair $X + Z_{\eps}$ and $Y + W_{\eps}$ satisfy the assumptions in Lemma~\ref{lemma:marginal-unbounded}. We can then find an independent random variable $V \in L_M$ (depending on $\eps$), such that 
\[
X + Z_{\eps} + V \geq_1 Y + W_{\eps} + V.
\]
Monotonicity and additivity of $\Phi$ then imply $\Phi(X) + \Phi(Z_{\eps}) \geq \Phi(Y) + \Phi(W_{\eps})$, after canceling out $\Phi(V)$. The desired result $\Phi(X) \geq \Phi(Y)$ follows from the lemma below, which shows that our perturbations only slightly affect the statistic value. 
\end{proof}

\begin{lemma}
For any $Z \in L_M$ with $Z \geq 0$, it holds that $\Phi(Z_{\eps}) \to 0$ as $\eps \to 0$. Similarly $\Phi(W_{\eps}) \to 0$ for any $W \in L_M$ with $W \leq 0$.
\end{lemma}

\begin{proof}
We focus on the case for $Z_{\eps}$. Suppose for contradiction that $\Phi(Z_{\eps})$ does not converge to zero. Note that as $\eps$ decreases, $Z_{\eps}$ decreases in first-order stochastic dominance. So $\Phi(Z_{\eps}) \geq 0$ also decreases, and non-convergence must imply there exists some $\delta > 0$ such that $\Phi(Z_{\eps}) > \delta$ for every $\eps > 0$. 
Let $\mu_{\eps}$ be image measure of $Z_{\eps}$. We now choose a sequence $\eps_n$ that decreases to zero very fast, and consider the measures
\[
\nu_n = \mu_{\eps_n}^{*n},
\]
which is the $n$-th convolution power of $\mu_{\eps_n}$. Thus the sum of $n$ i.i.d.\ copies of $Z_{\eps_n}$ is a random variable whose image measure is $\nu_n$. We denote this sum by $U_n$.

For each $n$ we choose $\eps_n$ sufficiently small to satisfy two properties: (i) $\eps_n \leq \frac{1}{n^2}$, and (ii) it holds that
\[
\E{ \ee^{n U_n} - 1 } \leq 2^{-n}.
\]
This latter inequality can be achieved because $\E{ \ee^{n U_n}} = \left(\E{ \ee^{n Z_{\eps_n}}}\right)^n$, and as $\eps_n \to 0$ we also have $\E{ \ee^{n Z_{\eps_n}}} = 1 - \eps_n + \eps_n \E{ \ee^{n Z}} \to 1$ since $Z \in L_M$.  

For these choices of $\eps_n$ and corresponding $U_n$, let $H_n(x)$ denote the c.d.f.\ of $U_n$, and define $H(x) = \inf_{n} H_n(x)$ for each $x \in \R$. Since $H_n(x) = 0$ for $x < 0$, the same is true for $H(x)$. Also note that each $H_n(x)$ is a non-decreasing and right-continuous function in $x$, and so is $H(x)$. 

We claim that $\lim_{x \to \infty} H(x) = 1$. Indeed, recall that $U_n$ is the $n$-fold sum of $Z_{\eps_n}$, which has mass $1 - \eps_n$ at zero. So $U_n$ has mass at least $(1-\eps_n)^n \geq (1-\frac{1}{n^2})^n \geq 1 - \frac{1}{n}$ at zero. In other words, $H_n(0) \geq 1 - \frac{1}{n}$. By considering the finitely many c.d.f.s $H_1(x), H_2(x), \dots, H_{n-1}(x)$, we can find $N$ such that $H_i(x) \geq 1 - \frac{1}{n}$ for every $i < n$ and $x \geq N$. Together with $H_i(x) \geq H_i(0) \geq 1 - \frac{1}{i} \geq 1 - \frac{1}{n}$ for $i \geq n$, we conclude that $H_i(x) \geq 1 - \frac{1}{n}$ whenever $x \geq N$, and so $H(x) \geq 1 - \frac{1}{n}$. Since $n$ is arbitrary, the claim follows. The fact that $H_n(x) \geq 1 - \frac{1}{n}$ also shows that in the definition $H(x) = \inf_{n} H_n(x)$, the ``inf'' is actually achieved as the minimum.

These properties of $H(x)$ imply that it is the c.d.f.\ of some non-negative random variable $U$. We next show $U \in L_M$, i.e., $\E{\ee^{aU}} < \infty$ for every $a \in \R$. Since $U \geq 0$, we only need to consider $a \geq 0$. To do this, we take advantage of the following identity based on integration by parts:
\[
\E{\ee^{aU_n}-1} = -\int_{x \geq 0} (\ee^{ax}-1) \,\dd (1-H_n(x)) = a \int_{x \geq 0} \ee^{ax} (1-H_n(x)) \,\dd x.
\]
Now recall that we chose $U_n$ so that $\E{ \ee^{n U_n} - 1 } \leq 2^{-n}$. So $\E{ \ee^{a U_n} - 1 } \leq 2^{-n}$ for every positive integer $n \geq a$. It follows that the sum $\sum_{n = 1}^{\infty} \E{\ee^{aU_n}-1}$ is finite for every $a \geq 0$. Using the above identity, we deduce that 
\[
a \int_{x \geq 0} \ee^{ax} \sum_{n = 1}^{\infty} (1-H_n(x)) \,\dd x < \infty,
\]
where we have switched the order of summation and integration by the Monotone Convergence Theorem. Since $H(x) = \min_{n} H_n(x)$, it holds that $1-H(x) \leq \sum_{n = 1}^{\infty} (1-H_n(x))$ for every $x$. And thus 
\[
\E{\ee^{aU}-1} = a \int_{x \geq 0} \ee^{ax} (1-H(x)) \,\dd x < \infty
\]
also holds. This proves $U \in L_M$.

We are finally in a position to deduce a contradiction. Since by construction the c.d.f.\ of $U$ is no larger than the c.d.f.\ of each $U_n$, we have $U \geq_1 U_n$ and $\Phi(U) \geq \Phi(U_n)$ by monotonicity of $\Phi$. But $\Phi(U_n) = n \Phi(Z_{\eps_n}) > n \delta$ by additivity, so this leads to $\Phi(U)$ being infinite. This contradiction proves the desired result.
\end{proof}

\subsection{Step 4: Functional Analysis}

To complete the proof of Theorem~\ref{thm:domains}, we also need to modify the functional analysis step in our earlier proof of Theorem~\ref{thm:main}. One difficulty is that for an unbounded random variable $X$, $K_a(X)$ takes the value $\infty$ as $a \to \infty$. Thus we can no longer think of $K_X(a) = K_a(X)$ as a real-valued continuous function on $\overline{\R}$. 

We remedy this as follows. Note first that if $\Phi$ is a monotone additive statistic defined on $L_M$, then it is also monotone and additive when restricted to the smaller domain of bounded random variables. Thus Theorem~\ref{thm:main} gives a probability measure $\mu$ on $\R \cup \{\pm \infty\}$ such that 
\[
\Phi(X) = \int_{\overline{\R}} K_a(X) \,\dd \mu(a)
\]
for all $X \in L^{\infty}$. In what follows, $\mu$ is fixed. We just need to show that this representation also holds for $X \in L_M$. 

As a first step, we show $\mu$ does not put any mass on $\pm \infty$. Indeed, if $\mu(\{\infty\}) = \eps > 0$, then for any bounded random variable $X \geq 0$, the above integral gives $\Phi(X) \geq \eps \cdot \max[X]$. Take any $Y \in L_M$ such that $Y \geq 0$ and $Y$ is unbounded from above. Then monotonicity of $\Phi$ gives $\Phi(Y) \geq \Phi(\min\{Y, n\}) \geq \eps \cdot n$ for each $n$. This contradicts $\Phi(Y)$ being finite. Similarly we can rule out any mass at $-\infty$. 

The next lemma gives a way to extend the representation to certain unbounded random variables. 

\begin{lemma}
Suppose $Z \in L_M$ is bounded from below by $1$ and unbounded from above, while $Y \in L_M$ is bounded from below and satisfies $\lim_{a \to \infty} \frac{K_a(Y)}{K_a(Z)} = 0$, then 
\[
\Phi(Y) = \int_{(-\infty, \infty)} K_a(Y) \,\dd \mu(a). 
\]
\end{lemma}

\begin{proof}
Given the assumptions, $K_a(Z) \geq 1$ for all $a \in \R$, with $\lim_{a \to \infty} K_a(Z) = \infty$. Let $L_M^Z$ be the collection of random variables $X \in L_M$ such that $X$ is bounded from below, and $\lim_{a \to \infty} \frac{K_a(X)}{K_a(Z)}$ exists and is finite. $L_M^Z$ includes all bounded $X$ (in which case $\lim_{a \to \infty} \frac{K_a(X)}{K_a(Z)} = 0$), as well as $Y$ and $Z$ itself. $L_M^Z$ is also closed under adding independent random variables. 

Now, for each $X \in L_M^Z$, we can define 
\[
K_{X \mid Z}(a) = \frac{K_a(X)}{K_a(Z)},
\]
which reduces to our previous definition of $K_X(a)$ when $Z$ is the constant $1$. This function $K_{X \mid Z}(a)$ extends by continuity to $a = -\infty$, where its value is $\frac{\min[X]}{\min[Z]}$, as well as to $a = \infty$ by definition of $L_M^Z$. Thus $K_{X \mid Z}(\cdot)$ is a continuous function on $\R$. 

Since $\Phi$ induces an additive statistic when restricted to $L_M^Z$, and $K_{X \mid Z} + K_{Y \mid Z} = K_{X+Y \mid Z}$, we have an additive functional $F$ defined on $\mathcal{L} = \{ K_{X \mid Z} : X \in L_M^Z \}$, given by  
\[
F(K_{X \mid Z}) = \frac{\Phi(X)}{\Phi(Z)}. 
\]
Because $Z \geq 1$ implies $\Phi(Z) \geq 1$, $F$ is well-defined, and $F(1) = 1$. By Lemma~\ref{lemma:monotone-unbounded}, $F$ is also monotone in the sense that $K_{X \mid Z}(a) \geq K_{Y \mid Z}(a)$ for each $a \in \R$ implies $F(K_{X \mid Z}) \geq F(K_{Y \mid Z})$. 

Likewise we can show $F$ is 1-Lipschitz. Note that $K_{X \mid Z}(a) \leq K_{Y \mid Z}(a) + \frac{m}{n}$ is equivalent to $K_a(X) \leq K_a(Y) + \frac{m}{n} K_a(Z)$ and equivalent to $K_a(X^{*n}) \leq K_a(Y^{*n} + Z^{*m})$, where we use the notation $X^{*n}$ to denote the sum of $n$ i.i.d.\ copies of $X$. If this holds for all $a$, then by Lemma~\ref{lemma:monotone-unbounded} we also have $\Phi(X^{*n}) \leq \Phi(Y^{*n} + Z^{*m})$, and thus $\Phi(X) \leq \Phi(Y) + \frac{m}{n} \Phi(Z)$ by additivity. An approximation argument shows that for any \emph{real} number $\eps > 0$, $K_{X \mid Z}(a) \leq K_{Y \mid Z}(a) + \eps$ for all $a$ implies $\Phi(X) \leq \Phi(Y) + \eps \Phi(Z)$. Thus the functional $F$ is 1-Lipschitz.

Given these properties, we can exactly follow the proof of Theorem~\ref{thm:main} to extend the functional $F$ to be a positive linear functional on the space of all continuous functions over $\overline{\R}$ (the majorization condition is again satisfied by constant functions, as $K_{Z \mid Z} = 1$). Therefore, by the Riesz Representation Theorem, we obtain a probability measure $\mu_Z$ on $\overline{\R}$ such that for all $X \in L_M^Z$,
\[
\frac{\Phi(X)}{\Phi(Z)} = \int_{\overline{\R}} \frac{K_a(X)}{K_a(Z)} \,\dd \mu_Z(a).
\]

In particular, for any $X$ bounded from below such that $\lim_{a \to \infty} \frac{K_a(X)}{K_a(Z)} = 0$, it holds that 
\[
\Phi(X) = \int_{[-\infty, \infty)} K_a(X) \cdot \frac{\Phi(Z)}{K_a(Z)} \,\dd \mu_Z(a),
\]
where we are able to exclude $\infty$ from the range of integration (this is useful below). 

If we define the measure $\hat{\mu}_Z$ by $\frac{\dd \hat{\mu}_Z}{\dd \mu_Z}(a) = \frac{\Phi(Z)}{K_a(Z)} \leq \Phi(Z)$, then since $K_a(X)$ is finite for $a < \infty$, we have
\[
\Phi(X) = \int_{[-\infty, \infty)} K_a(X) \,\dd \hat{\mu}_Z(a).
\]
This in particular holds for all bounded $X$, so plugging in $X = 1$ gives that $\hat{\mu}_Z$ is a probability measure. But now we have two probability measures $\mu$ and $\hat{\mu}_Z$ on $\overline{\R}$ that lead to the same integral representation for bounded random variables, so Lemma~\ref{lemma:unique} implies that $\hat{\mu}_Z$ coincides with $\mu$ and is supported on the standard real line. Plugging in $X = Y$ in the above display then yields the desired result. 
\end{proof}

The next lemma further extends the representation:

\begin{lemma}
For every $X \in L_M$ that is bounded from below, 
\[
\Phi(X) = \int_{(-\infty,\infty)} K_a(X) \,\dd \mu(a).
\]
\end{lemma}

\begin{proof}
It suffices to consider $X$ that is unbounded from above. Moreover, without loss we can assume $X \geq 0$,, since we can add any constant to $X$. Given the previous lemma, we just need to construct $Z \geq 1$ such that $\lim_{a \to \infty} \frac{K_a(X)}{K_a(Z)} = 0$. Note that $\E{\ee^{aX}}$ strictly increases in $a$ for $a \geq 0$. This means we can uniquely define a sequence $a_1 < a_2 < \cdots $ by the equation $\E{\ee^{a_nX}} = \ee^n$. This sequence diverges as $n \to \infty$. We then choose any increasing sequence $b_n$ such that $b_n > n$ and $a_n b_n > 2n^2$. 

Consider the random variable $Z$ that is equal to $b_n$ with probability $\ee^{-\frac{a_n b_n}{2}}$ for each $n$, and equal to $1$ with remaining probability. To see that $Z \in L_M$, we have
\[
\E{\ee^{aZ}} \leq \ee^{a} + \sum_{n = 1}^{\infty} \ee^{-\frac{a_n b_n}{2}} \cdot \ee^{a b_n} = \ee^{a} + \sum_{n = 1}^{\infty}\ee^{(a-\frac{a_n}{2}) \cdot b_n}.
\]
For any fixed $a$, $\frac{a_n}{2}$ is eventually greater than $a + 1$. This, together with the fact that $b_n > n$, implies the above sum converges. 

Moreover, for any $a \in [a_n, a_{n+1})$, we have
\[
\E{\ee^{aZ}} \geq \E{\ee^{a_nZ}} \geq \mathbb{P}[Z = b_n] \cdot \ee^{a_n b_n} \geq \ee^{\frac{a_n b_n}{2}} > \ee^{n^2},
\]
whereas $\E{\ee^{aX}} \leq \E{\ee^{a_{n+1}X}} \leq \ee^{n+1}$. Thus 
\[
\frac{K_a(X)}{K_a(Z)} = \frac{\log \E{\ee^{aX}}}{\log \E{\ee^{aZ}}} \leq \frac{n+1}{n^2},
\]
which converges to zero as $a$ (and thus $n$) approaches infinity. 
\end{proof}

\subsection{Step 5: Wrapping Up}

By a symmetric argument, the representation $\Phi(X) = \int_{(-\infty, \infty)} K_a(X) \,\dd \mu(a)$ also holds for all $X$ bounded from above. In the remainder of the proof, we will use an approximation argument to generalize this to all $X \in L_M$. We first show a technical lemma:

\begin{lemma}
The measure $\mu$ is supported on a compact interval of $\R$. 
\end{lemma}

\begin{proof}
Suppose not, and without loss assume the support of $\mu$ is unbounded from above. We will construct a non-negative $Y \in L_M$ such that $\Phi(Y) = \infty$ according to the integral representation. Indeed, by assumption we can find a sequence $2 < a_1 < a_2 < \cdots $ such that $a_n \to \infty$ and $\mu([a_n, \infty)) \geq \frac{1}{n}$ for all large $n$. Let $Y$ be the random variable that equals $n$ with probability $\ee^{-\frac{a_n \cdot n}{2}}$ for each $n$, and equals $0$ with remaining probability. Then similar to the above, we can show $Y \in L_M$. Moreover, $\E{\ee^{a_n Y}} \geq \ee^{\frac{a_n \cdot n}{2}}$, implying that $K_{a_n}(Y) \geq \frac{n}{2}$. Since $K_a(Y)$ is increasing in $a$, we deduce that for each $n$, 
\[
\int_{[a_n, \infty)} K_a(Y) \,\dd \mu(a) \geq K_{a_n}(Y) \cdot \mu([a_n, \infty)) \geq \frac{n}{2} \cdot \frac{1}{n} = \frac{1}{2}.
\]
The fact that this holds for $a_n \to \infty$ contradicts the assumption that $\Phi(Y) = \int_{(-\infty, \infty)} K_a(Y) \,\dd \mu(a)$ is finite. 
\end{proof}

Thus we can take $N$ sufficiently large so that $\mu$ is supported on $[-N, N]$. To finish the proof, consider any $X \in L_M$ that may be unbounded on both sides. For each positive integer $n$, let $X_n = \min\{X, n\}$ denote the truncation of $X$ at $n$. Since $X \geq_1 X_n$, we have
\[
\Phi(X) \geq \Phi(X_n) = \int_{[-N, N]} K_a(X_n) \,\dd \mu(a)
\]
Observe that for each $a \in [-N, N]$, $K_a(X_n)$ converges to $K_a(X)$ as $n \to \infty$. Moreover, the fact that $K_a(X_n)$ increases both in $n$ and in $a$ implies that for all $a$ and all $n$,
\[
\vert K_a(X_n) \vert \leq \max\{\vert K_a(X_1) \vert, \vert K_a(X) \vert\} \leq \max\{\vert K_{-N}(X_1) \vert, \vert K_{N}(X_1) \vert,  \vert K_{-N}(X) \vert, \vert K_{N}(X) \vert\}.
\]
As $K_a(X_n)$ is uniformly bounded, we can apply the Dominated Convergence Theorem to deduce
\[
\Phi(X) \geq \lim_{n \to \infty}\int_{[-N, N]} K_a(X_n) \,\dd \mu(a) = \int_{[-N, N]} K_a(X) \,\dd \mu(a).
\]
On the other hand, if we truncate the left tail and consider $X^{-n} = \max\{X, -n\}$, then a symmetric argument shows
\[
\Phi(X) \leq \lim_{n \to \infty}\int_{[-N, N]} K_a(X^{-n}) \,\dd \mu(a) = \int_{[-N, N]} K_a(X) \,\dd \mu(a).
\]
Therefore for all $X \in L_M$ it holds that 
\[
\Phi(X) = \int_{[-N, N]} K_a(X) \,\dd \mu(a).
\]
This completes the entire proof of Theorem~\ref{thm:domains}. 

\section{Omitted Proofs for Section \ref{sec:monetary}}

%%\section{Comparative Risk Aversion}

\subsection{Proof of Proposition~\ref{prop:risk-aversion}}

The result can be derived as a corollary of Proposition~\ref{prop:comparative-risk-aversion} which we prove below, but we also provide a direct proof here. We focus on the ``only if'' direction because the ``if'' direction follows immediately from the monotonicity of $K_a(X)$ in $a$. Suppose $\mu$ is not supported on $[-\infty, 0]$, we will show that the resulting monotone additive statistic $\Phi$ does not always exhibit risk aversion. Since $\mu$ has positive mass on $(0, \infty]$, we can find $\eps > 0$ such that $\mu$ assigns mass at least $\eps$ to $(\eps, \infty]$. Now consider a gamble $X$ which is equal to $0$ with probability $\frac{n-1}{n}$ and equal to $n$ with probability $\frac{1}{n}$, for some large positive integer $n$. Then $\E{X} = 1$ and $K_a(X) \geq \min[X] = 0$ for every $a \in \overline{\R}$. Moreover, for $a \geq \varepsilon$ we have 
\[
K_a(X) \geq K_{\eps}(X) = \frac{1}{\eps} \log\left(\frac{n-1}{n} + \frac{1}{n} \ee^{\eps n}\right) \geq \frac{n}{2}
\]
whenever $n$ is sufficient large. Thus 
\[
\Phi(X) = \int_{\overline{\R}} K_a(X) \,\dd\mu(a) \geq \int_{[\eps, \infty]} K_a(X) \,\dd\mu(a) \geq \frac{n}{2}\eps.
\]
We thus have $\Phi(X) > 1 = \E{X}$ for all large $n$, showing that the preference represented by $\Phi$ sometimes exhibits risk seeking. 

Symmetrically, if $\mu$ is not supported on $[0, \infty]$, then $\Phi$ must sometimes exhibit risk aversion (by considering $X$ equal to $0$ with probability $\frac{1}{n}$ and equal to $n$ with probability $\frac{n-1}{n})$. This completes the proof.

\subsection{Proof of Proposition~\ref{prop:comparative-risk-aversion}}

We first show that conditions (i) and (ii) are necessary for $\int_{\overline{\R}}K_a(X)\,\dd\mu_1(a) \leq \int_{\overline{\R}}K_a(Y)\,\dd\mu_2(a)$ to hold for every $X$. This part of the argument closely follows the proof of Lemma~\ref{lemma:unique}. Specifically, by considering the same random variables $X_{n,b}$ as defined there, we have the key equation \eqref{eq:comparison1}. Since the limit on the left-hand side is smaller for $\mu_1$ than for $\mu_2$, we conclude that for every $b > 0$, $\int_{[b,\infty]}\frac{a-b}{a}\,\dd\mu_1(a)$ on the right-hand side must be smaller than the corresponding integral for $\mu_2$. Thus condition (i) holds, and an analogous argument shows condition (ii) also holds. 

To complete the proof, it remains to show that when conditions (i) and (ii) are satisfied, 
\[
\int_{\overline{\R}}K_a(X)\,\dd\mu_1(a) \leq \int_{\overline{\R}}K_a(X)\,\dd\mu_2(a)
\]
holds for every $X$. Since $\mu_1$ and $\mu_2$ are both probability measures, we can subtract $\E{X}$ from both sides and arrive at the equivalent inequality 
\begin{equation}\label{eq:comparison2}
\int_{\overline{\R}_{\neq 0}}(K_a(X) - \E{X})\,\dd\mu_1(a) \leq \int_{\overline{\R}_{\neq 0}}(K_a(X)-\E{X})\,\dd\mu_2(a).
\end{equation}
Note that we can exclude $a = 0$ from the range of integration because $K_a(X) = \E{X}$ there. Below we show that condition (i) implies 
\begin{equation}\label{eq:comparison3}
\int_{(0,\infty]}(K_a(X) - \E{X})\,\dd\mu_1(a) \leq \int_{(0,\infty]}(K_a(X)-\E{X})\,\dd\mu_2(a).
\end{equation}
Similarly, condition (ii) gives the same inequality when the range of integration is $[-\infty, 0)$. Adding these two inequalities would yield the desired comparison in \eqref{eq:comparison2}. 

To prove \eqref{eq:comparison3}, we let $L_X(a) = a \cdot K_a(X) = \log \E{\ee^{aX}}$ be the cumulant generating function of $X$. It is well known that $L_X(a)$ is convex in $a$, with $L_X'(0) = \E{X}$ and $\lim_{a \to \infty} L_X'(a) = \max[X]$. Then the integral on the left-hand side of \eqref{eq:comparison3} can be calculated as follows:
\begin{align*} 
    \int_{(0,\infty]}(K_a(X) - \E{X})\,\dd\mu_1(a) &= \int_{(0,\infty)}(K_a(X) - \E{X})\,\dd\mu_1(a) + (\max[X]-\E{X})\cdot \mu_1(\{\infty\}) \\
    &=\int_{(0,\infty)}(L_X(a) - a\E{X})\,\dd \frac{\mu_1(a)}{a} + (\max[X]-\E{X})\cdot \mu_1(\{\infty\})
\end{align*}
Note that since the function $g(a) = L_X(a) - a\E{X}$ satisfies $g(0) = g'(0) = 0$, it can be written as 
\[
g(a) = \int_{0}^{a} g'(t) \,\dd t = \int_{0}^{a} \int_{0}^{t} g''(b) \,\dd b \,\dd t = \int_{0}^{a} g''(b) \cdot (a-b) \,\dd b. 
\]
Plugging back to the previous identity, we obtain 
\begin{align*} 
    &\quad \int_{(0,\infty]}(K_a(X) - \E{X})\,\dd\mu_1(a) \\ &=\int_{(0,\infty)}\int_{0}^{a} L_X''(b) \cdot (a-b) \,\dd b\,\dd \frac{\mu_1(a)}{a} + (\max[X]-\E{X})\cdot \mu_1(\{\infty\}) \\
    &=\int_{0}^{\infty} L_X''(b) \int_{[b,\infty)} (a-b) \,\dd \frac{\mu_1(a)}{a} \,\dd b + (L_X'(\infty)-L_X'(0))\cdot \mu_1(\{\infty\}) \\
    &=\int_{0}^{\infty} L_X''(b) \int_{[b,\infty)} \frac{a-b}{a} \,\dd \mu_1(a) \,\dd b + \int_{0}^{\infty} L_X''(b)\cdot \mu_1(\{\infty\}) \,\dd b \\
    &=\int_{0}^{\infty} L_X''(b) \int_{[b,\infty]} \frac{a-b}{a} \,\dd \mu_1(a) \,\dd b,
\end{align*}
where the last step uses $\frac{a-b}{a} = 1$ when $a = \infty > b$. 

The above identity also holds when $\mu_1$ is replaced by $\mu_2$. We then see that \eqref{eq:comparison3} follows from condition (i) and $L_X''(b) \geq 0$ for all $b$. This completes the proof. 

\subsection{Proof of Theorem~\ref{thm:R-W gambles}}

The ``if'' direction is straightforward: if $\succeq_1$ and $\succeq_2$ are both represented by a monotone additive statistic $\Phi$, then they satisfy responsiveness and continuity. In addition, combined choices are not stochastically dominated because if $X \succ_1 X'$ and $Y \succ_2 Y'$ then $\Phi(X) > \Phi(X')$ and $\Phi(Y) > \Phi(Y')$. Thus $\Phi(X+Y) > \Phi(X'+Y')$ and $X' + Y'$ cannot stochastically dominate $X + Y$.

Turning to the ``only if'' direction, we suppose $\succeq_1$ and $\succeq_2$ satisfy the axioms. We first show that these preferences are the same. Suppose for the sake of contradiction that $X \succeq_1 Y$ but $Y \succ_2 X$ for some $X, Y$. Then by continuity, there exists $\eps > 0$ such that $Y \succ_2 X + \eps$. By responsiveness, we also have $X \succeq_1 Y \succ Y - \frac{\eps}{2}$. Thus $X \succ_1 Y - \frac{\eps}{2}$, $Y \succ_2 X + \eps$, but $X + Y$ is strictly stochastically dominated by $Y - \frac{\eps}{2} + X + \eps = X + Y + \frac{\eps}{2}$, contradicting Axiom \ref{ax:R-W}. 

Henceforth we denote both $\succeq_1$ and $\succeq_2$ by $\succeq$. We next show that for any $X$ and any $\eps > 0$, $\max[X]+\eps \succ X \succ \min[X]-\eps$. To see why, suppose for contradiction that $X$ is weakly preferred to $\max[X]+\eps$ (the other case can be handled similarly). Then we obtain a contradiction to Axiom \ref{ax:R-W} by observing that $X \succ \max[X] + \frac{\eps}{2}$, $\frac{\eps}{4} \succ 0$ but $X + \frac{\eps}{4} <_{1} \max[X] + \frac{\eps}{2} + 0$. 

Given these upper and lower bounds for $X$, we can define $\Phi(X) = \sup \{c \in \R:~ c \preceq X \}$, which is well-defined and finite. By definition of the supremum and responsiveness, for any $\eps > 0$ it holds that $\Phi(X) - \eps \prec X \prec \Phi(X) + \eps$. Thus by continuity, $\Phi(X) \sim X$ is the (unique) certainty equivalent of $X$. 

It remains to show that $\Phi$ is a monotone additive statistic. For this we show that $X \sim Y$ implies $X + Z \sim Y + Z$ for any independent $Z$. Suppose for contradiction that $X + Z \succ Y + Z$. Then by continuity we can find $\eps > 0$ such that $X + Z \succ Y + Z + \eps$. By responsiveness, it also holds that $Y + \frac{\eps}{2} \succ Y \sim X$. But the sum $(X + Z) + (Y + \frac{\eps}{2})$ is stochastically dominated by $(Y + Z + \eps) + X$, contradicting Axiom \ref{ax:R-W}. 

Therefore, from $X \sim \Phi(X)$ and $Y \sim \Phi(Y)$ we can apply the preceding result twice to obtain $X + Y \sim \Phi(X) + Y \sim \Phi(X) + \Phi(Y)$ whenever $X, Y$ are independent, so that $\Phi(X + Y) = \Phi(X) + \Phi(Y)$ is additive. Finally, we show $\Phi$ is monotone. Consider any $Y \geq_1 X$, and suppose for contradiction that $X \succ Y$. Then there exists $\eps > 0$ such that $X \succ Y + \eps$. This leads to a contradiction since $X \succ Y + \eps$, $\frac{\eps}{2} \succ 0$, but $X + \frac{\eps}{2}$ is stochastically dominated by $Y + \eps + 0$.

This completes the proof that both preferences $\succeq_1$ and $\succeq_2$ are represented by the same certainty equivalent $\Phi(X)$, which is a monotone additive statistic.

\section{Monotone Additive Statistics and the Independence Axiom}
\label{appx:independence}

%In this appendix we compare our representation $\Phi(X) = \int_{\overline{\R}} K_a(X)\,\dd\mu(a)$ to CARA expected utility, which corresponds to the special case where $\mu$ is a point mass at some $a \in \R$. Note that because we are taking an average of certainty equivalents rather than an average of utility functions, the preference represented by $\Phi(X)$ is \emph{not} of expected utility form unless $\mu$ is a point mass (which returns to CARA). The following result formalizes this observation by showing that a preference in our bigger class typically violates the independence axiom. 

In this appendix we discuss the classic independence axiom and what it implies for preferences represented by monotone additive statistics. 

\begin{axiom}[Independence]\label{ax:independence}
    For all $X, Y, Z$ and all $\lambda \in (0,1)$, $X \succeq Y$ implies $X_\lambda Z \succeq Y_\lambda Z$.
\end{axiom}

\begin{proposition}\label{prop:independence}
Suppose a preference $\succeq$ is represented by a monotone additive statistic $\Phi(X) = \int_{\overline{\R}} K_a(X)\,\dd\mu(a)$. Then $\succeq$ satisfies the independence axiom if and only if $\mu$ is a point mass at some $a \in \overline{\R}$.  
\end{proposition}

\begin{proof}
The ``if'' direction is relatively straightforward. If $a = 0$ then $\Phi(X) = \E{X}$. In this case $\E{X} \geq \E{Y}$ does imply 
\[
\E{X_{\lambda}Z} = \lambda \E{X} + (1-\lambda) \E{Z} \geq \lambda \E{Y} + (1-\lambda) \E{Z} = \E{Y_{\lambda}Z}.
\]
If $a > 0$ then $\Phi(X) \geq \Phi(Y)$ implies $\E{\ee^{aX}} \geq \E{\ee^{aY}}$ and thus 
\[
\lambda\E{\ee^{aX}} + (1-\lambda)\E{\ee^{aZ}} \geq \lambda\E{\ee^{aY}} + (1-\lambda)\E{\ee^{aZ}}, 
\]
so that $\Phi(X_{\lambda}Z) \geq \Phi(Y_{\lambda}Z)$. A similar argument applies to the case of $a < 0$. Finally it is easy to see that $\max[X] \geq \max[Y]$ implies $\max[X_{\lambda}Z] \geq \max[Y_{\lambda}Z]$ and the same holds for the minimum. So the above independence axiom holds for $a = \pm \infty$ as well.\footnote{Note however that $\Phi(X) = \max[X]$ or $\min[X]$ would violate a stronger form of independence that additionally requires $X \succ Y$ to imply $X_{\lambda}Z \succ Y_{\lambda}Z$ with strict preferences. This is related to the fact that these extreme monotone additive statistics do not satisfy mixture continuity.}

We turn to the ``only if'' direction of the result. By the independence axiom, whenever $c$ is a constant we have $X \succeq c$ implies $X_{\lambda}c \succeq c$ and $c \succeq X$ implies $c \succeq X_{\lambda}c$. Therefore $X \sim c$ implies $X_{\lambda}c \sim c$, which allows us to directly apply Lemma \ref{lemma:weak betweenness} from before. It remains to show that independence rules out $\Phi(X) = \beta K_{-a\beta}(X) + (1-\beta)K_{a(1-\beta)}(X)$ for some $\beta \in (0,1)$ and $a \in (0, \infty]$. 

Suppose $\Phi$ takes the above form. If $a = \infty$ then $\Phi(X) = \beta \min[X] + (1-\beta) \max[X]$ for some $\beta \in (0,1)$. To see that it violates independence, choose $X$ supported on $0$ and $\frac{1}{1-\beta}$, and $Y = 1$ so that $\Phi(X) = \Phi(Y)$. But with $Z$ being a sufficiently large constant we see that $X_{\lambda} Z$ has the same maximum as $Y_{\lambda} Z$, but a strictly smaller minimum. Hence $\Phi(X_{\lambda}Z) < \Phi(Y_{\lambda}Z)$, contradicting independence. 

If instead $a \in (0, \infty)$, then we can do a similar construction by choosing $X$ and $Y$ such that $\Phi(X) > \Phi(Y)$ but $K_{-a\beta}(X) < K_{-a\beta}(Y)$. For example, let $Y=1$, and let $X$ be supported on $\{0,k\}$, with $\Pr{X=k}=\frac{1}{k}$. Then
\begin{align*}
    K_b(X) = \frac{1}{b}\log\E{1-\frac{1}{k}+\frac{\ee^{bk}}{k}}.
\end{align*}
For $k$ tending to infinity, $K_b(X)$ tends to zero if $b<0$, and to infinity if $b>0$. Hence, for $k$ large enough, $X$ and $Y$ will have the desired property.

Now let $Z=n$ where $n$ is a large positive integer. Then
\begin{align*}
    K_b(Y_\lambda n) &= \frac{1}{b}\log\E{\lambda\E{\ee^{bY}} +(1-\lambda)\ee^{bn}}\\
    K_b(X_\lambda n) &= \frac{1}{b}\log\E{\lambda\E{\ee^{bX}} +(1-\lambda)\ee^{bn}}
\end{align*}
and so
\begin{align*}
    K_b(Y_\lambda n) - K_b(X_\lambda n) 
    =\frac{1}{b}\log\left(\frac{\lambda \E{\ee^{bY}}+(1-\lambda)\ee^{bn}}{\lambda \E{\ee^{bX}}+(1-\lambda)\ee^{bn}}\right).
\end{align*}
It easily follows that for fixed $\lambda \in (0,1)$ and $b$, 
\begin{align*}
    \lim_{n \to \infty} K_b(Y_\lambda n) - K_b(X_\lambda n) &=0 \,\,\, \text{if} \,\, b>0; \\
    \lim_{n \to \infty} K_b(Y_\lambda n) - K_b(X_\lambda n) &= K_b(Y) - K_b(X) \,\,\, \text{if} \,\, b<0. 
\end{align*}
Thus, as $n$ tends to infinity,
\begin{align*}
    \lefteqn{\lim_n \Phi(Y_\lambda n) - \Phi(X_\lambda n)} \\
    &=  \lim_n \beta \left[K_{-a\beta}(Y_\lambda n)-K_{-a\beta}(X_\lambda n)\right]+(1-\beta)\left[K_{a(1-\beta)}(Y_\lambda n)-K_{a(1-\beta)}(X_\lambda n)\right]\\
    &=  \beta \left[K_{-a\beta}(Y_\lambda n)-K_{-a\beta}(X_\lambda n)\right] > 0.
\end{align*}
Therefore, for $n$ large enough, we have found $X$ and $Y$ such that $\Phi(X) > \Phi(Y)$ but $\Phi(X_\lambda n) < \Phi(Y_\lambda n)$. This implies $X \succ Y$ but $X_\lambda n \prec Y_\lambda n$, which contradicts the independence axiom and completes the proof of Proposition~\ref{prop:independence}.
\end{proof}

\subsection{Proof of Proposition~\ref{prop:time-lotteries-strong}}

We now prove Proposition~\ref{prop:time-lotteries-strong} as a corollary of Proposition \ref{prop:independence}. The first observation is that under time invariance, strong stochastic dynamic consistency is equivalent to the following property of the preference $\succeq$: 

\begin{axiom}[Strong Stochastic Stationarity]\label{ax:strong stationarity in time}
    For every pair of time lotteries $(x,T)$,  $(y,S)$ and every $D \in L^{\infty}_{+}$ not necessarily independent, if $(x,T_d) \succeq (y,S_d)$ for almost every realization $d$ of $D$, then $(x,T+D) \succeq (y,S+D)$.
\end{axiom}

Indeed, suppose strong stochastic dynamic consistency is satisfied, and $(x,T_d) \succeq (y,S_d)$ holds for almost every realization $d$ of $D$. Then by time invariance $(x,T_d) \succeq_{t+d} (y,S_d)$ also holds for almost every $d$. Strong stochastic dynamic consistency thus implies $(x,T+D) \succeq_t (y,S+D)$ and therefore strong stochastic stationarity. A similar argument shows that conversely, strong stochastic stationarity also implies strong stochastic dynamic consistency. 

For the ``only if'' direction of Proposition~\ref{prop:time-lotteries-strong}, suppose that $\succeq$ is an MSTP that satisfies strong stochastic stationarity. Let $\succeq_{*}$ denote the preference over random times induced by $\succeq$ when fixing the payoff. That is, $T \succeq_{*} S$ if and only if $(x, T) \succeq (x, S)$ for any and every $x > 0$. 

Fix any $X \succeq_{*} Y$ and any $Z \in L^\infty_{+}$, which can be considered as random times. For a given $\lambda \in (0,1)$, choose $D$ to be a random variable that is equal to either 0 or 1, with probability $\lambda$ and $1-\lambda$, respectively. Let $\tilde X$ be a random variable that conditioned on $D=0$ has the same distribution as $X+1$, and conditioned on $D=1$ has the same distribution as $Z$. Likewise, let $\tilde Y$ be a random variable that conditioned on $D=0$ has the same distribution as $Y+1$, and conditioned on $D=1$ has the same distribution as $Z$. 

By construction $\tilde X_D \succeq_{*} \tilde Y_D$ for every possible value of $D$, so by strong stochastic stationarity $\tilde X+D \succeq_{*} \tilde Y+D$ must hold. But $\tilde X+D$ has the same distribution as $(X_\lambda Z)+1$ while $\tilde Y+D$ has the same distribution as $(Y_\lambda Z)+1$, so $(X_\lambda Z)+1 \succeq_{*} (Y_\lambda Z)+1$. Since this is an MSTP, we deduce $X_\lambda Z \succeq_{*} Y_\lambda Z$ as the independence axiom requires. 

Note that even though $\succeq_{*}$ and the associated monotone additive statistic $\Phi$ are defined only for non-negative bounded random variables, it can be extended to all of $L^{\infty}$ as shown in the proof of Proposition~\ref{prop:non-negative}. Given additivity, it is easy to see that the extension preserves independence. So we can assume $\succeq_{*}$ and $\Phi$ satisfy independence on $L^{\infty}$. This allows us to apply Proposition~\ref{prop:independence} and deduce that $\Phi$ must have a point-mass mixing measure $\mu$, which proves the ``only if'' direction of Proposition~\ref{prop:time-lotteries-strong}. 

As for the ``if'' direction, we need to verify that an MSTP represented by $V(x,T) = u(x) \cdot \ee^{-r K_a(T)}$ does satisfy strong stochastic stationarity. First consider $a = 0$, in which case the representation simplifies to $u(x) \cdot \ee^{-\E{T}}$ with the normalization $r=1$. If $(x, T_d) \succeq (y, S_d)$ for almost every $d$, then $u(x) \cdot \ee^{-\E{T_d}} \geq u(y) \cdot \ee^{-\E{S_d}}$, which can be rewritten as $\E{S_d}-\E{T_d} \geq \log \left(u(y)/u(x)\right)$. Averaging across different realizations $d$, this implies $\E{S} - \E{T} \geq \log \left(u(y)/u(x)\right)$, and thus $\E{S+D} - \E{T+D} \geq \log \left(u(y)/u(x)\right)$. 
After rearranging, this yields $u(x) \cdot \ee^{-\E{T+D}} \geq u(y) \cdot \ee^{-\E{S+D}}$. So $(x, T+D) \succeq (y, S+D)$ as demanded by strong stochastic stationarity. 

Next consider $a > 0$. In this case we normalize $r = a$ and adjust $u$ accordingly, to arrive at an equivalent representation $V(x,T) = u(x) / \E{\ee^{aT}}$. From $(x, T_d) \succeq (y, S_d)$ we obtain $u(x) \cdot \E{\ee^{aS_d}} \geq u(y) \cdot \E{\ee^{aT_d}}$ and thus
\[
u(x) \cdot \E{\ee^{a(S_d+d)}} \geq u(y) \cdot \E{\ee^{a(T_d+d)}}.
\]
Averaging across different realizations $d$ then yields $u(x) \cdot \E{\ee^{a(S+D)}} \geq u(y) \cdot \E{\ee^{a(T+D)}}$, which after rearranging gives the desired conclusion $V(x, T+D) \geq V(y, S+D)$. 

If instead $a<0$, then we normalize $r=-a$ and recover the usual EDU representation $V(x,T) = u(x) \cdot \E{\ee^{aT}}$. Essentially the same argument as above applies to this case. 

Finally consider $a = \infty$, so that $V(x,T) = u(x) \cdot \ee^{-\max[T]}$ after normalizing $r = 1$. In this case $(x, T_d) \succeq (y, S_d)$ implies $\max[S_d]-\max[T_d] \geq \log \left(u(y)/u(x)\right)$, and thus 
\[
\max[S_d+d] - \max[T_d+d] \geq \log\left(u(y)/u(x)\right).
\]
Let $\alpha = \max[S+D]$ and $c = \log\left(u(y)/u(x)\right)$ be constants. Then the above implies that for almost every realization $d$ of $D$, $T_d + d \leq \alpha-c$. Thus $T+D \leq \alpha-c$ almost surely, which gives $\max[S+D] - \max[T+D] \geq c$. This implies $V(x,T+D) \geq V(y,S+D)$ as desired. 

A similar argument applies to the case of $a = -\infty$, completing the proof. 

%It remains to rule out the extreme cases of $\Phi(X) = \max[X]$ or $\Phi(X) = \min[X]$. Suppose $\Phi(X) = \max[X]$ defines the preference $\succeq_{*}$. Let $D$ be a random variable that is uniformly distributed on $[0,1]$. Then define random variables $X, Y$ such that $X = Y = 0$ when conditioned on $D \neq 1$ and, while $X = 0, Y = 1$ when conditioned on $D = 1$. We see that for \emph{almost every} value of $D$, $X_D$ and $Y_D$ are identically distributed and therefore we should have $X+D \succeq_{*} Y+D$ by strong stationarity. But the essential maximum of $X+D$ is 1 while the essential maximum of 

\end{document}